\documentclass[aps,pra,twocolumn,nofootinbib,superscriptaddress,10pt]{revtex4-2}

\usepackage{graphicx,color,amsmath,amsfonts,enumerate,amsthm,amssymb,mathtools,enumitem,thmtools,hyperref,subfigure,mathdots,enumitem,centernot,bm,soul,bbm,stmaryrd}
\usepackage[capitalise, noabbrev]{cleveref}
\usepackage{multirow}
\usepackage[figuresright]{rotating}

\usepackage{tikz,pgfplots}
\pgfplotsset{compat=1.17}
\usetikzlibrary{quantikz}
\hypersetup{colorlinks=true,linkcolor=blue,citecolor=blue,urlcolor=blue}

\usepackage{algorithm,algpseudocode}
\usepackage{algorithmicx}
\usepackage{mathrsfs}
\usepackage{soul}

\usepackage{mathtools,mleftright}
\mleftright
\let\originalleft\left
\let\originalright\right
\renewcommand{\left}{\mathopen{}\mathclose\bgroup\originalleft}
\renewcommand{\right}{\aftergroup\egroup\originalright}

\def\E{ {\mathbb E} } 

\def\>{\rangle}
\def\<{\langle}
\newcommand{\abs}[1]{\left| {#1} \right|} 
\newcommand{\ketbra}[2]{\ensuremath{\left|#1\right\rangle\!\left\langle#2\right|}}

\newcommand{\tr}[1]{\mathrm{Tr}\left( #1 \right)}

\newcommand{\norm}[1]{\left\|#1\right\|}

\DeclareMathOperator{\Var}{Var}

\DeclareMathOperator{\Tr}{Tr}
\DeclareMathOperator{\Cov}{Cov}

\newcommand{\mc}[1]{\mathcal{#1}}

\newcommand{\hphide}[1]{}

\newcommand{\bigo}[1]{\mathcal{O}\left(#1\right)}

\newcommand{\sym}{\Pi_{\operatorname{sym}}}
\newcommand{\symm}{\operatorname{S}}
\newcommand{\obs}{\mathrm{Obs}}

\theoremstyle{plain}
\newtheorem{thm}{Theorem}
\newtheorem{theorem}[thm]{Theorem}

\newtheorem{lemma}{Lemma}
\newtheorem{prop}{Proposition}
\newtheorem{cor}{Corollary}

\theoremstyle{definition} 
\newtheorem{defn}{Definition}

\newcommand{\bfk}{\mathbf{q}}
\newcommand{\US}{\mathrm{UShadow}}

\newcommand{\rmi}{\mathrm{i}}

\newcommand{\rme}{\operatorname{e}}

\newcommand{\caH}{\mathcal{H}}
\newcommand{\caU}{\mathcal{U}}
\newcommand{\caI}{\mathcal{I}}

\newcommand{\caD}{\mathcal{D}}
\newcommand{\caT}{\mathcal{T}}
\newcommand{\poly}{\mathrm{poly}}

\usepackage{graphicx}
\usepackage{bm}
\usepackage{amsmath}
\usepackage{lipsum}
\usepackage{balance}

\usetikzlibrary{positioning}
\usetikzlibrary{snakes}
\usetikzlibrary{decorations}
\definecolor{evergreen}{rgb}{0.27,0.62,0.20}
\newcommand{\xratio}{1.2}%

\tikzstyle{tensor_blue}  =[rectangle,draw=black,fill=blue!25,       thick,minimum size=0.6cm]
\tikzstyle{tensor_green} =[rectangle,draw=black,fill=green!20,      thick,minimum size=0.6cm]
\tikzstyle{tensor_purple}=[rectangle,draw=black,fill=blue!50!red!50,thick,minimum size=0.6cm]

\tikzstyle{Permute_2} = [rectangle, draw=black, thick, fill=red!25, minimum width=1.8cm,minimum height = 0.6cm]
\tikzstyle{Permute_3} = [rectangle, draw=black, thick, fill=red!25, minimum width=3.0cm,minimum height = 0.6cm]
\tikzstyle{Permute_4} = [rectangle, draw=black, thick, fill=red!25, minimum width=4.2cm,minimum height = 0.6cm]
\tikzstyle{Permute_5} = [rectangle, draw=black, thick, fill=red!25, minimum width=5.4cm,minimum height = 0.6cm]

\tikzstyle{UVU_6} = [rectangle, draw=black, thick, fill=blue!25, minimum width=2.1cm,minimum height = 0.6cm]

\tikzstyle{2party} = [rectangle, draw=black, thick, fill=green!20, minimum width=1.3cm,minimum height = 0.6cm]

\tikzstyle{terminal}  = [=, thick, minimum width=0.3cm, minimum height = 0.2cm]

\newcommand{\RN}[1]{%
	\textup{\uppercase\expandafter{\romannumeral#1}}%
}
\renewcommand{\thetable}{\Alph{section}\arabic{table}}

\definecolor{dullblue}{rgb}{.29,.47,.77}

\def\eqref#1{\textup{(\ref{#1})}}
\newcommand{\eref}[1]{Eq.~\textup{(\ref{#1})}}
\newcommand{\lref}[1]{Lemma~\ref{#1}}
\newcommand{\tref}[1]{Theorem~\ref{#1}}
\newcommand{\pref}[1]{Proposition~\ref{#1}}

\usepackage{orcidlink}

\begin{document}

\title{Nearly query-optimal classical shadow estimation of unitary channels}

\author{Zihao Li~\orcidlink{0000-0002-2766-7521}}
\email{zihaoli20@fudan.edu.cn}
\affiliation{State Key Laboratory of Surface Physics, Department of Physics, and Center for Field Theory and Particle Physics, Fudan University, Shanghai 200433, China}
\affiliation{Institute for Nanoelectronic Devices and Quantum Computing, Fudan University, Shanghai 200433, China}
\affiliation{Shanghai Research Center for Quantum Sciences, Shanghai 201315, China}

\author{Changhao Yi~\orcidlink{0009-0001-7944-2754}}
\affiliation{State Key Laboratory of Surface Physics, Department of Physics, and Center for Field Theory and Particle Physics, Fudan University, Shanghai 200433, China}
\affiliation{Institute for Nanoelectronic Devices and Quantum Computing, Fudan University, Shanghai 200433, China}
\affiliation{Shanghai Research Center for Quantum Sciences, Shanghai 201315, China}

\author{You Zhou~\orcidlink{0000-0003-0886-077X}}
\email{you\_zhou@fudan.edu.cn}
\affiliation{Key Laboratory for Information Science of Electromagnetic Waves (Ministry of Education), Fudan University, Shanghai 200433, China}
\affiliation{Hefei National Laboratory, Hefei 230088, China}

\author{Huangjun Zhu~\orcidlink{0000-0001-7257-0764}}
\email{zhuhuangjun@fudan.edu.cn}
\affiliation{State Key Laboratory of Surface Physics, Department of Physics, and Center for Field Theory and Particle Physics, Fudan University, Shanghai 200433, China}
\affiliation{Institute for Nanoelectronic Devices and Quantum Computing, Fudan University, Shanghai 200433, China}
\affiliation{Shanghai Research Center for Quantum Sciences, Shanghai 201315, China}
\affiliation{Hefei National Laboratory, Hefei 230088, China}

\begin{abstract}
Classical shadow estimation (CSE) is a powerful tool for learning the properties of quantum states and quantum processes. Here we consider the CSE task for quantum unitary channels. By querying an unknown unitary channel $\caU$ multiple times in quantum experiments, the goal is to learn a classical description from which one can accurately predict many different linear properties of the channel, i.e., the expectation values of arbitrary observables measured on the output of $\caU$ upon arbitrary input states. Based on collective measurements on multiple systems, we propose a query efficient protocol for this task, whose query complexity has a quadratic advantage over the previous best approach for this problem, and almost saturates the information-theoretic lower bound. To further enhance practicality, we also present a variant protocol using only single-copy measurements, 
which still offers much better query performance than previous protocols that do not use quantum memory,  
and can serve as a key subroutine for learning an arbitrary unknown Hamiltonian from dynamics.
In addition to linear properties of unitary channels, our protocol can also be applied to simultaneously predict many non-linear properties, such as out-of-time-ordered correlators.  
\end{abstract}

\date{\today}
\maketitle

\def \loopwidth {1em}
\def \colspacing {2em}
\def \rowspacing {1.5em}
\setlength\tabcolsep{20pt}
\newsavebox{\composition}
\newsavebox{\tensorproduct}
\newsavebox{\traceytrace}

\let\oldaddcontentsline\addcontentsline
\renewcommand{\addcontentsline}[3]{}
\section{Introduction}\label{sec:Introduction}

In quantum mechanics, the time dynamics of a closed $d$-dimensional system is governed by a unitary channel acting on quantum states. Learning complex quantum dynamics is a fundamental problem in physics and plays a crucial role in many applications in quantum information processing \cite{gebhart2023learning,PRXQuantum.2.010102,Kliesch2021Certification,PhysRevLett.97.170501,caro2023out}.
The traditional learning framework, known as process tomography, aims to give a full characterization of the unknown unitary channel $\caU$ with high precision in the diamond distance.
As shown by a recent work \cite{haah2023query}, at least (order) $d^2$ applications of the channel are necessary to achieve this task in any quantum experiment, even with access to quantum memory.

Actually, requiring the full classical description of the unknown unitary channel is excessive for many applications. Instead, it is usually sufficient to predict certain properties of $\caU$. This is similar to the case of quantum states, where the framework of \emph{shadow tomography} \cite{aaronson2018shadow,buadescu2021improved} and \emph{classical shadow estimation} (CSE) \cite{HKPshadow20,PhysRevLett.133.020602,Grier22,PhysRevResearch.5.023027,PRXQuantum.5.010352,PhysRevLett.131.240602,Zhou2023perform,arienzo2023closed,zhou2024hybrid,zhang2024minimal,liu2024auxiliary} has led to results on predicting useful properties of a quantum system by using much fewer samples than full tomography. 
CSE has a notable feature that the measurement protocol is independent of the specific properties to be predicted. In other words, the measurements should produce some classical data from which the properties of interest can be calculated later \cite{elben2023randomized,HKPshadow20}. 
This feature enables one to simultaneously predict multiple  different properties of the system 
without losing high efficiency.

Recently, the CSE formalism has been extended from quantum states to general 
quantum channels \cite{Kunjummen23,PhysRevResearch.6.013029}.  
Given query access to an unknown qudit quantum channel $\mathcal{C}$ that represents a physical process happening in a laboratory, these works considered the task of predicting a collection of $M$ linear properties of the form 
\begin{align}\label{eq:LinearProperties}
	\Tr\left(O_l\,\mathcal{C}(\rho_l)\right), \quad l=1,2,\dots,M,  
\end{align} 
for arbitrary quantum states $\rho_l$ and observables $O_l$. 
Inheriting the mindset of CSE, the learning protocols of Refs.~\cite{Kunjummen23,PhysRevResearch.6.013029} do not depend on the specific input states and output observables. Provided that the (squared) Frobenius norms of the observables are bounded by 
$B\geq\max_l\|O_l\|_2^2$, they require $\mathcal{O}(d^2 B \log M)$ 
applications of the channel $\mathcal{C}$ to achieve this learning task (see Sec.~\ref{sec:PriorWork} for details).  
Although many works are beginning to emerge to improve the query efficiency (i.e., reduce the number of applications) \cite{HCP23,caro2022learning,zhao2023learning,huang2024learning,HKPITB21,caro2022generalization,huang2022quantum}, they impose stringent restrictions on the channel or on the input states and output observables, for example, assuming that the channel can be generated in polynomial-time (see Sec.~\ref{sec:PriorWork}).

\begin{figure*}
\begin{center}
\includegraphics[width=17.0cm]{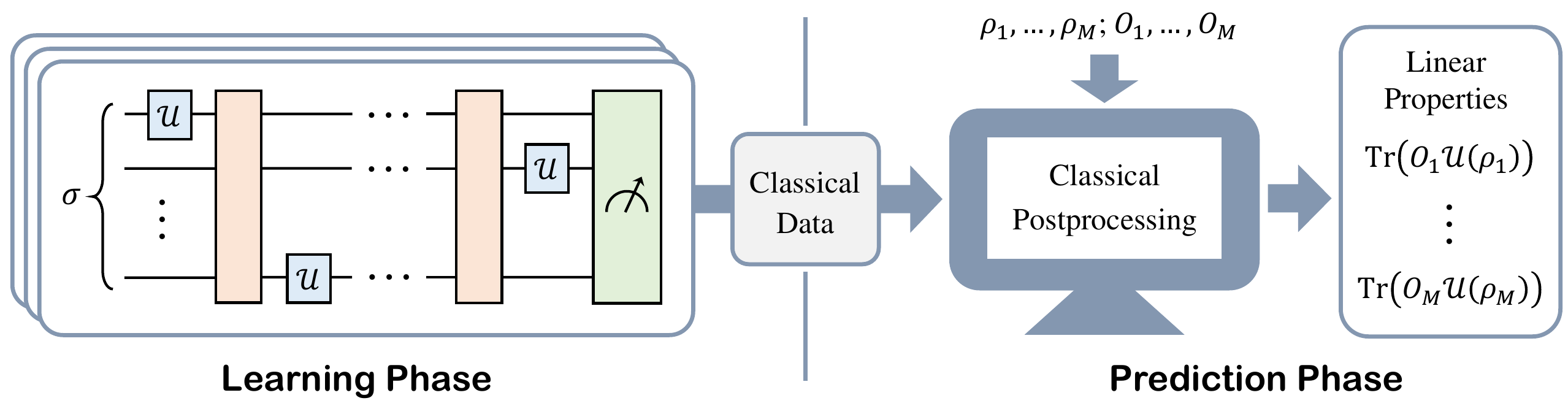}
\caption{\label{fig:CSEU}
Illustration of the CSEU task. The task contains two separate phases. 
In the learning phase, we apply the unknown unitary channel $\mathcal{U}$ multiple times in some quantum experiments  
and obtain some classical data. 
In the prediction phase, we are given a collection of quantum states $\rho_1,\dots, \rho_M$ and observables $O_1,\dots, O_M$. 
By performing classical postprocessing on the data collected from the learning phase,  
the goal is to accurately predict linear properties $\Tr\left( O_l \,\mathcal{U}(\rho_l)\right)$ for all $l=1,2,\dots, M$. 
}
\end{center}
\end{figure*}

\renewcommand{\thetable}{\arabic{table}}
\begin{table*}
\caption{\label{tab:compare}
Comparison of various protocols for the CSEU task (see Sec.~\ref{sec:PriorWork} for details). 
Here, $d$ is the system dimension; $B$ is an upper bound on the squared Frobenius norm of the target observables; 
$M$ is the number of linear properties to be predicted; and $\epsilon$ is the target additive error. 
The parameter $s$ in our protocol represents the number of systems collectively measured in the learning phase. 
Each protocol is classified based on
 whether the protocol requires quantum memory, and whether its query complexity achieves the optimal scaling behavior $\tilde{\mathcal{O}}(d)$  with respect to $d$. 
The performance of protocols in Refs.~\cite{HCP23,caro2022learning,zhao2023learning,huang2024learning,HKPITB21,caro2022generalization,huang2022quantum}
are not shown here because they have additional restrictions or assumptions on the input states, output observables, 
or the unitary channel to be learned. 
}
\begin{math} 
\begin{array}{c|ccc}
\hline\hline
\mbox{Protocol}
& \text{Quantum memory?}
& \text{Query complexity}
& \text{Is scaling optimal in $d$?}
\\[0.5ex]
\hline
\text{Ancilla-assisted CSEC \cite{Kunjummen23,PhysRevResearch.6.013029} (see Sec.~\ref{sec:CSEC})}
&\text{No}  &\bigo{ d^2 B \epsilon^{-2}\log M}      
&\text{No}  \\[0.5ex]
\hline
\text{Ancilla-free CSEC \cite{Kunjummen23,PhysRevResearch.6.013029} (see Sec.~\ref{sec:CSEC})}
&\text{No}  &\bigo{ d^3 B \epsilon^{-2}\log M}   
&\text{No}  \\[0.5ex]
\hline
\text{QPTU protocol in Ref.~\cite{surawy2022projected} (see Sec.~\ref{sec:ProcessTomg})}
&\text{No}  &\bigo{ d^4 \epsilon^{-2}}   
&\text{No}  \\[0.5ex]
\hline
\\[-2.8ex]
\text{Our protocol with $s=1$ (this work)} 
&\text{No}  &\bigo{ \big( d\epsilon^{-2} \!+ d^2\sqrt{B}\epsilon^{-1} \big) \log M}   
&\text{No}  
\\[1.1ex]
\hline
\text{Protocol based on Ref.~\cite{Grier22} (see Sec.~\ref{sec:Grier22})}
&\text{Yes}  &\bigo{ d^2 \epsilon^{-2}\log M}   
&\text{No}   \\[0.5ex]
\hline
\text{QPTU protocol in Ref.~\cite{haah2023query} (see Sec.~\ref{sec:ProcessTomg})}
&\text{Yes}  &\bigo{ d^2 \epsilon^{-1} }   
&\text{No}   \\[0.5ex]
\hline
\\[-2.8ex]
\text{Our protocol with $s=\Theta(d)$ (this work)}
&\text{Yes}  &\bigo{ d \,\big(\epsilon^{-2} \!+\! \sqrt{B}\epsilon^{-1} \big) \log M } 
&\text{Yes}  
\\[1.1ex]
\hline\hline
\end{array}	
\end{math}
\end{table*}

In this work, we reexamine the above significant learning problem under an alternative but natural and common assumption---the channel to be learned is unitary. 
We denote the learning task in this scenario as \emph{CSE of unitary channels} (CSEU), which is of fundamental and practical interest in quantum physics and quantum information processing. For instance, the quantum gates and circuits used in quantum computing are in general unitary. As shown in Fig.~\ref{fig:CSEU}, given access to an unknown unitary channel $\caU$, 
CSEU aims to accurately predict properties of the form specified in \eref{eq:LinearProperties}, with channel $\mathcal{C}$ replaced by $\caU$. Such properties play crucial roles in various research areas, including quantum machine learning \cite{biamonte2017quantum,PhysRevLett.122.040504,caro2023out,huang2022quantum} and variational quantum algorithms \cite{cerezo2021variational,kokail2019self,PhysRevResearch.6.013241}.

Based on collective measurements on multiple systems, we propose a protocol for CSEU that uses $\mathcal O(d\sqrt{B}\log M)$ applications of $\caU$. 
Compared with existing methods 
like \emph{quantum process tomography of unitary channels} (QPTU) \cite{haah2023query} or \emph{CSE of general channels} (CSEC)  \cite{Kunjummen23,PhysRevResearch.6.013029}, 
our protocol achieves a square-root reduction in 
the scaling behavior with the dimension $d$ (see Table~\ref{tab:compare}). 
In addition, by embedding the problem of learning an unknown Boolean function into the task of CSEU, 
we establish a complementary lower bound, indicating that 
our protocol has the asymptotically optimal query complexity with respect to the dimension (up to a $\log d$ factor). 
The lower bound is built in a worst-case sense, namely, the learner is required to make an accurate 
prediction for any input state and output observable. 
We further show that it can be relaxed when considering the average-case prediction over 
random input states, thereby 
highlighting the separation in query complexity 
between the worst- and average-case considerations for the CSEU task.

To address the limitation that collective measurements may be challenging to implement on current devices,  
we further propose a variant protocol for CSEU, which requires neither quantum memory nor ancillary systems. 
This variant is appealing to practical applications because it requires only the preparation of a computational-basis state, the application of a random Clifford unitary to a single system, and a measurement in the computational basis.
Compared with existing quantum memory-free learning protocols with bounded-ancilla, 
our protocol can  significantly enhance the efficiency (see Table~\ref{tab:compare}). 
As a key application, we combine this protocol with the polynomial interpolation technique \cite{stilck2024efficient,gu2024practical,caro2022learning} to learn an arbitrary Hamiltonian $H$ without any prior information on the structure of $H$.
Remarkably,  even without using quantum memory, our protocol can achieve
much better performance than previous approaches to this problem, 
in terms of both the number of queries to $H$ and total evolution time under $H$.

The rest of this paper is organized as follows. 
In Sec.~\ref{sec:preli}, we introduce the basic notation and necessary preliminaries for this work. 
In Sec.~\ref{sec:OurProtocol}, we present our protocol for CSEU and analyze its query complexity. 
In Sec.~\ref{sec:shadowLB}, we establish an information-theoretic lower bound for the query number required to complete the CSEU task.  
In Sec.~\ref{sec:AverageCase}, we study the performance of our protocol under the average-case scenario over random input states. 
In Sec.~\ref{sec:PriorWork}, we compare our CSEU protocol with a number of previous protocols. 
In Sec.~\ref{sec:Hamiltonian}, we apply our CSEU protocol to learn an arbitrary Hamiltonian.
In Sec.~\ref{sec:mainOTOC}, we generalize our protocol to predict non-linear properties of unitary channels and discuss an application in estimating out-of-time-ordered correlators (OTOCs). 
In Sec.~\ref{sec:Conclusion}  we summarize our results and discuss potential future directions.
To streamline the presentation, most technical proofs are relegated to the Appendices.

\section{Preliminaries}\label{sec:preli}
In this section, we define the key notation that will be frequently used throughout the paper, 
and review different models of learning quantum channels.

\subsection{Notation and definitions} 
Consider an $n$-qubit Hilbert space $\mathcal{H}$ with dimension $d=2^n$. 
We denote by 
$\mathcal{D}(\caH)$ the set of all density operators on $\mathcal{H}$, 
$\mathcal{L}(\mathcal{H})$ the set of linear operators on $\mathcal{H}$, and
$\mathcal{L}_{\rm H}(\mathcal{H})$ the set of Hermitian operators on $\mathcal{H}$. 
Let $I$ be the identity operator on $\caH$, 
and $|\Phi\>:=\sum_{i=1}^d|ii\>\in\mathcal{H}^{\otimes2}$ be an unnormalized maximally entangled state, 
where $\{|i\>\}_i$ represents the computational basis on $\caH$.
We denote by $U\in\mathcal{L}(\mathcal{H})$ an unknown unitary operator that we aim to learn, 
and by $\caU(\cdot)=U(\cdot)U^\dag$ the corresponding unitary channel on $\mathcal{L}(\mathcal{H})$.
For any quantum channel $\mathcal{C}:\mathcal{L}(\mathcal{H})\rightarrow\mathcal{L}(\mathcal{H})$, let 
\begin{align}
	\Upsilon_\mathcal{C}:=(\mathcal{C}\otimes \caI)\left( |\Phi\>\<\Phi|\right) 
\end{align} 
be the corresponding Choi operator, where $\mathcal{I}$ is the identity channel.  
For any $A\in\mathcal{L}(\mathcal{H})$, we denote by $\|A\|$ its operator norm and  
$\|A\|_2$ its Frobenius norm.
We write $\psi:=|\psi\>\<\psi|$ for any pure state $|\psi\>\in\mathcal{H}$. 
For $B\in[1,d\,]$, we define the set of observables  
\begin{align}\label{eq:obsB}
\obs(B) := \left\{O \in \mathcal{L}_{\rm H}(\mathcal{H}) \,\big| \Tr(O^2)\leq B, \|O\|\leq 1 \right\}.  
\end{align}

For positive integer $t$, we denote by $\sym^{(t)}$ the projector onto the symmetric subspace of $\caH^{\otimes t}$, 
and $\kappa_t=\binom{t+d-1}{t}$ the 
dimension of the symmetric subspace. 
An ensemble $\mathcal E$ of pure states in $\caH$ is said to form a state $t$-design if and only if
its $t$th moment agrees with that of the Haar distribution, i.e., 
\begin{align}\label{eq:statetdesign}
\underset{|\psi\rangle \sim \mathcal E}{\E} \left[ (|\psi\rangle\langle\psi|)^{\otimes t} \right] 
=
\underset{|\psi\rangle \sim \mu_{\rm H}}{\E} \left[ (|\psi\rangle\langle\psi|)^{\otimes t} \right]
\equiv \frac{\Pi_{\mathrm{sym}}^{(t)}}{\kappa_t },   
\end{align}
where the second expectation is over the Haar measure $\mu_{\rm H}$ on quantum states. 
For example, for any $|\phi\>\in\caH$, the ensemble $\{V|\phi\>\}_{V\in{\rm Cl}(n)}$ is a state 3-design if 
the unitary operator $V$ is sampled uniformly from the $n$-qubit Clifford group ${\rm Cl}(n)$ \cite{PhysRevA.96.062336,Webb16}. 
A more detailed introduction to state designs is given in Appendix~\ref{sec:Design}.

We will frequently use the big-$\mathcal{O}$, big-$\Omega$, and big-$\Theta$ notation for function asymptotics. 
In particular, for two functions $f, g: [0, \infty) \rightarrow[0, \infty)$, 
$f=\mathcal{O}(g)$ if there exists a $c>0$ such that $f(x) \leq c g(x)$ as $x \rightarrow \infty$,  
$f=\Omega(g)$ if $g=\mathcal{O}(f)$, 
and $f=\Theta(g)$ if both $f=\mathcal{O}(g)$ and $f=\Omega(g)$. 
We will also use $\tilde{\mathcal{O}}$ and $\tilde\Omega$ to hide factors that are poly-logarithmic in the leading-order term.

\subsection{Learning quantum channels or Hamiltonians with and without quantum memory} 

Following the classifications in Refs.~\cite{chen2022exponential,huang2022quantum,caro2022learning}, 
here we describe several different models of learning an unknown quantum channel $\mathcal{C}$ or Hamiltonian $H$, depending on 
the type of access to the quantum system and the amount of additional resources utilized.

The first type of learning protocols work  \emph{without quantum memory}
\footnote{Some references alternatively call Hamiltonian learning protocols without (with) quantum memory as 
incoherent (coherent) protocols \cite{bluhm2024hamiltonian}, or protocols without (with) quantum control \cite{Dutkiewicz2024advantageofquantum}.}. 
That is, the learner can use the channel $\mathcal{C}$ (or the time evolution under the Hamiltonian $H$) only once in each step.
To be more explicit, when learning $\mathcal{C}$, in the $k$th round of experiments, the learner can first prepare an arbitrary quantum state $\sigma_k$, then evolve it under $\mathcal{C}$ to yield the output state $(\mathcal{I}\otimes \mathcal{C})(\sigma_k)$, 
and finally perform an arbitrary quantum measurement on the output state, which is described by a \emph{positive operator-valued measure} (POVM); here, $\mathcal{I}$ is the identity channel acting on the ancillary system. 
When learning the Hamiltonian $H$, the channel $\mathcal{C}$ used in the above procedure is replaced by  $\caU_{t_k}(\cdot)=\rme^{-\rmi Ht_k}(\cdot) \rme^{\rmi Ht_k}$ for describing the time evolution under $H$, with $t_k>0$ being an arbitrary chosen time.
We call a protocol \emph{ancilla-assisted} if it uses an ancillary system, and \emph{ancilla-free} otherwise. 
If an ancillary system is allowed but its dimension is bounded by a constant, then the protocol is \emph{ancilla-bounded}.

The second type of protocols work \emph{with quantum memory},  
that is, they have the capability to store quantum data during experiments.
When learning an unknown channel $\mathcal{C}$, in each round of experiments, the learner is allowed to access $\mathcal{C}$ multiple times before measuring, 
possibly interleaving them with the application of other quantum channels.
The resulting output state before the measurement has the form 
\begin{align}\label{eq:outputC}
(\mathcal{C} \otimes \mathcal{I}) \circ \mathcal{C}_{L-1} \circ \cdots \circ
\mathcal{C}_2 \circ (\mathcal{C} \otimes \mathcal{I}) \circ
\mathcal{C}_1 \circ (\mathcal{C} \otimes \mathcal{I}) \left(\sigma\right),
\end{align}
where $\sigma$ is an arbitrary initial state, and  $\mathcal{C}_1,\mathcal{C}_2,\dots,\mathcal{C}_{L-1}$ are arbitrary quantum channels.  
Finally, the learner can perform an arbitrary POVM on the output state.  
When learning an Hamiltonian $H$, the channels $\mathcal{C}$ in the above procedure are replaced with time-evolutions under $H$, and the output state in \eref{eq:outputC} is replaced by
\begin{align}\label{eq:outputH}
(\caU_{t_L} \!\otimes \mathcal{I}) \circ \mathcal{C}_{L-1} \circ \cdots 
\circ (\caU_{t_2} \!\otimes \mathcal{I}) \circ
\mathcal{C}_1 \circ (\caU_{t_1} \!\otimes \mathcal{I}) \left(\sigma\right),\!
\end{align}
where $t_1,t_2,\dots,t_L>0$ are arbitrary chosen times.

\section{Protocol for CSEU}\label{sec:OurProtocol}
In this section, we present our protocol for CSEU and analyze its query complexity. 
The protocol is separated into the learning phase, in which the unknown unitary channel $\caU$ 
is applied several times in quantum experiments, 
and the prediction phase, in which we predict desired linear properties of $\caU$ by using 
the classical data obtained from the learning phase.

\subsection{Learning phase}\label{sec:LearnPhase}
In the learning phase, we use a measurement 
applied on quantum states in $\mathcal{D}(\caH^{\otimes s})$. 
This measurement, denoted by $\mathcal{M}_{s}$, is called the \emph{symmetric collective measurement}. 
Suppose $\{|\phi_j\>\in\caH\}_{j=1}^K$ forms a state $(s+2)$-design. 
Then $\mathcal{M}_{s}$ can be constructed as the following POVM:  
\begin{align}\label{eq:measurement}
\mathcal{M}_{s}= 
\left\lbrace A_{\phi_1},A_{\phi_2},\dots, A_{\phi_K}, \openone_s-\sym^{(s)} \right\rbrace ,
\end{align} 
where $A_{\phi_j}:= (\kappa_s/K) \, \phi_j^{\otimes s}$, 
and $\openone_s$ is the identity operator on $\caH^{\otimes s}$.
Thanks to \eref{eq:statetdesign}, the sum of the operators $A_{\phi_j}$ equals the projector $\sym^{(s)}$  onto the symmetric subspace. 
This measurement has been studied before and is known to be optimal for many learning tasks  \cite{PhysRevLett.74.1259,hayashi1998asymptotic,bruss1999optimal,zhu2018universally,Grier22,harrow2013church,anshu2022distributed}. 
In general, the implementation of $\mathcal{M}_{s}$ on large systems with current technology remains challenging.
Nevertheless, for systems containing only a few qubits, such measurements with a small $s$ have been successfully demonstrated in superconducting, trapped-ion, and photonic quantum systems \cite{hou2018deterministic,conlon2023approaching,zhou2023experimental,PhysRevLett.125.210401}.

We denote by $\mathcal{A}_{\text {learn}}$ our quantum algorithm working in the learning phase. 
To extract meaningful information about $\caU$, as shown in Fig.~\ref{fig:LearnPhase}~(a), 
$\mathcal{A}_{\text {learn}}$ runs the following learning procedure $m$ times independently:   
\begin{itemize}
\item[1.] Randomly choose a pure state $|\hat{\psi}\>\in\caH$ from a 4-design state ensemble $\{|\psi_i\>\}_{i=1}^L$; then prepare $s$ copies of the state $\hat{\psi}$.  
\item[2.] Apply the unitary channel $\caU$ on each copy of $\hat{\psi}$. 
\item[3.] Measure the rotated state $\caU(\hat{\psi})^{\otimes s}$ with the symmetric collective measurement $\mathcal{M}_s$. If the measurement outcome $A_{\hat\phi}$ is obtained ($\hat\phi\in\{\phi_1,\dots, \phi_K\}$), we then record the classical description of the corresponding pure state $\hat{\phi}$.
\end{itemize}
It is worth noting that we will never see the measurement outcome $(\openone_s-\sym^{(s)})$ in Step~3 of the above procedure, because $\{A_{\phi_1},\dots, A_{\phi_K}\}$ forms a POVM on the symmetric subspace of $\caH^{\otimes s}$, 
and the state $\caU(\hat{\psi})^{\otimes s}$ to be measured is supported in this subspace.

\begin{figure}
\begin{center}
\includegraphics[width=8.5cm]{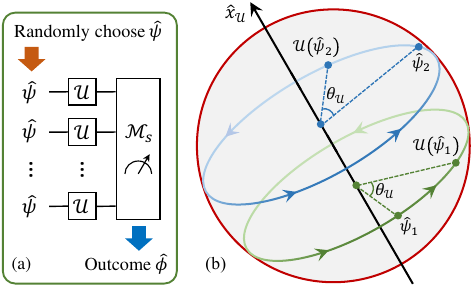}
\caption{\label{fig:LearnPhase}
Learning phase of our protocol for CSEU: procedure and intuition. 
(a) Schematic view of our learning algorithm $\mathcal{A}_{\text {learn}}$. 
In each round, we first prepare $s$ copies of a pure state $\hat{\psi}$ chosen from a 4-design ensemble, 
and rotate them with the unitary channel $\caU$, 
then measure the resulting states using the symmetric collective measurement $\mathcal{M}_s$. 
(b) Geometric intuition of our protocol illustrated on the single-qubit Bloch sphere. 
}
\end{center}
\end{figure}

Before providing a detailed explanation of $\mathcal{A}_{\text {learn}}$, we offer some additional remarks. 
First, if we choose parameter $s\geq 2$ in the learning procedure, 
then our protocol relies on collective measurements and requires quantum memory.
In contrast, when choosing $s=1$, our protocol relies solely on single-copy measurements and does not need quantum memory. 
In this special case, we use at most the 3-design property of the measurement $\mathcal{M}_{s=1}$.  
Since any orbit of the multiqubit Clifford group forms a state 3-design \cite{PhysRevA.96.062336,Webb16}, 
the \emph{random global Clifford measurement} (RGCM) used in the original CSE protocol \cite{HKPshadow20} can be applied in place of 
$\mathcal{M}_{s=1}$ in our learning procedure, and the analysis will be identical. 
RGCM only requires implementation of Clifford unitaries and measurements in the computational basis, 
and is thus appealing to practical applications.  
Second, recall that the initial state $\hat{\psi}$ is chosen from a 4-design state ensemble in the learning procedure. 
In fact, this 4-design property is utilized only in the case of $s\geq2$ throughout the paper. 
When $s=1$, the 3-design property is sufficient for proving our results. 
Therefore, in the $s=1$ case, we can prepare the initial state $\hat{\psi}$ by simply applying a random Clifford unitary 
to the computational-basis state $|0^n\>\<0^n|$.

After each round of the learning experiments,  
the description of the prepared state $\hat{\psi}$ and the outcome $\hat{\phi}$ are stored in a classical memory. 
It is instructive to look at the expectation of their tensor product, which is clarified in the following lemma; see Appendix~\ref{sec:Proof12Moment} for a proof. 

\begin{lemma}
\label{lem:first_moment}
The expectation of $\hat{\phi} \otimes \hat{\psi}$ (over both random state preparation and measurement outcomes) is
\begin{align}
	\E[\hat{\phi} \otimes \hat{\psi}] &= \frac{(d+1+s)(I\otimes I)+s (U\otimes U^{\dag})T_{(1,2)}}{d(d+1)(d+s)}, 
\end{align}
where $T_{(1,2)}$ is the swap operator. 
\end{lemma}

Therefore, after obtaining the description of $\hat{\psi}$ and $\hat{\phi}$ in a single run, 
we can apply a transformation to $\hat{\psi}\otimes\hat{\phi}$ to generate a classical snapshot $\hat X$, whose expectation contains all information of the unknown unitary channel $\caU$. Concretely, let 
\begin{align}\label{eq:defhatX}
\!\hat X := \frac{d(d+1)(d+s) (\hat{\phi}\otimes\hat{\psi}^{\top}) - (d+1+s)(I\otimes I)}{s}, \!
\end{align}
where $(\cdot)^{\top}$ denotes the transpose operation (with respect to the computational basis).
According to \lref{lem:first_moment}, $\hat X$ is an unbiased estimator 
for the Choi operator $\Upsilon_\caU$ of $\caU$: 
\begin{align}\label{eq:EhatX}
\E \big[ \hat X \big] = \big[ (U\otimes U^{\dag})T_{(1,2)}\big]^{\top_2}
	=\Upsilon_\caU, 
\end{align} 
where $\E$ denotes the average over the randomness of state preparation and 
quantum measurement in the learning phase, $(\cdot)^{\top_2}$ is the transpose operation on the second subsystem,      
and the second equality can easily be derived via the tensor-network diagrams introduced in Appendix~\ref{sec:Tensor}.

From the quantum experiments we obtain a collection of $m$ independent classical snapshots of $\caU$:
\begin{align}
	\US(\caU,m)=\big\lbrace \hat X_1,\dots,\hat X_m \big\rbrace.  
\end{align}
This collection, as the output of the learning phase, is called the \emph{classical shadow data}. 
As we shall see in the following sections, 
many properties of $\caU$ can be extracted efficiently  from $\US(\caU,m)$.

The reason why $\mathcal{A}_{\text {learn}}$ learns all information about $\caU$ efficiently can be intuitively explained by a toy model illustrated in Fig.~\ref{fig:LearnPhase}~(b), where an unknown single-qubit unitary channel $\caU$ can rotate the Bloch sphere around an axis $\hat{x}_\caU$ by an angle $\theta_\caU$. 
To recover the information about $\caU$, one needs to determine both the value of $\theta_\caU$ and the direction of $\hat{x}_\caU$. 
Mathematically, this goal can be achieved by randomly choosing a few points 
on the sphere (two points are enough in the usual case), and observe the positions to which they are rotated.
Physically, these operations correspond to first randomly selecting a few pure states $\hat{\psi}_1,\dots,\hat{\psi}_m$ on $\caH$ (Step 1 of  $\mathcal{A}_{\text {learn}}$), and then learning information about the rotated states $\caU(\hat{\psi}_i)$. 
We expect that the latter goal can be efficiently achieved by performing $\mathcal{M}_s$ on 
multiple copies of $\caU(\hat{\psi}_i)$ (Step 3 of $\mathcal{A}_{\text {learn}}$), 
because this measurement is optimal for pure states for many estimation problems \cite{hayashi1998asymptotic,bruss1999optimal,zhu2018universally,Grier22,harrow2013church,anshu2022distributed}.

\subsection{Prediction phase}\label{sec:PredPhase}
In the prediction phase, we are given classical descriptions of 
$M$ quantum states $\rho_1,\dots, \rho_M\in\caD(\caH)$ 
and $M$ observables $O_1,\dots, O_M\in\obs(B)$, where $\obs(B)$ is defined in \eref{eq:obsB}.
The goal is to accurately predict linear properties $\Tr\left( O_l\,\mathcal{U}(\rho_l)\right)$ for all $l=1,2,\dots,M$ using the 
classical shadow data $\US(\caU,m)=\lbrace\hat X_1,\dots,\hat X_m\rbrace$. 

Recall that each snapshot $\hat X_i$ exactly reproduces the Choi operator $\Upsilon_\caU$ of 
the unitary channel $\caU$ in expectation. As a consequence, we have 
\begin{align}
\E \left[ \Tr\left( (O\otimes\rho^{\top}) \hat{X}\right) \right] 
&= \Tr\left[ (O\otimes\rho^{\top})\, \Upsilon_\caU \right]
= \Tr\left( O\,\mathcal{U}(\rho)\right)  
\end{align} 
for any state $\rho$ and observable $O$,  
where the second equality holds by  channel-state duality \cite{choi1975completely}.
Hence, a natural way to predict $\Tr\left( O\,\mathcal{U}(\rho)\right)$ is 
to calculate $\Tr[(O\otimes\rho^{\top}) \hat{X}_i]$ for many independent 
snapshots $\hat{X}_i$ and then take their average.  
However, as shown in Appendix~\ref{sec:Trivialmean}, this method requires $\mathcal{O}(d^2 B/\epsilon^2)$ 
applications of $\caU$ to predict a single $\Tr\left( O\,\mathcal{U}(\rho)\right)$ within additive error $\epsilon$, even when collective measurements on multiple systems are employed. 
If only single-copy measurements are accessible, 
then the query complexity further increases, reaching up to $\mathcal{O}(d^3 B/\epsilon^2)$. 
Such high query costs are unsatisfactory.

\begin{figure}
\begin{algorithm}[H]
{\small
\hspace{-45pt}\textbf{Input:}  $\US(\caU,m)$, $\rho_1,\dots, \rho_M$, $O_1,\dots, O_M$ \\
\hspace{-103pt} \textbf{Output:} $\hat{E}(O_1,\rho_1), \dots, \hat{E}(O_1,\rho_M)$

\begin{algorithmic}[1]
\caption{{\small Algorithm for $\mathcal{A}_{\text {pred}}$ in the prediction phase} \ \ }
\label{alg:prediction}
\State{Split the shadow data $\US(\caU,m)$ into $R$ batches of equal size, 
where each batch contains $q=m/R$ snapshots (we may assume that $m$ is a multiple of $R$ for simplicity), and set 
\begin{align}\label{eq:definehatY(r)}
\quad
\hat{Y}_{(r)}:= \frac{1}{q(q-1)d}\sum_{i,j=(r-1)q+1 | i\ne j}^{rq} \hat{X}_i \hat{X}_j ,
\quad r=1,\dots,R. 
\end{align} }
\For{$l=1,2,\dots,M$,}
\State{Construct $R$ independent sample mean estimators: 
\begin{align}\label{eq:definetildeZ}
\quad
\hat{Z}_{(r)}(O_l,\rho_l):=\Tr\left[\left( O_l\otimes\rho_l^{\top}\right)\hat{Y}_{(r)} \right], 
\quad r=1,\dots,R . 
\end{align} }
\State{Compute the median of means: 
\begin{align}\label{eq:definehatE}
\quad
\hat{E}(O_l,\rho_l):=
\text{median}\left\lbrace \hat{Z}_{(1)}(O_l,\rho_l),\dots, \hat{Z}_{(R)}(O_l,\rho_l) \right\rbrace . \!
\end{align}}
\EndFor 
\State{\textbf{return} $\hat{E}(O_l,\rho_l)$ as the estimate for $\Tr\left( O_l\,\mathcal{U}(\rho_l)\right)$.} 
\end{algorithmic}
}
\end{algorithm}
\end{figure}

To improve the efficiency, we need to construct an estimator for the unitary channel in a different way.
Inspired by the CSE protocol for pure states in Ref.~\cite{Grier22}, we note that 
the product of two independent snapshots is still proportional to the Choi operator in expectation, 
\begin{align}\label{eq:ExpXiXj}
	\E \big[ \hat{X}_i \hat{X}_j \big]= \E \big[ \hat{X}_i \big]\E \big[ \hat{X}_j\big] =\Upsilon_\caU ^2=d\,\Upsilon_\caU,
	\quad \forall\, i\ne j. 
\end{align} 
In light of this, one can predict $\Tr\left( O\,\mathcal{U}(\rho)\right)$ by 
calculating the quadratic estimators 
\begin{align}
	\tilde\Lambda_{i,j}:=d^{-1}\Tr[(O\otimes\rho^{\top}) \hat{X}_i\hat{X}_j]
\end{align}  
for all index pairs $\{(i,j)|\,i\ne j\}$ and taking their average.
Since the number of different $\tilde\Lambda_{i,j}$ is much larger than that of $\Tr[(O\otimes\rho^{\top}) \hat{X}_i]$, 
this method is expected to enhance the estimation precision and thereby reduce the query cost compared with the former method. 
Indeed, as we shall see in Sec.~\ref{sec:performance}, it can  achieve nearly a square-root reduction in 
query complexity.

To be concrete, the procedure of our classical algorithm $\mathcal{A}_{\text {pred}}$ working in the prediction phase 
is summarized in Algorithm~\ref{alg:prediction}. 
Thanks to \eref{eq:ExpXiXj}, the operator $\hat{Y}_{(r)}$ defined in \eref{eq:definehatY(r)} equals 
the Choi operator $\Upsilon_\caU$ in expectation.   
This fact implies that $\hat{Z}_{(r)}(O_l,\rho_l)$ constructed in \eref{eq:definetildeZ} 
is an unbiased estimator for the desired linear property $\Tr\left( O_l\,\mathcal{U}(\rho_l)\right)$. 
In Step~4 of Algorithm~\ref{alg:prediction}, we use the \emph{median-of-means} method \cite{HKPshadow20} to  
suppress the failure probability of accurate prediction.

\subsection{Rigorous performance guarantees}\label{sec:performance}
Now we investigate the query complexity of our protocol for CSEU. That is, in order to achieve 
\begin{align}
\Pr\left\{ \big| \hat{E}(O_l,\rho_l) -\Tr\left( O_l\,\mathcal{U}(\rho_l)\right) \big|\!< \epsilon\ \, \forall\,  1\leq l\leq M \right\} \geq 1-\delta 
\end{align} 
for some $0<\delta,\epsilon<1$, how many applications of $\mathcal{U}$ are required by our protocol?  
To answer this question, we first analyze the statistical fluctuation of each sample mean estimator $\hat{Z}_{(r)}(O_l,\rho_l)$.  Using the definition of $\hat{Y}_{(r)}$ in \eref{eq:definehatY(r)}, 
for any quantum state $\rho$ and observable $O$ on $\caH$,
we can expand the variance  as 
\begin{align}\label{eq:VarY1}
&\Var\left[\hat{Z}_{(r)}(O,\rho)\right] 
=
\frac{1}{q^2(q-1)^2 d^2} \sum_{i \neq j} \sum_{k \neq \ell} \operatorname{Cov}\left( \Lambda_{i,j},\Lambda_{k,\ell}  \right), 
\end{align} 
where  
$q$ is the number of snapshots in each batch (see Step 1 of Algorithm~\ref{alg:prediction}), 
$\operatorname{Cov}(\cdot,\cdot)$ is the covariance between two variables, and 
\begin{align}
	\Lambda_{i,j}:=d\tilde\Lambda_{i,j}=\Tr\!\big[ \left( O\otimes \rho^{\top}\right) \hat{X}_i \hat{X}_j\big]. 
\end{align} 
The sum in \eref{eq:VarY1} consists of five different cases:  
when all indices $i,j,k,\ell$ are distinct, the covariance is 0 by independence; 
the other four cases are summarized in the following lemma, which is proved in Appendix~\ref{sec:covariances}. 
For simplicity, we denote by $\wp=\Tr(\rho^2)$ the purity of $\rho$. 
\begin{lemma} 
\label{lem:all_the_covariances}
The covariance term
$\operatorname{Cov}\left( \Lambda_{i,j},\Lambda_{k,\ell}  \right)$ 
has the following upper bound for each combination $(i,j,k,\ell)$ with $i \neq j$ and $k \neq \ell$: 
\begin{enumerate}
\item One index matches $(|\{ i, j \} \cap \{ k, \ell \}| = 1)$
\begin{itemize}
\item Match in different positions $(i = \ell$ or $j = k)$: 
\begin{align}
\quad \bigo{d^2 \min\{1, B\wp \}} .
\end{align}  

\item Match in the same position $(i = k$ or $j = \ell)$: 
\begin{align}
\quad \bigo{\frac{d^3\wp}{s} + d^2 \min\{1, B\wp \}},
\end{align} 
where $s$ is the number of systems collectively measured in the learning phase. 
\end{itemize}

\item Both indices match $(|\{ i, j \} \cap \{ k, \ell \}| = 2)$

\begin{itemize}
\item Order swapped $(i = \ell$ and $j = k)$:
\begin{align}
\qquad\quad 
\mathcal O &\bigg(
\left( \frac{d^3}{s^4}
+ \frac{d^2}{s^3} 
+ \frac{d\sqrt{d\wp}}{s^2} 
+ \frac{\sqrt{d\wp}}{s} 
+ \wp \right) d^2B\bigg).
\end{align}

\item Same order $(i = k$ and $j = \ell)$: 
\begin{align}
\qquad\quad 
\bigo{ \left( \frac{d^4}{s^4} +\frac{d^3}{s^3} +\frac{d^2}{s^2} +\frac{d}{s} + 1 \right) d^2B \wp }.
\end{align}
\end{itemize}

\end{enumerate}
\end{lemma}

For all combinations of $i,j,k,\ell$ in the $r$th batch with $i \neq j$ and $k \neq \ell$, there are $\mathcal O (q^4)$ combinations where all four indices  are distinct;
$\mathcal O (q^3)$ combinations where exactly one index matches; and $\mathcal O (q^2)$ combinations where both indices match.
These facts, together with \eref{eq:VarY1} and \lref{lem:all_the_covariances}, immediately imply the following proposition, 
which clarifies the variance of $\hat{Z}_{(r)}(O,\rho)$.   
\begin{prop}\label{prop:VarMain}
Given $O\in\obs(B)$ and $\rho\in\caD(\caH)$,  
the variance of $\hat{Z}_{(r)}(O,\rho)$ is upper bounded by
\begin{align}
\mathcal O  \bigg[\, \frac{1}{q} \left( \frac{d \wp}{s} + \min\left\lbrace 1, B\wp\right\rbrace  \right) 
  + \frac{1}{q^2 } \left( \frac{d^4}{s^4} + 1 \right) B \wp \bigg].  
\end{align} 
\end{prop}

Having obtained the variance, we can apply Chebyshev's inequality to bound the failure probability of each sample mean estimator $\hat{Z}_{(r)}(O,\rho)$: 
\begin{align}\label{eq:ChebyY}
\operatorname{Pr}\left\{ \left| \hat{Z}_{(r)}(O,\rho) - \Tr\left( O\,\mathcal{U}(\rho)\right) \right| \geq \epsilon\right\} 
\leq 
\frac{\Var\!\big[\hat{Z}_{(r)}(O,\rho)\big]}{\epsilon^2}  . 
\end{align} 
The query cost of our protocol can then be derived 
from this bound.    

\begin{figure}
\begin{center}
\includegraphics[width=7.2cm]{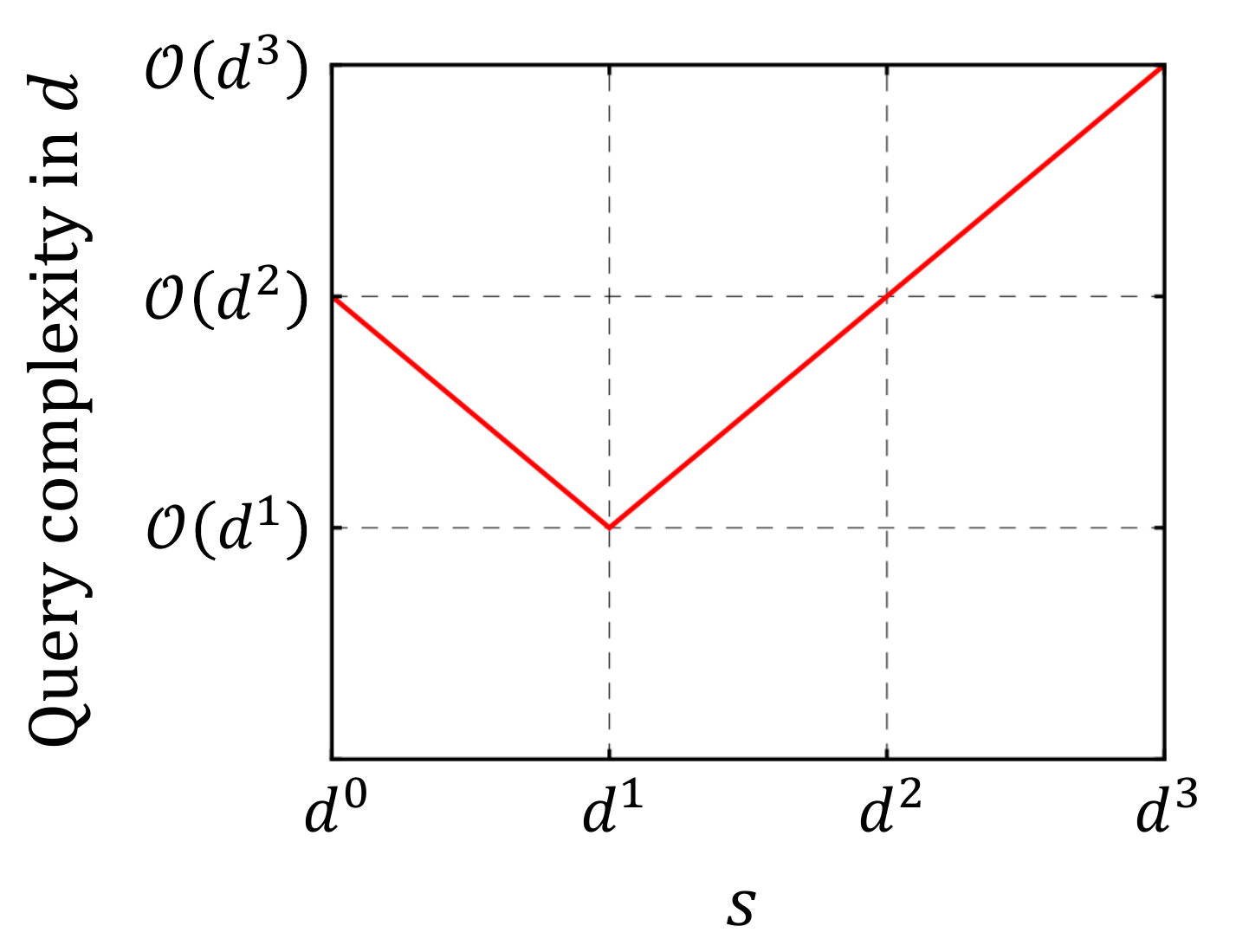}
\caption{\label{fig:QueryCompl}
The query complexity of our protocol [see \eref{eq:generalUB}] with respect to the dimension $d$. Here $s$ denotes the number of systems collectively measured in our learning phase. The optimal complexity $\bigo{d}$ is attained when $s=\Theta(d)$. 
}
\end{center}
\end{figure}

\begin{theorem}\label{thm:jointUpperBound}
Suppose $s\geq 1$ is the number of systems collectively measured in our learning phase, $\rho_1,\dots, \rho_M\in \caD(\caH)$, and  $O_1,\dots, O_M\in \obs(B)$. 
To estimate all $M$ properties $\Tr\left( O_l \,\mathcal{U}(\rho_l)\right)$
within additive error $\epsilon$ and failure probability at most $\delta$,
the number of queries required by our protocol for CSEU is upper bounded by 
\begin{align} \label{eq:generalUB}
	\bigo{ \left( \frac{\max\{d,s\}}{\epsilon^2} + \frac{\max\{d^2,s^2\}\sqrt{B}}{s\epsilon} \right) \log\left( \frac{M}{\delta}\right)}. 
\end{align} 
\end{theorem}

\tref{thm:jointUpperBound} is a simplified version of \tref{thm:UpperBoundFormal} in Appendix~\ref{sec:ProofThm1}. 
Note that the query complexity is closely related to the value of $s$; see Fig.~\ref{fig:QueryCompl} for an illustration. 
In particular, when $s=1$, i.e., single-copy measurements are used, the query number reduces to 
\begin{align} \label{eq:singleUB}
\bigo{ \left( \frac{d}{\epsilon^2} + \frac{d^2\sqrt{B}}{\epsilon} \right) \log\left(\frac{M}{\delta}\right)}.  
\end{align} 
As $s$ increases from 1, the query number decreases and reaches the following minimum value when $s=\Theta(d)$,     
\begin{align} \label{eq:jointUB}
	\bigo{ d \,\left(  \frac{1}{\epsilon^2} + \frac{\sqrt{B}}{\epsilon} \right) \log\left( \frac{M}{\delta}\right) }. 
\end{align} 
When $s$ exceeds $\Theta(d)$, however, the query number starts to increase. 
Hence, if we are allowed to perform collective measurements on multiple systems, the best choice of $s$ is $s=\Theta(d)$. Equations~\eqref{eq:singleUB} and \eqref{eq:jointUB}
show that our protocol significantly outperforms previous approaches for the CSEU task (see Sec.~\ref{sec:PriorWork} for details). This exceptional query performance constitutes a key advantage of our protocol. 

It is worth noting that while the symmetric collective measurement $\mathcal{M}_{s}$ with $s=\Theta(d)$ can achieve maximal reduction in query complexity, it demands both a state $\Theta(d)$-design and joint operations on $\Theta(d)$ systems. These requirements make its experimental implementation extremely challenging with current technology. 
By contrast, choosing $s=1$ only requires performing Clifford unitaries and measurements in the computational basis, which are significantly easier to realize.  
In practice, one may adjust the parameter $s$ to strike a balance between query cost and measurement feasibility.

\section{Lower bound for the query complexity}\label{sec:shadowLB}
Although our protocol has a much higher query efficiency than previous approaches for the CSEU problem, 
it still requires order $d$ applications of the unknown unitary channel $\caU$. 
One may further ask if a better protocol exists or
if there is a fundamental limitation on the query cost.
In this section, we show that indeed a $\tilde\Omega(d)$ query complexity lower bound exists for the CSEU task,   
which stems from information-theoretic reasons. 

The CSEU task is formally defined as follows; we refer to Fig.~\ref{fig:CSEU} for a visualization. 

\begin{defn}[CSEU task]\label{defn:UShadowTask}
The CSEU task involves 
developing a protocol $\mathcal{B}$ that contains two separate 
algorithms $\mathcal{B}_{\text {learn}}$ and $\mathcal{B}_{\text {pred}}$. 
The ``learning'' algorithm $\mathcal{B}_{\text {learn}}$ is given  access to an unknown $n$-qubit unitary channel $\mathcal{U}$; 
after applying it multiple times in some quantum experiments, $\mathcal{B}_{\text {learn}}$ should output some classical data. 
The ``prediction'' algorithm $\mathcal{B}_{\text {pred}}$ takes this classical data, 
a collection of states $\rho_1,\dots, \rho_M$ in $\caD(\caH)$ and 
observables $O_1,\dots, O_M$ in $\obs(B)$ as inputs, and should output $M$ estimates 
$E(O_l,\rho_l)$.  
The goal of the task is to reach 
\begin{align}\label{eq:goalCSEU}
\Pr\left\{ \left| E(O_l,\rho_l)-\Tr\left( O_l\,\mathcal{U}(\rho_l)\right) \right| 
          \!< \epsilon\ \, \forall\,  1\leq l\leq M \right\}\! \geq 1- \delta , 
\end{align} 
where $0<\delta,\epsilon<1$ are accuracy parameters. 
\end{defn}

\begin{theorem}\label{thm:jointLowerBound}
Suppose  $1\leq B\leq d$, $0<\delta,\epsilon\leq1/3$, and $M\geq 1$. 
Then any protocol has to query the channel $\mathcal{U}$ 
at least $\tilde\Omega\left(d\right)$ times to complete the CSEU task. 
\end{theorem}

By this theorem, an exponential query complexity (with respect to the qubit number) for the CSEU task is inevitable, even if highly complex adaptive or collective operations are used.  
In addition, the scaling of the lower bound in dimension $d$ matches (up to a $\log d$ factor) the query complexity of our protocol constructed in Sec.~\ref{sec:OurProtocol} [see \eref{eq:jointUB}]. 
In other words, our protocol for CSEU is nearly query-optimal.

\subsection{Proof of \tref{thm:jointLowerBound}}
In this subsection, we prove \tref{thm:jointLowerBound} via a reduction from a problem of learning Boolean functions to the CSEU task. 
That is, any protocol for the CSEU task can be modified to complete another task defined below, 
whose query complexity lower bound is known in the literature. 
This existing lower bound will imply that the CSEU task must require a significant amount of queries.

\begin{defn}\label{defn:LearnBoolean}
	The \emph{exact learning task of Boolean functions} involves developing a quantum protocol $\mathcal{P}$ for learning 
	an unknown Boolean function $c:\{0,1\}^n \rightarrow\{0,1\}$. 
	Here $\mathcal{P}$ is given access to an oracle $\mathrm{MQ}(c)$, which is a quantum unitary satisfying  
	$\mathrm{MQ}(c)|x, b\>=|x, b \oplus c(x)\>$ for all $x \in\{0,1\}^n$ and $b \in\{0,1\}$. 
	The goal of the task is to achieve: 
	with probability at least $2/3$, $\mathcal{P}$ outputs a classical function $h$ such that $h(x)=c(x)$ for all $x \in\{0,1\}^n$.
\end{defn}

\begin{lemma}[\cite{Kothari14,arunachalam2017survey}]\label{lem:LearnBoolean}
Any learning protocol has to query the oracle $\mathrm{MQ}(c)$ at least $\Omega\left(2^n\right)$ times 
to achieve the exact learning task of Boolean functions. 
\end{lemma}

Now we are ready to prove \tref{thm:jointLowerBound}. 
Suppose the unitary channel $\mathcal{U}$ is acting on $n+1$ qubits, and 
there is a protocol $\mathcal{B}=\left( \mathcal{B}_{\text {learn}}, \mathcal{B}_{\text {pred}}\right)$ that can use $N$ queries of $\mathcal{U}$ to successfully achieve the CSEU task, with parameters $B=M=1$ and $\delta,\epsilon=1/3$. 
In the following, we shall prove that necessarily $N\geq \Omega\left(2^n/n\right)$, which implies \tref{thm:jointLowerBound}. 
Note that here we assume that the qubit number is $n+1$ instead of $n$ for the convenience of discussion.
The proof below is divided into two parts.

\subsubsection*{Part 1: Predicting many expectation values}
In this part, we show that there exists a new protocol 
$\mathcal{B}^*=( \mathcal{B}^*_{\text{learn}},\mathcal{B}^*_{\text {pred}})$ that can use 
$N \mathcal{O} \big(\!\log(M/\delta) \big)$ queries of $\mathcal{U}$ to complete the CSEU task 
with parameters $B=1$, $\epsilon=1/3$, any $0<\delta<1$, and any integer $M\geq1$.

Based on $\mathcal{B}$, a new protocol $\mathcal{B}^*$ is constructed as follows using the median-of-means technique:  
in the learning phase, we run
the algorithm $\mathcal{B}_{\text {learn}}$ in total $k$ times independently 
to obtain $k$ pieces of classical data; 
in the prediction phase, we classically process these data using $\mathcal{B}_{\text {pred}}$ to obtain $k$ 
estimates $E_1(O_l,\rho_l),\dots,E_k(O_l,\rho_l)$ for each $1\leq l \leq M$, 
and we take the median of them as our final estimate for 
$\Tr\left( O_l\,\mathcal{U}(\rho_l)\right)$: 
\begin{align}
\tilde{E}(O_l,\rho_l):=\operatorname{median} \left\{E_1(O_l,\rho_l),\dots,E_k(O_l,\rho_l)\right\}. 
\end{align}
To estimate all $\Tr\left( O_l\,\mathcal{U}(\rho_l)\right)$ accurately with failure probability 
at most $\delta$, it suffices to choose $k=\mathcal{O} \big(\!\log(M/\delta) \big)$
by an analysis identical to that used in the proof of \tref{thm:jointUpperBound}.
Thus, the total number of queries required by $\mathcal{B}^*$ is $N\cdot k=N \mathcal{O} \big(\!\log(M/\delta) \big)$. 
In particular, when $\delta=1/3$ and $M=2^n$, the query number reads 
\begin{align}
N \mathcal{O} \big(\!\log(M/\delta) \big)= N \mathcal{O} (\log (3\times 2^n)) = N \mathcal{O}(n). 
\end{align}

\subsubsection*{Part 2: Reduction to the learning task of Boolean functions}
In this part, we show that the protocol $\mathcal{B}^*$ constructed in Part 1 can be used to complete the 
exact learning task of Boolean functions given in Definition~\ref{defn:LearnBoolean}. 
Consequently, the lower bound on the query complexity for this exact learning task can be turned into a lower bound for the CSEU task.

For $x\in\{0,1\}^n$, define quantum state $\rho_x=|x, 0\rangle\langle x, 0|$ and observable $O_x=|x, 1\rangle\langle x, 1|\in\obs(1)$. 
Then for any Boolean function $c:\{0,1\}^n \rightarrow\{0,1\}$, its corresponding $\mathrm{MQ}(c)$ oracle (see Definition~\ref{defn:LearnBoolean}) satisfies 
\begin{align}\label{eq:Bool=Tr}
	\operatorname{Tr}\left[O_x \mathrm{MQ}(c) \rho_x \mathrm{MQ}(c)^{\dagger}\right]
	= \big| \< x, 1| x, c(x)\> \big|^2 
	= c(x). 
\end{align}
Recall that by using $N \mathcal{O} (n)$ queries of the $\mathrm{MQ}(c)$ oracle, the protocol $\mathcal{B}^*$ can 
estimate all $\Tr\left[O_x \mathrm{MQ}(c) \rho_x \mathrm{MQ}(c)^{\dagger}\right]$ up to precision $\epsilon=1/3$, with failure probability at most $1/3$. 
Let $\tilde{E}_x$ be the estimate for $\Tr\left[O_x \mathrm{MQ}(c) \rho_x \mathrm{MQ}(c)^{\dagger}\right]$ and define the function 
\begin{align}
	h(x):=\left\{\begin{array}{lll}
		0 & \text { when } & \tilde{E}_x\leq 1/2,  \\
		1 & \text { when } & \tilde{E}_x>    1/2.
	\end{array}\right. 
\end{align}
According to Eq.~\eqref{eq:Bool=Tr}, with probability at least $2/3$, we have $h(x)=c(x)$ for all $x \in\{0,1\}^n$.
Therefore, $\mathcal{B}^*$ can successfully achieve the exact learning task of Boolean functions by using $N \mathcal{O} (n)$ queries.  
This query number cannot be smaller than the lower bound in \lref{lem:LearnBoolean};
that is, $N \mathcal{O} (n) \geq \Omega\left(2^n\right)$. 
It follows that $N \geq \Omega\left(2^n/n\right)$, which completes the proof of Theorem~\ref{thm:jointLowerBound}.

\section{CSEU in the average case}\label{sec:AverageCase}
We have so far considered CSEU in the worst-case scenario, where the learner is asked to make an accurate 
prediction for any input state $\rho$ and observable $O$. 
As shown in Sec.~\ref{sec:shadowLB}, at least $\tilde\Omega\left(d\right)$ queries of the 
unitary channel $\mathcal{U}$ are required to achieve this goal. 
In this section, we turn to consider CSEU in the average-case scenario, and show that 
the query complexity lower bound $\tilde\Omega\left(d\right)$ can be broken in many parameter regimes.

In the average-case scenario, we still consider that an arbitrary observable $O\in\obs(B)$ is input into the prediction phase. 
Using the classical data obtained from the learning phase, the goal is to find a real-valued function $E(O,\rho)$ 
that can predict $\Tr\left( O\,\mathcal{U}(\rho)\right)$ within a small mean absolute error
\begin{align}\label{eq:Xi(caT,O)}
\Xi(\caT,O):= \underset{\rho \sim \caT}{\rm{Ave}} \, \left|E(O,\rho)-\Tr\left( O\,\mathcal{U}(\rho)\right) \right| ,  
\end{align}
where the average is over quantum states $\rho$ that obey some specific distribution $\caT$ on $\caD(\caH)$.\footnote{In this section we use the notation ``$\underset{\rho \sim \caT}{\rm{Ave}}$'' instead of ``$\underset{\rho \sim \caT}{\E}$'' to avoid confusion with the meaning of the expectation $\E$ in Sec.~\ref{sec:OurProtocol}.}

The above task can be solved via our learning protocol developed in Sec.~\ref{sec:OurProtocol}:  
in the prediction phase, we set the number of batches $R=1$, and use $\hat{E}(O,\rho)$ in \eref{eq:definehatE} 
as the function for estimating $\Tr\left( O\,\mathcal{U}(\rho)\right)$. 
As clarified in the following lemma, the performance of our protocol in the average-case scenario is 
closely related to the average variance of  $\hat{E}(O,\rho)$, see Appendix~\ref{sec:AverageGoal} for a proof.  
\begin{lemma}\label{lem:AverageGoal}
Suppose $\caT$ is a state distribution on $\caD(\caH)$, $O\in\obs(B)$, accuracy parameters $0<\delta,\epsilon<1$, and 
\begin{align}\label{eq:AverageGoalcrit}
\underset{\rho \sim \caT}{\rm{Ave}} \, \Var \left[ \hat{E}(O,\rho) \right] 
\leq \frac{\delta\epsilon^2}{4},  
\end{align}
Then, with failure probability at most $\delta$, the mean absolute error $\Xi(\caT,O)< \epsilon$. 
\end{lemma}

When $R=1$, the function $\hat{E}(O,\rho)$ reduces to the mean estimate $\hat{Z}_{(r)}(O,\rho)$ in \eref{eq:definetildeZ}, whose 
variance has been clarified in \pref{prop:VarMain}. 
So we can use \pref{prop:VarMain} and \lref{lem:AverageGoal} to determine the number of queries needed by our protocol in the average-case scenario for any state distribution $\caT$. 

In the following, we assume that $\caT=\pi_{d,\lambda}$ is a distribution of random density operators in $\caD(\caH)$ 
constructed by taking partial trace of Haar-random pure states on a composite system \cite{collins2016random,nechita2007asymptotics,zyczkowski2011}.

\begin{defn} \label{defn:StateDist}
The \emph{$\lambda$-induced distribution $\pi_{d,\lambda}$ of density matrix}
is the distribution on $\mathcal{D}(\mathcal{H})$ induced by the uniform distribution of pure states in $\mathcal{H} \otimes \mathcal{H}_\lambda$, where the dimensions of $\mathcal{H}$ and $\mathcal{H}_\lambda$ are $d$ and $\lambda$, respectively. 
A state $\rho$ following the distribution $\pi_{d,\lambda}$ can be generated by $\rho=\operatorname{Tr}_\lambda(|\varphi\rangle\langle\varphi|)$, 
where $|\varphi\rangle\in\mathcal{H} \otimes \mathcal{H}_\lambda$ is a pure state drawn from the Haar measure. 
\end{defn}

The average purity of quantum states $\rho\sim\pi_{d,\lambda}$ reads $\frac{d+\lambda}{\lambda d +1}\approx\frac{1}{d}+\frac{1}{\lambda}$ \cite{lubkin1978entropy}. 
This fact together with \pref{prop:VarMain} immediately implies the following corollary.

\begin{cor}\label{cor:VarAverage}
Suppose $1\leq s\leq d$ is the number of systems collectively measured in the learning phase, 
$R=1$ is the number of batches we set in the prediction phase, dimension $\lambda\leq d$, and observable $O\in\obs(B)$. Then  
\begin{align}
&\underset{\rho \sim \pi_{d, \lambda}}{\rm{Ave}}  \Var \!\big[ \hat{E}(O,\rho) \big] 
\nonumber\\
&\quad 
=\mathcal O  \bigg[\, \frac{1}{m} \left( \frac{d}{s\lambda} + \min\left\lbrace 1, \frac{B}{\lambda} \right\rbrace  \right) 
+ \frac{d^4 B}{m^2 s^4 \lambda} \bigg], 
\end{align}
where $m$ is the number of snapshots obtained from the learning phase.   
In particular, when $s=1$ we have  
\begin{align}
\underset{\rho \sim \pi_{d, \lambda}}{\rm{Ave}}  \Var \!\big[ \hat{E}(O,\rho) \big]  
	&=\mathcal O  \left( \frac{d}{m \lambda} + \frac{d^4B}{m^2 \lambda}\right). 
\end{align}
\end{cor}

From Corollary~\ref{cor:VarAverage} and \lref{lem:AverageGoal}, we can easily 
determine the query number $m\cdot s$ of our protocol required to carry out the CSEU task in the average-case scenario with distribution $\pi_{d, \lambda}$.
First, if we are limited to single-copy measurements, i.e., $s=1$, then 
\begin{align}
m\cdot s= \bigo{ \frac{d}{\lambda \,\epsilon^2} + \frac{d^2}{\epsilon} \sqrt{\frac{B}{\lambda}} } 
\end{align}
queries are sufficient to ensure that the mean absolute error satisfies $\Xi(\pi_{d, \lambda},O)<\epsilon$ with high probability. 
This query complexity is lower than that of our protocol in the worst-case scenario [see~\eref{eq:singleUB}].

Second, if collective measurements on multiple systems are accessible, then the conditions 
in \lref{lem:AverageGoal} can be satisfied by choosing $s=\lfloor d\lambda^{-1/4} \rfloor$ and 
\begin{align}
m= 
\mathcal O  \left[\, \frac{1}{\epsilon^2} \left( \frac{1}{\lambda^{3/4}} + \min\left\lbrace 1, \frac{B}{\lambda} \right\rbrace  \right) 
+ \frac{\sqrt{B}}{\epsilon}  \,\right],  
\end{align}
in which case the query number reads 
\begin{align}
m\cdot s= 
\mathcal O  \left[\,  \frac{d}{\epsilon^2} \left( \frac{1}{\lambda} + \frac{\min\left\lbrace \lambda, B \right\rbrace}{\lambda^{5/4}}   \right) + \frac{d \sqrt{B}}{\lambda^{1/4} \epsilon}  \,\right].   
\end{align}
When $B$ and $\epsilon$ are independent of the system size, the query complexity reduces to $m\cdot s=\mathcal O(d\lambda^{-1/4})$.  
Notably, $\mathcal O(d^{3/4})$ queries of $\caU$ are sufficient when $\lambda=\Theta(d)$. 
Hence, the query lower bound in \tref{thm:jointLowerBound} for the worst-case scenario does not apply to the average-case scenario.  
Previously, such remarkable separation in query complexity between the worst- and average-case considerations 
has been demonstrated in certain channel learning tasks \cite{HKPITB21,zhao2023learning}.

\section{Comparison with previous works}\label{sec:PriorWork}
In this section, we survey existing works related to CSEU and compare them with our results. 
The performance of these protocols on the CSEU task are summarized in Table~\ref{tab:compare} in Sec.~\ref{sec:Introduction}.

\subsection{CSE of general quantum channels}\label{sec:CSEC}
By applying the concepts of classical shadow formalism to quantum states, two recent works \cite{Kunjummen23,PhysRevResearch.6.013029} considered CSE of general quantum channels (CSEC). 
They proposed two protocols for CSEC that do not utilize any quantum memory.  
The first protocol is ancilla-assisted: the learner first prepares the Choi state $\frac{1}{d}\Upsilon_{\mathcal{C}}$ of the unknown channel $\mathcal{C}$ by applying $\mathcal{C}$ to the first system of the maximally entangled state $|\Phi\>\<\Phi|/d$ on $\caH^{\otimes 2}$, and then performs CSE on this Choi state. 
For any state $\rho$ and observable $O$, the property $\Tr\left( O\,\mathcal{C}(\rho)\right)$ can be recovered by virtue of channel-state duality \cite{choi1975completely}:
\begin{align}\label{eq:CSDuality}
\Tr\left( O\,\mathcal{C}(\rho)\right)
= 
d \Tr\left[ \left( O\otimes\rho^{\top}\right)  \frac{\Upsilon_{\mathcal{C}}}{d} \right].
\end{align}
If the $M$ observables $O_l$ to be estimated are chosen from $\obs(B)$, this protocol requires $\mathcal{O}(d^2 B \epsilon^{-2}\log M)$ applications of $\mathcal{C}$
to estimate the properties in \eref{eq:LinearProperties} all within additive error $\epsilon$. 
However, this protocol requires an ancillary system with dimension $d$, which may limit its practical applicability.

The second protocol for CSEC is ancilla-free \cite{Kunjummen23,PhysRevResearch.6.013029}.  
In each round of experiments, one prepares a random pure state on $\caH$, 
evolves it under $\mathcal{C}$, and then performs a randomized measurement on the output state. 
This procedure is identical to that of our protocol based on single-copy measurements. 
Using the experimental results, one can construct an unbiased estimator for the Choi operator $\Upsilon_{\mathcal{C}}$, 
from which any desired linear property can be recovered. 
To predict $M$ properties in \eref{eq:LinearProperties} all 
within additive error $\epsilon$, this protocol requires 
\begin{align}
\mathcal{O} \left( \frac{d^2 \log M}{\epsilon^2} \max_{1\leq l\leq M} \,\left\|S(O_l)\right\|  \right) 
\end{align} 
applications of $\mathcal{C}$
(see Theorem 2 in Ref.~\cite{Kunjummen23}), where   
\begin{align}\label{eq:S(O)}
S(O):=4 [\operatorname{Tr}(O)]^2 I +2\operatorname{Tr}(O^2)I+4 \operatorname{Tr}(O) O+4 O^2 .
\end{align}
Notice that $\|S(O)\|$ can be very large for some specific observables in $\obs(B)$. 
In particular, due to the first term in \eref{eq:S(O)}, it attains its maximum scaling $\Theta(dB)$ when $O=\sqrt{d^{-1}B}\, I\in\obs(B)$. 
As a result, if all observables $O_l$ to be predicted are chosen from $\obs(B)$, this protocol has a $\mathcal{O}(d^3 B \epsilon^{-2}\log M)$ query complexity.

In contrast to the two protocols mentioned above, our protocol is designed for predicting properties of unitary channels rather than general channels. 
However, our protocol based on single-copy measurements not only works without quantum memory and ancilla, but also has much higher query efficiency than the above  protocols (see Table~\ref{tab:compare}).
To better visualize this advantage, let us consider the case in which $\{O_l\}$ is the collection of $4^n$ Pauli observables, 
and high-precision estimation with error $\epsilon=1/d$ is desired. 
In this case, our CSEU protocol that does not use quantum memory needs $\tilde{\mathcal{O}}(d^{3.5})$ queries, whereas the CSEC protocols mentioned above \cite{Kunjummen23,PhysRevResearch.6.013029} need $\tilde{\mathcal{O}}(d^{5})$ queries.
The intrinsic reason behind this advantage is that our protocol ‌‌employs the property of the snapshots $\hat{X}_i$ presented in \eref{eq:ExpXiXj} to construct the quadratic estimator $d^{-1}\!\Tr[(O\otimes\rho^{\top}) \hat{X}_i\hat{X}_j]$ (see Sec.~\ref{sec:PredPhase}). For a general quantum channel, \eref{eq:ExpXiXj} may fail and  such an unbiased quadratic estimator cannot be constructed.

\subsection{Learning protocol based on Ref.~\texorpdfstring{\cite{Grier22}}{[Grier22]}}\label{sec:Grier22}
A recent work \cite{Grier22} found that when collective measurements are used, 
the task of CSE for pure states can be solved with a lower sample complexity 
than the original shadow protocol in Ref.~\cite{HKPshadow20}. More specifically, suppose we  
are given multiple copies of an unknown pure state $\sigma$, and $Q_1,\dots,Q_M$ are arbitrary observables in $\obs(B)$.  
To predict $M$ expectation values $\Tr(Q_i\sigma)$ all within additive error $\epsilon$, 
Ref.~\cite{Grier22} showed that it suffices to use the following number of samples: 
\begin{align}\label{eq:ComplexityPure}
\bigo{ \bigg( \frac{\sqrt{B}}{\epsilon} + \frac{1}{\epsilon^2}\bigg) \log M}  .
\end{align}

Based on this result, we can naively construct the following protocol for CSEU:  
one first prepares the (pure) Choi state $\frac{1}{d}\Upsilon_{\caU}$ of the unitary channel $\caU$, then performs 
CSE on the Choi state using the protocol of Ref.~\cite{Grier22}. 
Suppose $\rho_1,\rho_2,\dots, \rho_M$ are quantum states in $\caD(\caH)$ and $O_1,O_2,\dots, O_M$ are observables in $\obs(B)$. 
According to \eref{eq:ComplexityPure}, 
to estimate $\Tr[ ( O_l\otimes\rho_l^{\top}) \frac{1}{d}\Upsilon_{\caU}]$ within error $0<\eta<1$ for all $l$, 
it suffices to query the channel 
\begin{align}
\bigo{ \bigg( \frac{\sqrt{B}}{\eta} + \frac{1}{\eta^2}\bigg) \log M}
\end{align}
times. In addition, by \eref{eq:CSDuality}, to estimate 
$\Tr\left( O_l\,\caU(\rho_l)\right)$ within $\epsilon$ precision, 
one needs to keep the additive error $\eta$ in estimating $\Tr[ ( O_l\otimes\rho_l^{\top}) \frac{1}{d}\Upsilon_{\caU}]$ smaller than $d^{-1}\epsilon$. 
As a result, the query complexity of this protocol reads  
\begin{align}
\bigo{ \left( \frac{\sqrt{B}}{d^{-1}\epsilon} + \frac{1}{[d^{-1}\epsilon]^2}\right) \log M }  
=\bigo{ \frac{d^2}{\epsilon^2} \log M},  
\end{align}
where we use the relation $B\leq d$. 
By contrast, the query complexity of our protocol based on collective measurements is 
$\mathcal O \big( d \,\big(\epsilon^{-2} \!+\! \sqrt{B}\epsilon^{-1} \big) \log M\big)$, which improves the scaling behavior in the dimension $d$ quadratically.

\subsection{Process tomography of unitary channels}\label{sec:ProcessTomg}
Another approach for completing the CSEU task is to first use quantum process tomography of unitary channels (QPTU)  \cite{PhysRevLett.93.080502,PhysRevLett.125.210501,haah2023query,surawy2022projected} 
to construct a classical description of the unknown channel $\mathcal{U}$, and then use this description 
to estimate $\Tr\left( O\,\mathcal{U}(\rho)\right)$. 
When quantum memory is accessible, Ref.~\cite{haah2023query} showed that $\Theta(d^2\epsilon^{-1})$ 
queries of $\mathcal{U}$ are necessary and sufficient to output a classical
description that is $\epsilon$-close to $\mathcal{U}$ in the diamond norm. 
Consequently, any linear property $\Tr\left(O\,\mathcal{U}(\rho)\right)$ 
can be estimated within additive error $\epsilon$  from $\mathcal{O}(d^2\epsilon^{-1})$ queries. 
In the tomography protocol of Ref.~\cite{haah2023query}, the unknown channel $\mathcal{U}$ is applied in series on one system of dimension $d$, so quantum memory is needed but ancillary systems are not. 
Compared with our protocol based on collective measurements, this protocol has a smaller space overhead, 
but a larger query complexity for completing the CSEU task when $d>\epsilon^{-1}$, which is the usual case.

When quantum memory is not accessible, the protocol in Ref.~\cite{surawy2022projected} 
has the highest efficiency among all QPTU protocols as far as we are aware. 
Using $\mathcal{O}(d^2\eta^{-2})$ queries of the unknown channel $\mathcal{U}$, 
this protocol can learn a classical description of the Choi state $\frac{1}{d}\Upsilon_{\caU}$ 
within $0<\eta<1$ trace distance (see Theorem~1 in Ref.~\cite{surawy2022projected}).
From this description, one can then estimate $\Tr[(O\otimes\rho^{\top}) \frac{1}{d}\Upsilon_{\caU}]$ 
within additive error $\eta$ for any state $\rho$ and observable $O$. 
According to \eref{eq:CSDuality}, this implies a $\mathcal{O}(d^4\epsilon^{-2})$ 
query complexity for estimating $\Tr\left( O\,\mathcal{U}(\rho)\right)$ within $\epsilon$ 
error---nearly a quadratic overhead compared with our CSEU protocol that does not use quantum memory.

\subsection{Other works}\label{sec:OtherWorks}
Recently, Ref.~\cite{HCP23} proposed an efficient protocol without quantum memory 
for learning to predict $\Tr\left( O\,\mathcal{C}(\rho)\right)$ for an arbitrary unknown channel $\mathcal{C}$, 
achieving a small average error over input states $\rho$ drawn from any locally flat distribution. 
Here, $O$ is restricted to a \emph{bounded-degree observable} that  can be expressed as a sum of local observables, 
each has support on a constant number of qubits. 
In contrast, our protocol for CSEU focuses on unitary channels, but 
can deal with the worst-case input states and general observables. 

Using quantum memory, Ref.~\cite{caro2022learning} also gave an efficient procedure 
for predicting linear properties $\Tr\left( O\,\mathcal{C}(\rho)\right)$ for an arbitrary unknown channel $\mathcal{C}$. 
However, here the input state $\rho$ is restricted to  Pauli-sparse quantum states, and the observable $O$ is restricted to  bounded Pauli-sparse observables.

Instead of general quantum channels, recent works \cite{zhao2023learning,huang2024learning,huang2022quantum,HKPITB21,caro2022generalization} 
 considered quantum channels or unitaries $\mathcal{C}$ that can be generated in polynomial-time, 
and showed that a polynomial amount of queries are sufficient to estimate $\Tr\left( O\,\mathcal{C}(\rho)\right)$ in this restricted class. 
By contrast, our protocol for CSEU applies to arbitrary unitary channels, 
in which case an exponential number of queries is inevitable according to Theorem~\ref{thm:jointLowerBound}.

\section{Application in Hamiltonian learning}\label{sec:Hamiltonian}

\subsection{Background and our result}
Learning an unknown Hamiltonian is a central task in many-body quantum physics.  
In an isolated quantum system, the Hamiltonian $H$ fully characterizes the system's 
time evolution through the dynamic operator $\rme^{-\rmi Ht}$, 
so all information about the future state of the system is determined by its initial state and $H$. 
Beyond its fundamental significance in understanding quantum dynamics, 
Hamiltonian learning has broad applications across various fields, 
including quantum metrology and sensing \cite{giovannetti2011advances,RevModPhys.89.035002}, 
and quantum device engineering \cite{shulman2014suppressing,innocenti2020supervised}.

Here we consider the problem of learning an unknown $n$-qubit Hamiltonian $H$ with respect to the \emph{normalized Frobenius norm} (NFN). 
More explicitly, given time evolution access to $H$,
our goal is to learn a classical description $\hat{H}\in \mathcal{L}_{\rm H}(\mathcal{H})$ of the Hamiltonian, 
such that 
\begin{equation}\label{eq:HamiltonianGoal}
\Pr\left\{\frac{1}{\sqrt{2^n}}\left\|\hat{H}-H\right\|_2\leq\epsilon\right\} \geq 1-\delta 
\end{equation}
for accuracy $0< \delta,\epsilon<1/3$. 
Here we do not make prior assumptions on the structure of $H$, 
but only have two natural constraints $\Tr(H) = 0$ and $\|H\|=\mathcal{O}(\poly(n))$. 
According to Theorem~1.4 of Ref.~\cite{bluhm2024hamiltonian}, even with access to quantum memory, 
any protocol has to make at least $\tilde\Omega(d^2)$ applications of $H$ to solve this problem.

We choose NFN as the error metric because it reflects how well we can predict the system on average based on the learned description $\hat{H}$ of the Hamiltonian \cite{ma2024learning}.
Suppose we are given a known initial state $\rho_0$, and we want to predict observation 
expectation values of the system at a given time $t$ by using $\hat{H}$. 
Starting from $\rho_0$, the true state $\rho(t)$ at time $t$ and the state $\hat\rho(t)$ we predict will be
\begin{equation}
\rho(t) = \rme^{-\rmi Ht}\rho_0 \rme^{\rmi Ht},
\qquad 
\hat\rho(t) = \rme^{-\rmi \hat{H} t}\rho_0 \rme^{\rmi \hat{H} t}.
\end{equation}
Accordingly, the true and predicted expectation values of a given observable $O$ will be 
\begin{equation}
\<O\>_{\rho(t)}=\Tr[O\rho(t)],
\qquad 
\<O\>_{\hat\rho(t)}=\Tr[O\hat\rho(t)], 
\end{equation}
respectively.
The absolute value of their difference represents our prediction error based on $\hat H$.
If the initial state $\rho_0=\ketbra{\psi_0}{\psi_0}$ is a random pure state drawn from the Haar distribution, then 
the mean squared error can be bounded as \cite{ma2024learning}
\begin{align}\label{eq:ExpectBound}
&\underset{|\psi_0\> \sim \mu_{\rm H}}{\E} \left[ \abs{\<O\>_{\rho(t)} - \<O\>_{\hat\rho(t)}}^2 \right]
\leq  2 \biggl(\!t\,\|O\|_2 \frac{\|\hat{H}-H\|_2}{\sqrt{2^n}} \biggr)^2. 
\end{align}
Therefore, 
a small NFN-error bound guarantees a high average accuracy in predicting 
the expectation values of observables under a random initial state \cite{ma2024learning,Zhao2022Random}.

In recent years, the task of learning an unknown Hamiltonian $H$ through its unitary dynamics 
has been extensively studied in the literature \cite{Dutkiewicz2024advantageofquantum,PhysRevLett.122.020504,Yu2023robustefficient,hu2025ansatz,zhao2024learning,gu2024practical,stilck2024efficient,ma2024learning,bluhm2024hamiltonian,innocenti2020supervised,caro2022learning,PhysRevLett.130.200403,haah2024learning,castaneda2023hamiltonian,bakshi2024structure1}.
Most existing protocols use repeated applications of the Hamiltonian's time evolution in each experiment, 
thereby demanding the resource of quantum memory. 
In addition, these protocols generally rely on strong assumptions about the structure of $H$---for instance, assuming that $H$ exhibits sparse interactions or contains Pauli terms with constant weights. 
However, in many practical scenarios,  the Hamiltonian may have an unknown structure  
and involve arbitrary interaction couplings, which restrict the applicability of these protocols.

Furthermore, it is worth mentioning that learning a Hamiltonian $H$ with respect to NFN is much harder than 
simply estimating its Pauli coefficients, although the latter is the focus of most previous Hamiltonian learning protocols.
More precisely, estimating all Pauli coefficients of $H$ within additive error $\epsilon$
only implies learning $H$ within error $d\epsilon$ in NFN (see Appendix~\ref{sec:appHamiltonian}), where $d=2^n$ is the dimension of the system. 
Due to this reason, when attempting to extend existing protocols to learn an arbitrary Hamiltonian in NFN, 
the number of required experimental runs and total evolution time under $H$ would become prohibitively large, as shown in Appendix~\ref{sec:appHamiltonian}.

Here, to demonstrate the power of our CSEU protocol, we show that it can be combined with the
polynomial interpolation technique to learn an arbitrary Hamiltonian $H$ in NFN. 
The learning protocol is non-adaptive, and does not need quantum memory or ancillary qubits. 
In each round of experiments, 
it requires only the preparation of a computational-basis state, the application of a random Clifford unitary to a single system, a measurement in the computational basis, and a real $\mathcal{O}(1/\|H\|)$-time evolution under $H$. 
The protocol's query complexity and evolution time complexity are summarized in the following theorem, 
which is a simplified version of \tref{thm:HamiltonianMore} in Appendix~\ref{sec:appHamiltonian}.

\begin{theorem}\label{thm:Hamiltonian}
Given time evolution access to an arbitrary $n$-qubit traceless Hamiltonian $H$ with $\|H\|=\mathcal{O}(\poly(n))$. 
There is an ancilla-free protocol without quantum memory that outputs a classical description $\hat{H}\in \mathcal{L}_{\rm H}(\mathcal{H})$ such that \eref{eq:HamiltonianGoal} holds. 
Both the number of queries to $H$ and the total evolution time required by the protocol scale as
\begin{align}\label{eq:s1Hnumber}
	\tilde{ \mathcal{O}} 
	\left( \bigg( \frac{d^2}{\epsilon^2} + \frac{d^3}{\epsilon} \bigg) \log(\delta^{-1})\right).                       
\end{align} 
\end{theorem}

To the best of our knowledge, this is the first protocol based on $\tilde{\mathcal{O}}(d^3)$ queries and evolution time for learning an arbitrary Hamiltonian with respect to NFN. 
Notably, even without using quantum memory, the query complexity is still close to the known theoretical lower bound $\tilde\Omega(d^2)$ \cite{bluhm2024hamiltonian}, demonstrating both high efficiency and practical applicability. 
In contrast, the most efficient protocol in the literature
requires $\tilde{\mathcal{O}}(d^4)$ Hamiltonian queries and total evolution time to achieve this task even if quantum memory is accessible (see Appendix~\ref{sec:appHamiltonian}).

\subsection{Schematic overview of our Hamiltonian learning protocol}\label{sec:HamiltProtocol}
In this section, we formalize  the basic principles and methodology of our Hamiltonian learning protocol. 
Additional technical details regarding the protocol can be found in Appendix~\ref{sec:appHamiltonian}.

Our goal is to learn an arbitrary unknown traceless  Hamiltonian $H$, 
which can be expanded in the form: 
\begin{align}
	H=\sum_{\bfk\in\{0,1,2,3\}^n} \mu(\bfk) \sigma_\bfk.  
\end{align}  
Here,  $\sigma_\bfk$ are Pauli operators, and $\mu(\bfk)\in \mathbb{R}$ are Pauli coefficients with $\mu(0^n)=0$. 
For a system with an initial state $\rho$ evolving under $H$, the expectation value of an operator $O$ at time $t$ can be written as \cite{stilck2024efficient,gu2024practical,caro2022learning}:  
\begin{align}\label{eq:timeexpand}
\<O\>_{\rho(t)}
=\Tr[O\,\caU_t(\rho)]
=\sum_{k=0}^\infty \frac{(\rmi \, t)^k}{k!} \Tr\left[ g_k(H,O)\rho\right],
\end{align}
where $\caU_t(\cdot):=\rme^{-\rmi Ht}(\cdot) \rme^{\rmi Ht}$ denotes the unitary channel describing the time evolution, and 
\begin{equation}
g_k(H,O) := \underbrace{[H,[H,[\ldots, [H}_{k\textrm{ times}},O]\ldots]]] ,
\ \  
g_{k=0}(H,O) :=O. 
\end{equation}
As shown in Refs.~\cite{stilck2024efficient,gu2024practical,caro2022learning}, 
for any Pauli operator $\sigma_\bfk\ne I$, there is a quantum state $\rho_{\bfk}$
and another Pauli operator $\sigma_{\bfk'}$ that satisfy $\Tr(\rmi[H,\sigma_{\bfk'}]\rho_{\bfk})=2\mu(\bfk)$. 
This fact together with \eref{eq:timeexpand} leads to the following key observation: 
For any non-identity Pauli operator $\sigma_{\bfk}$, if the initial state $\rho$ and  observable $O$ are properly chosen, 
then the first-order time derivative of $\<O\>_{\rho(t)}$ at $t=0$ tells us the Pauli coefficient $\mu(\bfk)$. 
In addition, this derivative can be extracted by performing polynomial interpolation on the function $t \mapsto \<O\>_{\rho(t)}$ \cite{stilck2024efficient,gu2024practical,caro2022learning}. 

Building on the above observations, our protocol integrates  the CSEU method developed in Sec.~\ref{sec:OurProtocol} and the derivative estimation procedure of Ref.~\cite{caro2022learning} to learn the unknown Hamiltonian. 
In the first step, for several different time points $t$, we treat $\caU_t$ as an unknown unitary channel, 
and use our CSEU protocol to accurately estimate the expectation value $\<O\>_{\rho(t)}$ corresponding to each $\sigma_\bfk\ne I$. 
Due to the high efficiency of our CSEU method, this step requires much fewer Hamiltonian queries 
(and thus less total evolution time) than existing approaches.
The second step of our protocol is classical postprocessing, in which we estimate each Pauli coefficient $\mu(\bfk)$ within additive error $\epsilon/d$ by using polynomial interpolation, following the same procedure as in Ref.~\cite{caro2022learning}. 
Let $\hat\mu(\bfk)$ be the estimate for $\mu(\bfk)$; 
then our protocol outputs $\hat{H}=\sum_{\bfk} \hat\mu(\bfk) \sigma_\bfk$ as the learned description of $H$, which satisfies 
\begin{align}
\frac{\|\hat{H}-H\|_2}{\sqrt{d}}
&= \sqrt{\frac{1}{d} \Tr\left[( \hat{H}-H)^2\right] }
\nonumber\\ 
&= \sqrt{\sum_{\bfk\in\{0,1,2,3\}^n} |\hat{\mu}(\bfk)-\mu(\bfk)|^{2} }
\nonumber\\ 
&\leq \sqrt{\sum_{\bfk\in\{0,1,2,3\}^n} \left( \frac{\epsilon}{d} \right)^{2} }
=\epsilon. 
\end{align}
In this way, the unknown Hamiltonian $H$ is learned within error $\epsilon$ with respect to NFN.

\section{Application in estimating OTOCs}\label{sec:mainOTOC}
So far we have considered predicting linear functions of the unknown 
unitary channel $\mathcal{U}$. Actually, the classical shadow data $\US(\caU,m)$ obtained in our learning phase 
may also be used to predict non-linear functions of $\mathcal{U}$. 
One particularly crucial non-linear function is the 
out-of-time-ordered correlators (OTOCs). 
As a tool to quantify quantum information scrambling, OTOC plays critical roles in both quantum many-body physics \cite{PRXQuantum.5.010201,chen2023speed} and quantum information \cite{PhysRevLett.124.200504,PhysRevX.8.021013}. 
Suppose $U(t)$ is a unitary time-evolution operator and $W$ is a local Hermitian operator. 
For chaotic dynamics, the Heisenberg operator $W_{U(t)}=U(t)^{\dag}WU(t)$ rapidly
becomes nonlocal and non-commutative with an initially non-overlapping local operator $V$. 
The degree of this non-commutativity can be quantified by the OTOC 
\begin{align}
C(U(t))=\Tr\left(\rho W_{U(t)} V^{\dagger} W_{U(t)} V\right), 
\end{align}
with $\rho$ being the initial quantum state. 
To simplify the notation, in the following we often omit ``$(t)$'' and use $U$ to denote the 
unitary time-evolution operator.

For given $W$ and $V$, below we consider the task of predicting OTOC by querying the unknown time-evolution operator $U$ several times. 
Using the tensor-network diagram (see Appendix~\ref{sec:Tensor}), the OTOC can be written as 
\begin{align}
&C(U)
= \ \
\begin{tikzpicture}[baseline={([yshift=-.5ex]current bounding box.center)},inner sep=-4mm]
	\node[tensor_green] (U3) at (1*\xratio, -1) {$U$};
	\node[tensor_green] (U4) at (1*\xratio, -2.5) {$U^{\dag}$};
	\node[tensor_green] (U1) at (3*\xratio, -1) {$U$};
	\node[tensor_green] (U2) at (3*\xratio, -2.5) {$U^{\dag}$};
	\node[tensor_blue] (W1) at (1*\xratio, -3.4) {$W$};
	\node[tensor_blue] (V1) at (2*\xratio, -3.4) {$V^{\dag}$};
	\node[tensor_blue] (W2) at (3*\xratio, -3.4) {$W$};
	\node[tensor_blue] (V2) at (4*\xratio, -3.4) {$V\rho$};
	\draw[>>-,draw=black] (1*\xratio, -0.3) -- (U3.north);
	\draw[>>-,draw=black] (3*\xratio, -0.3) -- (U1);
	\draw[-<<,draw=black] (U3.south) .. controls (1*\xratio, -1.7) and (2*\xratio, -2.1) ..  (2*\xratio, -0.3);
	\draw[-<<,draw=black] (U1.south) .. controls (3*\xratio, -1.7) and (4*\xratio, -2.1) ..  (4*\xratio, -0.3);
	\draw[-,draw=black]   (U2.north) .. controls (3*\xratio, -1.8) and (4*\xratio, -1.4) ..  (4*\xratio, -2.7);
	\draw[-,draw=black]   (U4.north) .. controls (1*\xratio, -1.8) and (2*\xratio, -1.4) ..  (2*\xratio, -2.7);
	\draw[-,draw=black]   (2*\xratio, -2.7) .. controls (2*\xratio, -2.9) and (1.6*\xratio, -2.7) ..  (1.6*\xratio, -3.4);
	\draw[-,draw=black]   (V1.south) .. controls (2*\xratio, -3.9) and (1.6*\xratio, -4.1) ..  (1.6*\xratio, -3.4);
	\draw[-,draw=black]   (V1.north) .. controls (2*\xratio, -2.9) and (2.4*\xratio, -2.6) ..  (2.4*\xratio, -3.4);
	\draw[-,draw=black]   (2.4*\xratio, -3.4) .. controls (2.4*\xratio, -4.1) and (2*\xratio, -3.7) .. (2*\xratio, -4.1) ;
	\draw[-,draw=black]   (4*\xratio, -2.7) .. controls (4*\xratio, -2.9) and (3.6*\xratio, -2.7) ..  (3.6*\xratio, -3.4);
	\draw[-,draw=black]   (V2.south) .. controls (4*\xratio, -3.9) and (3.6*\xratio, -4.1) ..  (3.6*\xratio, -3.4);
	\draw[-,draw=black]   (V2.north) .. controls (4*\xratio, -2.9) and (4.4*\xratio, -2.6) ..  (4.4*\xratio, -3.4);
	\draw[-,draw=black]   (4.4*\xratio, -3.4) .. controls (4.4*\xratio, -4.1) and (4*\xratio, -3.7) .. (4*\xratio, -4.1) ;
	\draw[-,draw=black] (W1) -- (U4);
	\draw[-,draw=black] (W2) -- (U2);
	\draw[-<<,draw=black] (W1.south) .. controls (1*\xratio, -4.9) and (3*\xratio, -4.3) ..  (3*\xratio, -5.5);
	\draw[-<<,draw=black] (W2.south) .. controls (3*\xratio, -4.9) and (1*\xratio, -4.3) ..  (1*\xratio, -5.5);
	\draw[-<<,draw=black] (2*\xratio, -4.1) -- (2*\xratio, -5.5);
	\draw[-<<,draw=black] (4*\xratio, -4.1) -- (4*\xratio, -5.5);
\end{tikzpicture}
\nonumber\\[1ex]
&= \Tr \left[ (\Upsilon_\caU\otimes\!\Upsilon_\caU) 
(W \! \otimes (V^{\dag})^{\top}\! \otimes W \! \otimes (V\rho)^{\top})  T_{(1,3)}\right] .  \!
\end{align}
In this diagram, the order in which operators are multiplied is from top to bottom, 
“\,$\begin{tikzpicture}[baseline={([yshift=-.5ex]current bounding box.center)},inner sep=-4mm]
	\draw[-<<,draw=black] (1*\xratio, -0.3)-- (1*\xratio, 0);
\end{tikzpicture}$\,”
and 
“\,$\begin{tikzpicture}[baseline={([yshift=-.5ex]current bounding box.center)},inner sep=-4mm]
	\draw[-<<,draw=black] (1*\xratio, 0.3)-- (1*\xratio, 0);
\end{tikzpicture}$\,”
means the periodic boundary condition for top and bottom legs, and $T_{(1,3)}$ is the swap operator on the first and third systems. 
Recall that each snapshot $\hat{X}$ in the classical shadow data $\US(\caU,m)$ equals the Choi operator $\Upsilon_\caU$ in expectation. So one can predict $C(U)$ by using the following estimator constructed from pairs of distinct snapshots, 
\begin{align}\label{eq:hatCU}
\hat{C}(U)
&:= 
\frac{1}{m(m-1)}\sum_{i\ne j} \Tr\!\Big[ \left( \hat{X}_i\otimes\hat{X}_j\right) 
\nonumber\\
&\qquad\quad  \times\left( W\otimes (V^{\dag})^{\top} \otimes W\otimes (V\rho)^{\top}\right)   T_{(1,3)}\Big] .  \!  
\end{align}
Here the summation is over all $m(m-1)$ quadratic terms with $i,j=1,\dots,m$ and $i\ne j$. 
Due to the independence between $\hat{X}_i$ and $\hat{X}_j$, this estimator correctly predicts the OTOC in expectation.

To study the query efficiency of this approach, 
following Refs.~\cite{PhysRevX.9.021061,PhysRevResearch.3.033155}, 
here we focus on the infinite temperature thermal state $\rho=I/d$, 
and assume that $W$ and $V$ are both unitary, traceless, and Hermitian operators. 
In this case, as proved in Appendix~\ref{sec:proofVarOTOC}, the variance of our estimator reads 
\begin{align}\label{eq:VarOTOC}
	\Var \left[ \hat{C}(U) \right] 
	&= \mathcal O \left[\, \frac{1}{m} \left( \frac{d^2}{s^2} + 1 \right)
	+ \frac{d^2}{m^2} \left( \frac{d^4}{s^4} + 1 \right) \right] .                         
\end{align} 
By combining this result with Chebyshev's inequality, we immediately arrive at the following proposition, 
which clarifies the query complexity for estimating OTOC. 
\begin{prop}\label{prop:complexityOTOC}
In order to estimate $C(U)$ within additive error $0<\epsilon<1$ and some high (constant) success probability, 
the number of queries used by our protocol is 
\begin{align}\label{otoc:single}
m\cdot s = \mathcal O \bigg( \frac{d^2}{\epsilon^2} + \frac{d^3}{\epsilon} \bigg) 
\end{align} 
when we use only single-copy measurements, i.e., $s=1$; and is
\begin{align}\label{otoc:collective}
m\cdot s = \mathcal O \bigg( \frac{d}{\epsilon^2} + \frac{d^2}{\epsilon} \bigg) 
\end{align}
when we use collective measurements on $d$ systems, i.e., $s=d$.
\end{prop}	

Due to its shadow nature, our protocol can estimate many OTOCs for different initial operators $W$ and $V$ simultaneously by classical post-processing.  In contrast, previous methods \cite{PhysRevX.9.021061,PhysRevResearch.3.033155} need to encode the information on $W$ and $V$ into the initial state preparation, unitary evolution, and final measurements. As a result, the measurement outcomes can only  be used to estimate the chosen OTOC. In addition, our protocol based on collective measurements can significantly enhance the query efficiency compared with Ref.~\cite{PhysRevResearch.3.033155}, as shown in \eref{otoc:collective}.
Furthermore, it is  straightforward to generalize 
our protocol to predict higher-order OTOCs \cite{PhysRevResearch.3.033155,Leone2021quantumchaosis}.

\section{Conclusion}\label{sec:Conclusion}
In this work, we initiated a comprehensive investigation of the CSEU problem, 
which is of fundamental interest to both theoretical studies and practical applications.
By utilizing the power of collective measurements, we proposed a novel protocol for the CSEU task, 
whose query complexity is nearly asymptotically optimal for prediction in the worst case, and 
can be further reduced for prediction in the average case. 
To further enhance practicality, we also introduced an alternative ancilla-free protocol based on single-copy measurements, which achieves much higher query efficiency compared with previous protocols that do not use quantum memory. 
This protocol can serve as a key subroutine for learning an arbitrary unknown Hamiltonian with respect to NFN, significantly outperforming existing approaches for this task.
Furthermore, our protocols can be extended to predict non-linear properties of unitary channels, and have demonstrated high efficiency in estimating OTOC, a crucial quantity in quantum many-body physics. 
This work represents a significant advance in the study of learning quantum channels, 
and is expected to trigger a cascade of future research works on this subject.

There are some intriguing problems that merit further investigation. 
First, although our protocol based on collective measurements achieves a nearly optimal query complexity, it applies the unknown unitary channel $\caU$ in parallel, resulting in a large space overhead. It is natural to ask whether this overhead can be avoided by considering sequential applications of $\caU$ \cite{haah2023query}, which would make the protocol more practical. 
Second, our protocol applies quadratic estimators in the prediction phase to estimate linear properties (see Sec.~\ref{sec:PredPhase}). It would be interesting to study the potential benefits of higher-order estimators constructed from more snapshots, which may further reduce the query complexity with respect to the parameter $B$. 
The third question is whether our protocol based on single-copy measurements can achieve the optimal query complexity with respect to the system dimension $d$ over all protocols that do not rely on quantum memory. 
We leave these lines of research to future study.

\section*{Acknowledgments}
The authors thank Penghui Yao, Zhenhuan Liu, Datong Chen, Fuchuan Wei, Qi Ye, and Masahito Hayashi 
for inspiring  discussions and valuable comments. 
Z.L., C.Y., and H.Z. are supported by 
the Shanghai Science and Technology Innovation Action Plan (Grant No.~24LZ1400200), 
the Shanghai Municipal Science and Technology Major Project (Grant No.~2019SHZDZX01), 
the National Key Research and Development Program of China (Grant No.~2022YFA1404204), 
and the National Natural Science Foundation of China (Grant No.~92165109). 
H. Z. is also supported by the Innovation Program for Quantum Science and Technology (Grant No.~2024ZD0300101).
Y.Z. is supported by 
the Innovation Program for Quantum Science and Technology (Grant Nos.~2024ZD0301900 and 2021ZD0302000), 
the National Natural Science Foundation of China (Grant No.~12205048), 
the Shanghai Science and Technology Innovation Action Plan (Grant No.~24LZ1400200), 
and the start-up funding of Fudan University.

\let\addcontentsline\oldaddcontentsline


\begin{appendix}
	
\begin{widetext}
\section{Frequently used symbols}\label{sec:symbols}
In Table~\ref{tab:symbols}, we list the frequently used symbols throughout this paper for readers' reference.

\begin{table}[htbp]
\caption{\label{tab:symbols}
Frequently used symbols and their meanings. 
}		
\begin{math}
\begin{array}{c|c}
\hline\hline
\text{Symbol}  & \text{Meaning} \\[0.5ex]
\hline
\caU       & \text{the unknown unitary channel we aim to learn}            \\[0.5ex]
\Upsilon_\caU      & \text{the Choi operator of $\caU$}            \\[0.5ex]
\caH       & \text{the $n$-qubit Hilbert space that $\caU$ acts on}    \\[0.5ex]
d          & \text{$d=2^n$ is the dimension of $\caH$}    \\[0.5ex]
\caD(\caH) & \text{the set of all density operators on $\caH$}    \\[0.5ex]
\obs(B)    & \text{the set of observables on $\caH$ with bounded norms; see \eref{eq:obsB}}    \\[0.5ex]
\mathcal{A}_{\text {learn}} & \text{our algorithm working in the learning phase; see Sec.~\ref{sec:LearnPhase}}    \\[0.5ex]
\mathcal{A}_{\text {pred}}  & \text{our algorithm working in the prediction phase; see Sec.~\ref{sec:PredPhase}}    \\[0.5ex]
\hat X_i   & \text{a classical snapshot of $\caU$ obtained by $\mathcal{A}_{\text {learn}}$; see \eref{eq:defhatX}}    \\[0.5ex]
\mathcal{M}_{s}    & \text{the symmetric collective measurement  used in $\mathcal{A}_{\text {learn}}$; see \eref{eq:measurement}}    \\[0.5ex]
s          & \text{the number of copies jointly measured by $\mathcal{M}_{s}$ in the learning phase}   \\[0.5ex]
m          & \text{the number of snapshots obtained by $\mathcal{A}_{\text {learn}}$}         \\[0.5ex]
R          & \text{the number of batches used in $\mathcal{A}_{\text {pred}}$}                  \\[0.5ex]	
q          & \text{$q=m/R$ is the number of snapshots in each batch}                  \\[0.5ex]	
\hat{Z}_{(r)}    & \text{the mean estimator constructed from the $r$th batch; see \eref{eq:definetildeZ}}     \\[0.5ex]	
\hat{E}(O,\rho)   & \text{our protocol's final estimate for $\Tr\left( O\,\mathcal{U}(\rho)\right)$; see \eref{eq:definehatE}}     \\[0.5ex]	
\epsilon   & \text{the error that we can tolerate}                  \\[0.5ex]	
\delta     & \text{the failure probability that we can tolerate}                  \\[0.5ex]	
\Xi(\caT,O)\!      & \text{the mean absolute error for predicting $\Tr\left( O\,\mathcal{U}(\rho)\right)$; see \eref{eq:Xi(caT,O)}}       \\[0.5ex]
\pi_{d,\lambda}	   & \text{the $\lambda$-induced distribution of density matrices; see Definition~\ref{defn:StateDist}}    \\[0.5ex]
H            & \text{the unknown Hamiltonian that we aim to learn}    \\[0.5ex]
\sym^{(t)}   	   & \text{the projector onto the symmetric subspace of $\caH^{\otimes t}$}    \\[0.5ex]
\mathrm{S}_t   	   & \text{the permutation group of order $t$}    \\[0.5ex]
T_\pi   	       & \text{the permutation operator on $\mc{H}_d^{\otimes t}$ corresponding to $\pi\in\mathrm{S}_t$}    \\[0.5ex]
\hline\hline
\end{array}	
\end{math}
\end{table}

\end{widetext}

\section{Preliminaries}\label{sec:Prelimi}
For our proofs, it will be useful to have some familiarity with quantum state designs and tensor-network diagrams.  
A more detailed review can be found in Ref.~\cite{Mele23}.

\subsection{Quantum state designs}\label{sec:Design}
A \emph{Haar-random state} $|\psi\>$ in the Hilbert space $\mathcal{H}$ is a pure state drawn randomly according to the Haar measure $\mu_{\rm H}$ on quantum states. 
For integer $t\geq 1$, the expectation of $(|\psi\>\<\psi|)^{\otimes t}$ with respect to the Haar measure reads 
\begin{align}
\underset{|\psi\rangle \sim \mu_{\rm H}}{\E} \left[ (|\psi\rangle\langle\psi|)^{\otimes t} \right]
=\frac{1}{\kappa_t }\Pi_{\mathrm{sym}}^{(t)}, 
\quad
\kappa_t=\binom{t+d-1}{t}, 
\end{align}
where $\sym^{(t)}$ is the projector onto the symmetric subspace of $\caH^{\otimes t}$, 
and $\kappa_t$ is the dimension of the symmetric subspace.

An ensemble $\mathcal E$ of pure states in $\caH$ is said to be a state $t$-design if and only of
its $t$th moment agrees with that of the Haar distribution, i.e., 
\begin{align}
\underset{|\psi\rangle \sim \mathcal E}{\E} \left[ (|\psi\rangle\langle\psi|)^{\otimes t} \right] 
= \underset{|\psi\rangle \sim \mu_{\rm H}}{\E} \left[ (|\psi\rangle\langle\psi|)^{\otimes t} \right]
= \frac{1}{\kappa_t }\Pi_{\mathrm{sym}}^{(t)}.  
\end{align}
If $\mathcal E$ is a state $t$-design, then it is also a state $(t-1)$-design by definition.

Let $\mathrm{S}_t$ be the permutation group of order $t$. 
For any element $\pi\in\mathrm{S}_t$, we define its corresponding permutation operator 
to be the unitary operator $T_\pi$ on $\caH^{\otimes t}$ that satisfies 
\begin{align}
T_\pi \,\left( \left|\varphi_1\right\rangle \otimes \cdots \otimes\left|\varphi_t\right\rangle\right) 
=\left|\varphi_{\pi^{-1}(1)}\right\rangle \otimes \cdots \otimes\left|\varphi_{\pi^{-1}(t)}\right\rangle  
\end{align}
for $\left|\varphi_1\right\rangle, \left|\varphi_2\right\rangle, \ldots,\left|\varphi_t\right\rangle \in \caH$. 
For example, for $\pi=(1,2)\in\mathrm{S}_2$, the corresponding permutation operator $T_{(1,2)}$ is the swap operator. 
An important property of the projector $\sym^{(t)}$ is that it can be decomposed into a sum of permutation operators: 
\begin{align}\label{eq:PiSymDecompose}
\Pi_{\mathrm{sym}}^{(t)}=\frac{1}{t !} \sum_{\pi \in \mathrm{S}_t} T_\pi. 
\end{align}

\subsection{Tensor-network diagrams}\label{sec:Tensor}
Tensor networks provide a graphical method that simplifies the analysis of tensor operations. 
In tensor-network diagrams, a linear operator $A\in\mathcal{L}(\caH)$ on the space $\caH$ is represented as a box with open legs:   
\begin{align}
A=\ 
\begin{tikzpicture}[baseline={([yshift=-.5ex]current bounding box.center)},inner sep=-4mm]
\node[tensor_green] (A) at (1*\xratio, 0) {$A$};
\draw[-,draw=black] (A) -- (1*\xratio, -0.6);
\draw[-,draw=black] (A) -- (1*\xratio, 0.6);
\end{tikzpicture}\ , 
\end{align}
where the top leg stands for row index of $A$, and the bottom leg stands for column index of $A$. 
The trace of $A$ and the transpose of $A$ are represented as 
\begin{align}
\tr{A}=\ 
\begin{tikzpicture}[baseline={([yshift=-.5ex]current bounding box.center)},inner sep=-4mm]
	\node[tensor_green] (A) at (1*\xratio, 0) {$A$};
	\draw[-<<,draw=black] (A) -- (1*\xratio, -0.7);
	\draw[-<<,draw=black] (A) -- (1*\xratio, 0.7);
\end{tikzpicture}
\ =\ 
\begin{tikzpicture}[baseline={([yshift=-.5ex]current bounding box.center)},inner sep=-4mm]
	\node[tensor_green] (A) at (1*\xratio, 0) {$A$};
	\draw[-,draw=black] (A.north) .. controls (1*\xratio, 0.8) and (1.7*\xratio, 0.8) ..  (1.7*\xratio, 0);
	\draw[-,draw=black] (A.south) .. controls (1*\xratio,-0.8) and (1.7*\xratio,-0.8) ..  (1.7*\xratio,-0);
\end{tikzpicture}
\quad\text{and}\quad
A^{\top} =\ 
\begin{tikzpicture}[baseline={([yshift=-.5ex]current bounding box.center)},inner sep=-4mm]
	\node[tensor_green] (A) at (0*\xratio, 0) {$A$};
	\draw[-,draw=black] (A.north) .. controls (0*\xratio, 0.6) and (0.5*\xratio, 0.7) ..  (0.5*\xratio, 0.1);
	\draw[-,draw=black] (0.5*\xratio, 0.1) .. controls (0.5*\xratio, -0.7) and (0*\xratio, -0.1) ..  (0*\xratio, -1);
	\draw[-,draw=black] (A.south) .. controls (0*\xratio, -0.6) and (-0.5*\xratio, -0.7) ..  (-0.5*\xratio, -0.1);
	\draw[-,draw=black] (-0.5*\xratio, -0.1) .. controls (-0.5*\xratio, 0.7) and (0*\xratio, 0.1) ..  (0*\xratio, 1);
\end{tikzpicture}\ , 
\end{align}
respectively, where  
“\,$\begin{tikzpicture}[baseline={([yshift=-.5ex]current bounding box.center)},inner sep=-4mm]
	\draw[-<<,draw=black] (1*\xratio, -0.3)-- (1*\xratio, 0);
\end{tikzpicture}$\,”
and 
“\,$\begin{tikzpicture}[baseline={([yshift=-.5ex]current bounding box.center)},inner sep=-4mm]
	\draw[-<<,draw=black] (1*\xratio, 0.3)-- (1*\xratio, 0);
\end{tikzpicture}$\,”
means the periodic boundary condition for the top and bottom legs. 
In particular, the identity operator on $\caH$ and its trace are represented as: 
\begin{align}
I=\ 
\begin{tikzpicture}[baseline={([yshift=-.5ex]current bounding box.center)},inner sep=-4mm]
	\draw[-,draw=black] (1*\xratio, -0.6) -- (1*\xratio, 0.6);
\end{tikzpicture}
\ \quad\text{and}\quad
\Tr(I)=\  
\begin{tikzpicture}[baseline={([yshift=-.5ex]current bounding box.center)},inner sep=-4mm]
	\draw[-,draw=black] (1*\xratio, 0) .. controls (1*\xratio, 0.8) and (1.7*\xratio, 0.8) ..  (1.7*\xratio, 0);
	\draw[-,draw=black] (1*\xratio, 0) .. controls (1*\xratio,-0.8) and (1.7*\xratio,-0.8) ..  (1.7*\xratio,-0);
\end{tikzpicture}
\ =\dim(\caH). 
\end{align}

Given operators $A,B\in\mathcal{L}(\caH)$, their product $AB$ is represented as a connection (a contraction of the indices) between two boxes, and their 
tensor product $A\otimes B$ is represented as putting $A$ and $B$ together with no leg connection: 
\begin{align}
AB=\ 
\begin{tikzpicture}[baseline={([yshift=-.5ex]current bounding box.center)},inner sep=-4mm]
	\node[tensor_green] (A) at (1*\xratio, 0) {$A$};
	\node[tensor_green] (B) at (1*\xratio, -0.9) {$B$};
	\draw[-,draw=black] (A) -- (B);
	\draw[-,draw=black] (A) -- (1*\xratio, 0.6);
	\draw[-,draw=black] (B) -- (1*\xratio, -1.5);
\end{tikzpicture}\ ,
\qquad
A\otimes B =\ 
\begin{tikzpicture}[baseline={([yshift=-.5ex]current bounding box.center)},inner sep=-4mm]
\node[tensor_green] (A) at (1*\xratio, 0) {$A$};
\draw[-,draw=black] (A) -- (1*\xratio, -0.7);
\draw[-,draw=black] (A) -- (1*\xratio, 0.7);
\node[tensor_green] (B) at (1.7*\xratio, 0) {$B$};
\draw[-,draw=black] (B) -- (1.7*\xratio, -0.7);
\draw[-,draw=black] (B) -- (1.7*\xratio, 0.7);
\end{tikzpicture}\ .
\end{align}

An operator $A\in\mathcal{L}(\caH^{\otimes 2})$ is represented as a box that has two legs on the top and two on the bottom, and 
its partial transpose is represented as changing the directions of legs on one side: 
\begin{align}
A =\ 
\begin{tikzpicture}[baseline={([yshift=-.5ex]current bounding box.center)},inner sep=-4mm]
\node[2party] (rho) at (0*\xratio, 0) {$A$};
\draw[-,draw=black] (-0.3*\xratio, 0.3) -- (-0.3*\xratio, 0.7);
\draw[-,draw=black] ( 0.3*\xratio, 0.3) -- ( 0.3*\xratio, 0.7);
\draw[-,draw=black] (-0.3*\xratio,-0.3) -- (-0.3*\xratio,-0.7);
\draw[-,draw=black] ( 0.3*\xratio,-0.3) -- ( 0.3*\xratio,-0.7);
\end{tikzpicture}\ ,
\qquad
A^{\top_2} =\ 
\begin{tikzpicture}[baseline={([yshift=-.5ex]current bounding box.center)},inner sep=-4mm]
\node[2party] (rho) at (0*\xratio, 0) {$A$};
\draw[-,draw=black] (-0.3*\xratio, 0.3) -- (-0.3*\xratio, 1);
\draw[-,draw=black] (-0.3*\xratio,-0.3) -- (-0.3*\xratio,-1);
\draw[-,draw=black] (0.3*\xratio, 0.3) .. controls (0.3*\xratio, 0.5) and (0.7*\xratio, 0.6) ..  (0.7*\xratio, 0.1);
\draw[-,draw=black] (0.7*\xratio, 0.1) .. controls (0.7*\xratio, -0.7) and (0.3*\xratio, -0.2) ..  (0.3*\xratio, -1);
\draw[-,draw=black] (0.3*\xratio,-0.3) .. controls (0.3*\xratio, -0.6) and (0.8*\xratio, -0.7) ..  (0.8*\xratio, 0);
\draw[-,draw=black] (0.8*\xratio, 0) .. controls (0.8*\xratio, 1) and (0.3*\xratio, 0.2) ..  (0.3*\xratio, 1);
\end{tikzpicture}\ . 
\end{align}

In tensor-network diagrams, permutation operators are represented as changing the position of legs. 
For example, the swap operator is represented as exchanging two legs, 
and cyclic permutation operators are represented as sequentially moving each leg to its neighboring position:   
\begin{align}
T_{(1,2)}=\ 
\begin{tikzpicture}[baseline={([yshift=-.5ex]current bounding box.center)},inner sep=-4mm]
\draw[-,draw=black] (-0.35*\xratio, 0.6) .. controls (-0.35*\xratio, -0.1) and (0.35*\xratio, 0.1) ..  (0.35*\xratio, -0.6);
\draw[-,draw=black] (0.35*\xratio, 0.6) .. controls (0.35*\xratio, -0.1) and (-0.35*\xratio,  0.1) ..  (-0.35*\xratio, -0.6);
\end{tikzpicture}\ ,
\qquad
T_{(1,3,2)}=\ 
\begin{tikzpicture}[baseline={([yshift=-.5ex]current bounding box.center)},inner sep=-4mm]
\draw[-,draw=black] (0.0*\xratio, 0.6) .. controls (0.0*\xratio, -0.1) and (0.6*\xratio, 0.1) ..  (0.6*\xratio, -0.6);
\draw[-,draw=black] (0.6*\xratio, 0.6) .. controls (0.6*\xratio, -0.1) and (1.2*\xratio, 0.1) ..  (1.2*\xratio, -0.6);
\draw[-,draw=black] (1.2*\xratio, 0.6) .. controls (1.2*\xratio, -0.2) and (0.0*\xratio, 0.2) ..  (0.0*\xratio, -0.6);
\end{tikzpicture}\ .
\end{align}

A particularly important diagram is that of the unnormalized maximally entangled state:  
\begin{align}
|\Phi\>\<\Phi|=\ 
\begin{tikzpicture}[baseline={([yshift=-.5ex]current bounding box.center)},inner sep=-4mm]
\draw[-,draw=black] (-0.4*\xratio, 0.8) .. controls (-0.4*\xratio,-0.16) and (0.4*\xratio,-0.16) ..  (0.4*\xratio, 0.8);
\draw[-,draw=black] (-0.4*\xratio,-0.8) .. controls (-0.4*\xratio, 0.16) and (0.4*\xratio, 0.16) ..  (0.4*\xratio,-0.8);
\end{tikzpicture}\ , 
\end{align}
from which we can easily deduce that $\left( |\Phi\>\<\Phi|\right)^{\top_2}=T_{(1,2)}$.

\section{Proof of \lref{lem:first_moment}}\label{sec:Proof12Moment}
The expectation of $\hat{\phi} \otimes \hat{\psi}$ (over both random state preparation and measurement outcomes) is  
\begin{align}\label{eq:proof1moment1}
&\E\big[ \hat{\phi} \otimes \hat{\psi}\big]  
 = \sum_{i=1}^L \sum_{j=1}^K (\phi_j \otimes \psi_i) \Pr\left\{ \hat{\psi} = \psi_i,\hat{\phi} = \phi_j\right\} 
\nonumber\\
&= \sum_{i=1}^L \sum_{j=1}^K (\phi_j \otimes \psi_i) \frac{1}{L} \Tr\left[ A_{\phi_j} (U\psi_iU^{\dag})^{\otimes s}\right]  
\nonumber\\
&= \frac{1}{L} \sum_{i=1}^L \Bigg( \frac{\kappa_s}{K}  \sum_{j=1}^K \phi_j  \Tr\! \Big[ \phi_j^{\otimes s} (U\psi_iU^{\dag})^{\otimes s}\Big] \Bigg)  \otimes \psi_i. 
\end{align}
According to \lref{lem:GrierLemma13} below and the fact that $\{|\phi_j\>\}_{j=1}^K$ forms a state $(s+1)$-design (see Sec.~\ref{sec:LearnPhase}), the term in the parentheses in the last line can be rewritten as 
\begin{align}\label{eq:proof1moment2}
\frac{\kappa_s}{K} \sum_{j=1}^K \phi_j  \Tr\!\Big[ \phi_j^{\otimes s} (U\psi_iU^{\dag})^{\otimes s}\Big]
=
\frac{I+s U\psi_i U^{\dag}}{d+s}.  
\end{align}
So we have 
\begin{align}
&\E\big[ \hat{\phi} \otimes \hat{\psi}\big]  
 = \frac{1}{L} \sum_{i=1}^L  \, \left( \frac{I+s U\psi_i U^{\dag}}{d+s}\right)   \otimes \psi_i
\nonumber\\
&= \frac{I}{d+s} \otimes \left( \frac{1}{L} \sum_{i=1}^L \psi_i\right)  
\nonumber\\
&\qquad + \frac{s}{d+s}  (U\otimes I) \left( \frac{1}{L} \sum_{i=1}^L \psi_i\otimes\psi_i\right) (U^{\dag}\otimes I)
\nonumber\\
&\stackrel{(a)}{=} \frac{I}{d+s} \otimes \frac{I}{\kappa_1}  + \frac{s}{d+s}(U\otimes I) \frac{\Pi_{\mathrm{sym}}^{(2)}}{\kappa_2} (U^{\dag}\otimes I)
\nonumber\\
&\stackrel{(b)}{=} \frac{I\otimes I}{\kappa_1(d+s)} +\frac{s}{2\kappa_2(d+s)} \sum_{\pi \in \symm_{2}}(U\otimes I)T_{\pi}(U^{\dag}\otimes I) 
\nonumber\\
&= \frac{(d+1+s)(I\otimes I)+s (U\otimes I)T_{(1,2)}(U^{\dag}\otimes I)}{d(d+1)(d+s)} 
\nonumber\\
&\stackrel{(c)}{=} \frac{(d+1+s)(I\otimes I)+s (U\otimes U^{\dag})T_{(1,2)}}{d(d+1)(d+s)}, 
\end{align}
which confirms \lref{lem:first_moment} in the main text. 
Here $(a)$ follows from \eref{eq:statetdesign} in the main text and the fact that $\{|\psi_i\>\}_{i=1}^L$ is a state 2-design (see Sec.~\ref{sec:LearnPhase}), 
$(b)$ follows from \eref{eq:PiSymDecompose}, and $(c)$ holds because $T_{(1,2)}(U^{\dag}\otimes I)=(I\otimes U^{\dag})T_{(1,2)}$.

In this proof, we use at most the $(s+1)$-design propertity of the ensemble $\{|\phi_j\>\}_{j=1}^K$ and the 2-design property of the ensemble $\{|\psi_i\>\}_{i=1}^L$. 
The $(s+2)$-design propertity of $\{|\phi_j\>\}_{j=1}^K$ and the $4$-design propertity of $\{|\psi_i\>\}_{i=1}^L$ will be used later in the proof of \lref{lem:second_moment}.

The following auxiliary lemma follows directly from Lemma~13 in Ref.~\cite{Grier22}. 

\begin{lemma}
\label{lem:GrierLemma13}
Suppose integer $s\geq 1$,  pure state $|\varphi\>\in\caH$, and $\{|\phi_j\>\}_{j=1}^K$ forms a state $(s+1)$-design on $\caH$. Then
\begin{align}
\frac{\kappa_s}{K} \sum_j \phi_j \Tr \left(\phi_j^{\otimes s} \varphi^{\otimes s}\right) 
=\frac{I+s \varphi}{d+s} .
\end{align}
\end{lemma}

\section{Auxiliary lemmas}\label{sec:UsefulLemma}
To establish our results in the main text, here we prepare several auxiliary lemmas. 

Recall that our algorithm $\mathcal{A}_{\text {learn}}$ working in the learning phase outputs a classical description of the prepared state $\hat{\psi}$ and the measurement outcome $\hat{\phi}$ in each round (see Sec.~\ref{sec:LearnPhase}). 
In the main text, we have clarified the expectation of their tensor product $\hat{\phi}\otimes\hat{\psi}$ in \lref{lem:first_moment}. 
Here, we calculate the expectation of $\hat{\phi}^{\otimes 2} \otimes \hat{\psi}^{\otimes 2}$, which is useful for proving \lref{lem:all_the_covariances} in the main text (see Appendix~\ref{sec:covariances}).

\begin{widetext}
\begin{lemma}
\label{lem:second_moment}
The expectation of $\hat{\phi}^{\otimes 2} \otimes \hat{\psi}^{\otimes 2}$ 
(over both random state preparation and measurement outcomes) is
\begin{align}
\E\left[ \hat{\phi}^{\otimes 2} \otimes \hat{\psi}^{\otimes 2}\right]  = \frac{2}{(d+s)(d+s+1)} \left( \frac{1}{\kappa_2} \Delta_1+\frac{s}{\kappa_3}\Delta_2+\frac{s}{\kappa_3}\Delta_3+\frac{s(s-1)}{2\kappa_4} \Delta_4 \right) , 
\end{align}
where
\begin{align}\label{eq:Tj}
\Delta_1&:=\sym^{(2)}\otimes \sym^{(2)}, \\
\Delta_2&:=(I\otimes U\otimes I \otimes I)\left[ I_1\otimes \left( \sym^{(3)}\right)_{2,3,4}\right](I\otimes U^{\dag}\otimes I \otimes I)\left( \sym^{(2)}\otimes I \otimes I\right) , \\
\Delta_3&:=(U\otimes I\otimes I \otimes I)\left[ I_2\otimes \left( \sym^{(3)}\right)_{1,3,4}\right] (U^{\dag}\otimes I\otimes I \otimes I)\left( \sym^{(2)}\otimes I \otimes I\right) , \\
\Delta_4&:=
(U\otimes U\otimes I \otimes I)\, \sym^{(4)}\,  (U^{\dag}\otimes U^{\dag}\otimes I \otimes I), 
\end{align} 
$I_i$ denotes the identity operator acting on the $i$th system, and 
$(\cdot)_{i,j,k}$ denotes the operator acting on the $i$th, $j$th, and $k$th systems.
\end{lemma}

\begin{proof}[Proof of \lref{lem:second_moment}]
We have 
\begin{align}\label{eq:proof2moment1}
\E\left[ \hat{\phi}^{\otimes 2} \otimes \hat{\psi}^{\otimes 2}\right] 
&=
\sum_{i=1}^L \sum_{j=1}^K \left( \phi_j^{\otimes 2} \otimes \psi_i^{\otimes 2}\right)  \Pr\left\{ \hat{\psi} = \psi_i,\hat{\phi} = \phi_j\right\}
\nonumber\\ 
&= \sum_{i=1}^L \sum_{j=1}^K \left( \phi_j^{\otimes 2} \otimes \psi_i^{\otimes 2}\right) \frac{1}{L} \Tr\left[ A_{\phi_j} (U\psi_iU^{\dag})^{\otimes s}\right]  
\nonumber\\
&= \frac{1}{L} \sum_{i=1}^L \Bigg( \frac{\kappa_s}{K}  \sum_{j=1}^K \phi_j^{\otimes 2} \Tr\! \Big[ \phi_j^{\otimes s} (U\psi_iU^{\dag})^{\otimes s}\Big] \Bigg)  \otimes \psi_i^{\otimes 2} . 
\end{align}
According to \lref{lem:GrierLemma14} below and the fact that $\{|\phi_j\>\}_{j=1}^K$ forms a state $(s+2)$-design (see Sec.~\ref{sec:LearnPhase}), 
\begin{align}\label{eq:proof2moment2}
\frac{\kappa_s}{K}  \sum_{j=1}^K \phi_j^{\otimes 2} \Tr\! \Big[ \phi_j^{\otimes s} (U\psi_iU^{\dag})^{\otimes s}\Big] 
= 
\frac{2}{(d+s)(d+s+1)} \left[ \left( I + s U\psi_i U^{\dag}\right) ^{\otimes 2} - \frac{s(s+1)}{2} \left( U\psi_i U^{\dag}\right) ^{\otimes 2}\right] \sym^{(2)} . 
\end{align}
So we have 
\begin{align}\label{eq:proof2moment3}
&\E\left[ \hat{\phi}^{\otimes 2} \otimes \hat{\psi}^{\otimes 2}\right] 
=\frac{2}{(d+s)(d+s+1)} \cdot \frac{1}{L} \sum_{i=1}^L \left[ \left( I + s U\psi_i U^{\dag}\right) ^{\otimes 2} - \frac{s(s+1)}{2} \left( U\psi_i U^{\dag}\right)^{\otimes 2}\right] \sym^{(2)} \otimes \psi_i^{\otimes 2}
\nonumber\\
&= \frac{2}{(d+s)(d+s+1)} \cdot \frac{1}{L} \sum_{i=1}^L \left[ I^{\otimes 2} +s I\otimes \left( U\psi_i U^{\dag}\right) + s \left( U\psi_i U^{\dag}\right)\otimes I + \frac{s(s-1)}{2}\left( U\psi_i U^{\dag}\right)^{\otimes 2}\right] \sym^{(2)} \otimes \psi_i^{\otimes 2}.  
\end{align}
Since $\{|\psi_i\>\}_{i=1}^L$ is a state 4-design (see Sec.~\ref{sec:LearnPhase}), we can apply \eref{eq:statetdesign} in the main text to calculate the sum in \eref{eq:proof2moment3} term-by-term. 
The first term is calculated as 
\begin{align}
\frac{1}{L} \sum_{i=1}^L \left( I^{\otimes 2} \sym^{(2)}\right)  \otimes \psi_i^{\otimes 2} 
&= \frac{1}{\kappa_2}  \sym^{(2)}\otimes \sym^{(2)}
 = \frac{\Delta_1}{\kappa_2},  
\end{align}
where we have used the 2-design property of $\{|\psi_i\>\}_{i=1}^L$. 
The second term is calculated as 
\begin{align}
\frac{1}{L} \sum_{i=1}^L \left[ s I\otimes \left( U\psi_i U^{\dag}\right) \right]  \sym^{(2)} \otimes \psi_i^{\otimes 2}
&= \frac{s}{L} \sum_{i=1}^L (I\otimes U\otimes I \otimes I)\left( I\otimes \psi_i^{\otimes 3}\right) (I\otimes U^{\dag}\otimes I \otimes I)\left( \sym^{(2)}\otimes I \otimes I\right)  
\nonumber\\
&= s (I\otimes U\otimes I \otimes I)\bigg( I\otimes \frac{\sym^{(3)}}{\kappa_3}\bigg) (I\otimes U^{\dag}\otimes I \otimes I)\left( \sym^{(2)}\otimes I \otimes I\right)  
= \frac{s}{\kappa_3}\Delta_2,
\end{align}
where we have used the 3-design property of $\{|\psi_i\>\}_{i=1}^L$. 
By a similar reasoning, the third term is given by 
\begin{align}
\frac{1}{L} \sum_{i=1}^L \left[ s\left( U\psi_i U^{\dag}\right)\otimes I \right]  \sym^{(2)} \otimes \psi_i^{\otimes 2}
=\frac{s}{\kappa_3}\Delta_3,
\end{align}
The last term is calculated as 
\begin{align}\label{eq:proof2moment4}
&\frac{1}{L} \sum_{i=1}^L \left[ \frac{s(s-1)}{2} \left( U\psi_i U^{\dag}\right) ^{\otimes 2} \,\sym^{(2)}\right]  \otimes \psi_i^{\otimes 2} 
\stackrel{(a)}{=}
\frac{s(s-1)}{2} \cdot \frac{1}{L} \sum_{i=1}^L \left( U\psi_i U^{\dag}\right) ^{\otimes 2} \otimes \psi_i^{\otimes 2} 
\nonumber \\
& = \frac{s(s-1)}{2} \cdot \frac{1}{L} \sum_{i=1}^L (U\otimes U\otimes I \otimes I)\left( \psi_i^{\otimes 4}\right) (U^{\dag}\otimes U^{\dag}\otimes I \otimes I)
\nonumber \\
&  \stackrel{(b)}{=} \frac{s(s-1)}{2} (U\otimes U\otimes I \otimes I)\, \frac{\sym^{(4)}}{\kappa_4}\,  (U^{\dag}\otimes U^{\dag}\otimes I \otimes I)
= \frac{s(s-1)}{2\kappa_4} \Delta_4,  
\end{align}
where $(a)$ holds because $(U\psi_i U^{\dag})^{\otimes 2} \,\sym^{(2)}=(U\psi_i U^{\dag})^{\otimes 2}$, and $(b)$
holds because $\{|\psi_i\>\}_{i=1}^L$ is a 4-design. 
Equations~\eqref{eq:proof2moment3}--\eqref{eq:proof2moment4} together confirm \lref{lem:second_moment}. 
\end{proof}
\end{widetext}

Note that, in the above proof, the 4-design property of $\{|\psi_i\>\}_{i=1}^L$ is only utilized in the case of $s\geq2$. 
When $s=1$, the last term $\frac{s(s-1)}{2}\left( U\psi_i U^{\dag}\right)^{\otimes 2}$ in the bracket of \eref{eq:proof2moment3} vanishes and thus the 3-design property of $\{|\psi_i\>\}_{i=1}^L$ is sufficient to prove \lref{lem:second_moment}.

The following auxiliary lemma follows directly from Lemma~14 in Ref.~\cite{Grier22}. 

\begin{lemma}
\label{lem:GrierLemma14}
Suppose integer $s\geq 1$,  pure state $|\varphi\>\in\caH$, and $\{|\phi_j\>\}_{j=1}^K$ forms a state $(s+2)$-design on $\caH$. Then
\begin{align}
&\frac{\kappa_s}{K} \sum_j \phi_j^{\otimes 2} \Tr \left(\phi_j^{\otimes s} \varphi^{\otimes s}\right) 
\nonumber\\
&=
\frac{2}{(d+s)(d+s+1)} \left[ \left( I + s \varphi\right) ^{\otimes 2} - \frac{s(s+1)}{2} \varphi ^{\otimes 2} \right] \sym^{(2)}.
\end{align}
\end{lemma}

\begin{lemma}[H\"{o}lder's inequality]\label{lem:Holder}
Suppose $A,B$ are operators on $\caH$. Then 
\begin{align}
\left| \tr{AB} \right| \leq \|A\|_1 \|B\|, 
\quad 
\left| \tr{AB} \right| \leq \|A\|_2 \|B\|_2. 
\end{align}
\end{lemma}

\begin{lemma}\label{lem:UsefulIneqs}
Suppose $U$ is a unitary operator on $\caH$, $\rho\in\caD(\caH)$, and $O\in\obs(B)$.
Let $O_U:=U^{\dag}OU$ and $\wp:=\Tr(\rho^2)$. Then 
\begin{align}
\left| \tr O \right| &\leq \sqrt{d B}. \label{eq:UsefulIneq1}
\\
\left|\Tr(O_U\rho)\right| &\leq \min\lbrace \sqrt{B\wp},1 \rbrace, \label{eq:UsefulIneq2}
\\
\left|\Tr(O_U^2\rho)\right| &\leq \min\lbrace \sqrt{B\wp},1 \rbrace,  \label{eq:UsefulIneq3}
\\
\left|\Tr(O_U\rho^2)\right| &\leq \wp, \label{eq:UsefulIneq4}
\\
\left|\Tr(O_U^2\rho^2)\right| &\leq \wp, \label{eq:UsefulIneq5}
\\
\big|\!\Tr(O_U\rho\, O_U\rho)\big| &\leq \wp.   \label{eq:UsefulIneq6}
\end{align}
\end{lemma}

\begin{proof}[Proof of Lemma~\ref{lem:UsefulIneqs}]
Equation \eqref{eq:UsefulIneq1} can be derived as 
\begin{align}
	\left| \tr O \right| &\leq \left\| I \right\|_2 \left\| O \right\|_2 =  \sqrt{d \Tr(O^2)} \leq \sqrt{d B}, 
\end{align}
where the first inequality follows from \lref{lem:Holder}, and the second inequality holds because $O\in\obs(B)$. 

Similarly, \eref{eq:UsefulIneq2} can be derived as 
\begin{align}
\left|\Tr(O_U\rho)\right| 
&\leq \left\| O_U \right\|_2 \| \rho\|_2 =\left\|O\right\|_2 \| \rho\|_2 
\nonumber\\
&= \sqrt{\Tr(O^2)\tr{\rho^2}} 
\leq \sqrt{B\wp},
\\
\left|\Tr(O_U\rho)\right| 
&\leq \left\| O_U \right\| \| \rho\|_1 =\left\|O\right\| \tr\rho \leq1, 
\end{align}
given that $\left\| O \right\|\leq 1$ when $O\in\obs(B)$. 

Equation \eqref{eq:UsefulIneq3} follows from \eref{eq:UsefulIneq2} and the fact that $O^2\in \obs(B)$ when $O\in \obs(B)$.

Equation \eqref{eq:UsefulIneq4} can be derived as 
\begin{align}
\left|\Tr(O_U\rho^2)\right| 
&\leq \left\| O_U \right\| \| \rho^2\|_1 =\left\|O\right\| \tr{\rho^2} \leq \wp. 
\end{align}

Equation \eqref{eq:UsefulIneq5} follows from \eref{eq:UsefulIneq4} and the fact that $O^2\in \obs(B)$ when $O\in \obs(B)$.

Equation \eqref{eq:UsefulIneq6} can be derived as 
\begin{align}
\left|\Tr(O_U\rho O_U\rho)\right| 
&\leq \left\| O_U\rho \right\|_2^2 = \Tr[(O_U\rho)^\dag (O_U\rho) ]
\nonumber\\
&= \Tr(O_U^2\rho^2) 
\leq \wp,
\end{align}
where the last inequality follows from Eq.~\eqref{eq:UsefulIneq5}. 
\end{proof}

\section{Predicting \texorpdfstring{$\Tr\left( O\,\mathcal{U}(\rho)\right)$}{} by averaging \texorpdfstring{$\Tr[(O\otimes\rho^{\top}) \hat{X}_i]$}{}}\label{sec:Trivialmean}
Recall that our protocol for CSEU 
outputs the classical shadow data $\US(\caU,m)=\lbrace\hat X_1,X_2,\dots,\hat X_m\rbrace$ in the learning phase,
which contains $m$ independent classical snapshots. 
Given any quantum state $\rho\in \caD(\caH)$ and observable $O\in\obs(B)$ in the prediction phase, 
as mentioned in Sec.~\ref{sec:PredPhase}  of the main text, 
a simple way to predict $\Tr\left( O\,\mathcal{U}(\rho)\right)$ is 
to calculate $\Tr[(O\otimes\rho^{\top}) \hat{X}_i]$ and take their average: 
\begin{align}\label{eq:tildeZ}
\tilde{Z}_m(O,\rho) := \frac{1}{m} \sum_{i=1}^m \Tr\left[ (O\otimes\rho^{\top}) \hat{X}_i\right] .   
\end{align} 
In this appendix we shall show that, with this method, 
the number of queries required to estimate $\Tr\left( O\,\mathcal{U}(\rho)\right)$ within additive error $\epsilon$ 
scales as $\mathcal O(d^2 B/\epsilon^2)$ even when collective measurements on multiple systems are used.  If only single-copy measurements are accessible, 
then the query complexity further increases, reaching up to $\mathcal{O}(d^3 B/\epsilon^2)$.

The fluctuations of $\tilde{Z}_m(O,\rho)$ around its expectation $\Tr\left( O\,\mathcal{U}(\rho)\right)$ are controlled by the 
variance. By \lref{lem:Varhat{Z}} below,  we have 
\begin{align}
\Var\left[ \tilde{Z}_m(O,\rho) \right] 
&= \frac{1}{m}\Var\left[ \Tr\left( (O\otimes\rho^{\top}) \hat{X}\right) \right] 
\nonumber\\
&\leq \frac{1}{m}  \,\Theta \left[ \left( 1+\frac{d^2}{s^2} \right) dB \right] . 
\end{align} 
where the bound can be saturated when $\rho$ is a pure state and $O=\sqrt{B/d}\,I\in\obs(B)$.

By Chebyshev's inequality, in order to suppress the additive error of $\tilde{Z}_m(O,\rho)$ to $\epsilon$ with 
some constant success probability $1-\delta$, it suffices to choose 
$m$ such that $\Var\!\big[ \tilde{Z}_m(O,\rho) \big]\leq \delta\epsilon^2$; that is, to choose 
\begin{align}
	m\geq C \left(  1+\frac{d^2}{s^2} \right)  dB \epsilon^{-2}
\end{align}
 for some constant $C>0$. 
So the total number of required queries reads 
\begin{align}\label{eq:smAverage}
	s\cdot m=C \left( 1+\frac{d^2}{s^2} \right) \frac{sdB}{\epsilon^2} . 
\end{align}
In the special case $s=1$, i.e., only single-copy measurements are used, we have $s\cdot m=\mathcal O(d^3B/\epsilon^2)$. 
This query complexity agrees with that of the protocol for CSEC developed in Ref.~\cite{Kunjummen23}, 
which is ancilla-free and does not use quantum memory (see Sec.~\ref{sec:PriorWork} in the main text). 
If collective measurements on multiple systems are accessible, the query number can be further reduced by increasing $s$. 
When $s=\Theta(d)$, it can be reduced to the greatest extent, in which case the query complexity reads $s\cdot m=\mathcal O(d^2 B/\epsilon^2)$.

In contrast to the above method, our method for estimating $\Tr\left( O\,\mathcal{U}(\rho)\right)$ in the main text is based on the quadratic estimator $d^{-1}\Tr[(O\otimes\rho^{\top}) \hat{X}_i\hat{X}_j]$ instead of $\Tr[(O\otimes\rho^{\top}) \hat{X}_i]$. 
As shown in Table~\ref{tab:DirectMean} below, 
in both single-copy and collective measurement scenarios, the quadratic estimator demonstrates strictly better query efficiency.

\begin{table*}[htbp]
\caption{\label{tab:DirectMean}
The query complexity of two different methods for estimating $\Tr\left( O\,\mathcal{U}(\rho)\right)$. 
}
\begin{math} 
\begin{array}{c|cc}
\hline\hline 
\mbox{Method}
& \text{Single-copy measurements with $s=1$\ }
& \text{\ Collective measurements with $s=\Theta(d)$}
\\[0.6ex]
\hline
\\[-2.8ex]
\text{Averaging $\Tr\!\big[(O\otimes\rho^{\top}) \hat{X}_i\big]$}
&\bigo{ d^3 B \epsilon^{-2} }   
&\bigo{ d^2 B \epsilon^{-2} } 
\\[1.1ex]
\hline
\\[-2.4ex]
\text{Averaging $d^{-1}\Tr\!\big[(O\otimes\rho^{\top}) \hat{X}_i\hat{X}_j\big]$}
&\bigo{ d\,\epsilon^{-2} + d^2\sqrt{B}\epsilon^{-1}} 
\ \text{[see \eref{eq:singleUB}]}
&\bigo{ d\,\big(\epsilon^{-2} + \sqrt{B}\epsilon^{-1}\big) } \ \text{[see \eref{eq:jointUB}]}
\\[1.5ex]
\hline\hline
\end{array}	
\end{math}
\end{table*}

\begin{lemma}\label{lem:Varhat{Z}}
Suppose $\rho\in \caD(\caH)$ and $O\in\obs(B)$. Then 
\begin{align}\label{eq:Varhat{Z}}
\Var\left[ \Tr\left( (O\otimes\rho^{\top}) \hat{X}\right) \right] 
\leq 
\Theta \left[ \left( 1+\frac{d^2}{s^2} \right) dB \right] ,  
\end{align} 
where the inequality can be saturated when $\rho$ is a pure state and $O=\sqrt{d^{-1}B}\,I$.
\end{lemma}

\begin{widetext}
Note that when the observable $O$ is proportional to the identity, i.e., $O=cI$, the single-shot estimator 
\begin{align}
\Tr\left[ (O\otimes\rho^{\top}) \hat{X}\right]
&=
\Tr\left[ (cI\otimes\rho^{\top}) \frac{d(d+1)(d+s) (\hat{\phi}\otimes\hat{\psi}^{\top}) - (d+1+s)(I\otimes I)}{s} \right]
\nonumber\\
&= \frac{c\,d(d+1)(d+s)}{s} \Tr(\rho\hat{\psi}) - \frac{c\,d(d+1+s)}{s} 
\end{align}
is not a constant, and has fluctuation due to the randomness of $\hat{\psi}$. 
Hence, the variance of $\Tr\left[ (O\otimes\rho^{\top}) \hat{X}\right]$ does not vanish in this case.

\begin{proof}[Proof of \lref{lem:Varhat{Z}}]
Using the definition of $\hat{X}$ in \eref{eq:defhatX} of the main text, we can expand the variance  as 
\begin{align}
\Var\left[ \Tr\left( (O\otimes\rho^{\top}) \hat{X}\right) \right] 
&= \Var \left[ \frac{d(d+1)(d+s)}{s} \Tr\left[(O\otimes\rho^{\top})(\hat{\phi}\otimes\hat{\psi}^{\top}) \right] - \frac{d+1+s}{s}\Tr(O\otimes\rho^{\top}) \right] 
\nonumber\\
&= \left[ \frac{d(d+1)(d+s)}{s}\right]^2 \Var\left[  \Tr\left((O\otimes\rho)(\hat{\phi}\otimes\hat{\psi})\right)\right] 
\nonumber\\
&= \left[ \frac{d(d+1)(d+s)}{s}\right]^2 
\left\{ \E\left[ \left(\Tr[(O\otimes\rho)(\hat{\phi}\otimes\hat{\psi})]\right)^2 \right] - \left(\E\left[ \Tr((O\otimes\rho)(\hat{\phi}\otimes\hat{\psi})) \right]\right)^2 \right\}
\nonumber\\
&\leq \left[ \frac{d(d+1)(d+s)}{s}\right]^2 
\Tr\Big[ \left( O^{\otimes 2}\otimes\rho^{\otimes 2}\right) \, \E\left[ \hat{\phi}^{\otimes 2} \otimes \hat{\psi}^{\otimes 2}\right]\Big].
\label{eq:naiveVariance1}
\end{align} 
According to \lref{lem:second_moment}, the trace term in the last line of Eq.~\eqref{eq:naiveVariance1} can be expanded as 
\begin{align}\label{eq:naiveVariance11}
\Tr\!\Big[ \left( O^{\otimes 2}\otimes\rho^{\otimes 2}\right) \, \E\left[ \hat{\phi}^{\otimes 2} \otimes \hat{\psi}^{\otimes 2}\right]\Big]
=
\frac{2}{(d+s)(d+s+1)}
\Tr\left[  \left( O^{\otimes 2}\otimes\rho^{\otimes 2}\right) \left( \frac{\Delta_1}{\kappa_2} +\frac{s\Delta_2}{\kappa_3}+\frac{s\Delta_3}{\kappa_3}+\frac{s(s-1)\Delta_4 }{2\kappa_4} \right) \right] .
\end{align} 

Next, we calculate $\Tr\left[\left( O^{\otimes 2}\otimes\rho^{\otimes 2}\right) \Delta_j \right]$ for $j=1,2,3,4$, respectively. 
First, 
\begin{align}
	\Tr\left[  \left( O^{\otimes 2}\otimes\rho^{\otimes 2}\right)\Delta_1\right] 
	&=
	\Tr\left( O^{\otimes 2}\sym^{(2)}\right) \Tr\left( \rho^{\otimes 2}\sym^{(2)}\right)
	=
	\frac{1}{4}\Tr\left[ O^{\otimes 2}\left(I+T_{(12)} \right) \right]\Tr\left[ \rho^{\otimes 2}\left(I+T_{(12)} \right) \right]
	\nonumber\\
	&=
	\frac{1}{4}\left[ \Tr(O^2)+(\tr{O})^2\right] \left[ \Tr(\rho^2)+1\right]
	= \bigo{Bd},
\end{align} 
where the last equality follows from \lref{lem:UsefulIneqs} in Appendix~\ref{sec:UsefulLemma}.
Second, $\Tr\left[\left( O^{\otimes 2}\otimes\rho^{\otimes 2}\right) \Delta_2 \right]$ can be calculated via tensor-network diagrams:  
\begin{align}
\!\!\!\!\!\!\!\!\!\!
\Tr\left[  \left( O^{\otimes 2}\otimes\rho^{\otimes 2}\right) \Delta_2 \right] 
&=\Tr\left[  \left( O^{\otimes 2}\otimes\rho^{\otimes 2}\right)(I\otimes U\otimes I \otimes I)
             \left( I\otimes  \sym^{(3)}\right)(I\otimes U^{\dag}\otimes I \otimes I)
             \left( \sym^{(2)}\otimes I \otimes I\right)\right] 
\nonumber\\[0.5ex]
&=\  \begin{tikzpicture}[baseline={([yshift=-.5ex]current bounding box.center)},inner sep=-4mm]
\node[tensor_blue] (O1) at (1*\xratio, 0) {$O$};
\node[tensor_blue] (O2) at (2*\xratio, 0) {$O$};
\node[tensor_blue] (rho1) at (3*\xratio, 0) {$\rho$};
\node[tensor_blue] (rho2) at (4*\xratio, 0) {$\rho$};
\node[tensor_green] (U1) at (2*\xratio, -1) {$U$};
\node[Permute_3] (permute_a) at (3*\xratio, -2) {$\sym^{(3)}$};
\node[tensor_green] (U2) at (2*\xratio, -3) {$U^{\dag}$};
\node[Permute_2] (permute_b) at (1.5*\xratio, -4) {$\sym^{(2)}$};
\draw[-,draw=black] (O1) -- (1*\xratio, -3.7);
\draw[-,draw=black] (O2) -- (U1);
\draw[-,draw=black] (U1) -- (2*\xratio, -1.7);
\draw[-,draw=black] (2*\xratio, -2.3) -- (U2);
\draw[-,draw=black] (U2) -- (2*\xratio, -3.7);
\draw[-,draw=black] (rho1) -- (3*\xratio, -1.7);
\draw[-,draw=black] (rho2) -- (4*\xratio, -1.7);
\draw[-<<,draw=black] (O1) -- (1*\xratio, +0.7);
\draw[-<<,draw=black] (O2) -- (2*\xratio, +0.7);
\draw[-<<,draw=black] (rho1) -- (3*\xratio, +0.7);
\draw[-<<,draw=black] (rho2) -- (4*\xratio, +0.7);
\draw[-<<,draw=black] (1*\xratio, -4.3) -- (1*\xratio, -4.7);
\draw[-<<,draw=black] (2*\xratio, -4.3) -- (2*\xratio, -4.7);
\draw[-<<,draw=black] (3*\xratio, -2.3) -- (3*\xratio, -4.7);
\draw[-<<,draw=black] (4*\xratio, -2.3) -- (4*\xratio, -4.7);
\end{tikzpicture}
\ = \frac{1}{3! \cdot 2!} \sum_{\pi_1\in\mathrm{S}_3} \sum_{\pi_2\in\mathrm{S}_2} 
\begin{tikzpicture}[baseline={([yshift=-.5ex]current bounding box.center)},inner sep=-4mm]
\node[tensor_blue] (O1) at (1*\xratio, 0) {$O$};
\node[tensor_blue] (O2) at (2*\xratio, 0) {$O$};
\node[tensor_blue] (rho1) at (3*\xratio, 0) {$\rho$};
\node[tensor_blue] (rho2) at (4*\xratio, 0) {$\rho$};
\node[tensor_green] (U1) at (2*\xratio, -1) {$U$};
\node[Permute_3] (permute_a) at (3*\xratio, -2) {$T_{\pi_1}$};
\node[tensor_green] (U2) at (2*\xratio, -3) {$U^{\dag}$};
\node[Permute_2] (permute_b) at (1.5*\xratio, -4) {$T_{\pi_2}$};
\draw[-,draw=black] (O1) -- (1*\xratio, -3.7);
\draw[-,draw=black] (O2) -- (U1);
\draw[-,draw=black] (U1) -- (2*\xratio, -1.7);
\draw[-,draw=black] (2*\xratio, -2.3) -- (U2);
\draw[-,draw=black] (U2) -- (2*\xratio, -3.7);
\draw[-,draw=black] (rho1) -- (3*\xratio, -1.7);
\draw[-,draw=black] (rho2) -- (4*\xratio, -1.7);
\draw[-<<,draw=black] (O1) -- (1*\xratio, +0.7);
\draw[-<<,draw=black] (O2) -- (2*\xratio, +0.7);
\draw[-<<,draw=black] (rho1) -- (3*\xratio, +0.7);
\draw[-<<,draw=black] (rho2) -- (4*\xratio, +0.7);
\draw[-<<,draw=black] (1*\xratio, -4.3) -- (1*\xratio, -4.7);
\draw[-<<,draw=black] (2*\xratio, -4.3) -- (2*\xratio, -4.7);
\draw[-<<,draw=black] (3*\xratio, -2.3) -- (3*\xratio, -4.7);
\draw[-<<,draw=black] (4*\xratio, -2.3) -- (4*\xratio, -4.7);
\end{tikzpicture}
\nonumber\\[1ex]
&= \frac{1}{12} \tr O \left[\tr O  + \tr O \Tr(\rho^2) + 2 \Tr(OU\rho U^{\dag}) + 2\Tr(OU\rho^2 U^{\dag})\right] \nonumber\\
&\qquad +\frac{1}{12} \left[ \Tr(O^2) + \Tr(O^2) \Tr(\rho^2) + 2 \Tr(O^2 U\rho U^{\dag}) + 2\Tr(O^2 U\rho^2 U^{\dag}) \right] 
\nonumber\\
&= \bigo{Bd} ,  
\end{align}
where we use \lref{lem:UsefulIneqs} to derive the last equality. 
Third, from the symmetry between $\Delta_2$ and $\Delta_3$ we have 
\begin{align}
	\Tr\left[\left( O^{\otimes 2}\otimes\rho^{\otimes 2}\right) \Delta_3 \right]
	= \Tr\left[\left( O^{\otimes 2}\otimes\rho^{\otimes 2}\right) \Delta_2 \right]= \bigo{Bd}.
\end{align}
Forth, 
\begin{align}
	\Tr\left[  \left( O^{\otimes 2}\otimes\rho^{\otimes 2}\right)\Delta_4\right] 
	&=
	\Tr\left[(U^{\dag}\otimes U^{\dag}\otimes I\otimes I)
	\left(O^{\otimes 2}\otimes\rho^{\otimes2}\right)(U\otimes U\otimes I\otimes I)\, \sym^{(4)}\right] 
	\nonumber\\
	&=
	\frac{1}{4!} \sum_{\pi\in\mathrm{S}_4} \Tr\left[  \left( O_U^{\otimes 2}\otimes\rho^{\otimes 2}\right) T_\pi \right] 
	\nonumber\\
	&=
	\frac{1}{24}\big[ (\Tr(O_U))^2 + \Tr(O_U^2) + 4\Tr(O_U)\Tr(O_U\rho) + (\Tr(O_U))^2\Tr(\rho^2) + 4\Tr(O_U)\Tr(O_U\rho^2)
	\nonumber\\
	&\qquad\quad  + 4\Tr(O_U^2\rho) + \Tr(O_U^2)\Tr(\rho^2) + 2(\Tr(O_U\rho))^2 + 4\Tr(O_U^2\rho^2) + 2\Tr(O_U\rho O_U\rho)\big]
	\nonumber\\
	&= \bigo{Bd} ,  
\end{align}
where $O_U=U^{\dag}OU$, and the last equality follows from \lref{lem:UsefulIneqs} again.

Therefore, Eq.~\eqref{eq:naiveVariance11} can be rewritten as 
\begin{align}\label{eq:naiveVariance4}
\Tr\!\Big[ \left( O^{\otimes 2}\otimes\rho^{\otimes 2}\right) \, \E\left[ \hat{\phi}^{\otimes 2} \otimes \hat{\psi}^{\otimes 2}\right]\Big]
= 
\frac{2}{(d+s)(d+s+1)}\left(\frac{1}{\kappa_2} +\frac{2s}{\kappa_3} + \frac{s(s-1)}{2\kappa_4} \right) \mathcal{O}(dB). 
\end{align} 
By plugging Eq.~\eqref{eq:naiveVariance4} into Eq.~\eqref{eq:naiveVariance1}, we have 
\begin{align}
\Var\left[ \Tr\left( (O\otimes\rho^{\top}) \hat{X}\right) \right] 
\leq 
\left[ \frac{d(d+1)(d+s)}{s}\right]^2  \frac{\mathcal{O}(dB)}{(d+s)(d+s+1)} 
\left(\frac{1}{\kappa_2} +\frac{2s}{\kappa_3} + \frac{s(s-1)}{2\kappa_4} \right) 
=\mathcal O \left[ \left( 1+\frac{d^2}{s^2} \right) dB \right], 
\end{align}
which confirms the inequality \eqref{eq:Varhat{Z}}.

At last, we show that the inequality \eqref{eq:Varhat{Z}} can be saturated when $\rho$ is a pure state and $O=\sqrt{B/d}\,I\in\obs(B)$. 
In this case,
\begin{align}
	\Var\left[ \Tr\left( (O\otimes\rho^{\top}) \hat{X}\right) \right] 
	&=  
	\Var\left[ \frac{d(d+1)(d+s)}{s} \Tr\left[ (O\otimes\rho^{\top})(\hat{\phi}\otimes\hat{\psi}^{\top})  \right]  
	-\frac{d+1+s}{s} \Tr(O\otimes\rho^{\top}) \right]
	\nonumber\\
	&=  
	\Var\left[\frac{d(d+1)(d+s)}{s} \sqrt{\frac{B}{d}} \Tr(\rho\hat{\psi}) \right]
	\nonumber\\
	&=  
	\frac{B}{d}  \left[\frac{d(d+1)(d+s)}{s} \right]^2  \Big[ \Tr\left((\rho\otimes\rho) \E[ \hat{\psi}\otimes\hat{\psi} ] \right) - \left(\Tr(\rho\E[\hat{\psi}])\right)^2 \Big]
	\nonumber\\
	&\!\stackrel{(a)}{=}  
	\frac{B}{d}  \left[\frac{d(d+1)(d+s)}{s} \right]^2  \bigg[ \Tr\bigg( (\rho\otimes\rho) \frac{\sym^{(2)}}{\kappa_2} \bigg)  
	- \big( \Tr\left( \rho I/\kappa_1  \right) \big)^2 \bigg]
	\nonumber\\
	&=  
	\frac{B}{d}  \left[\frac{d(d+1)(d+s)}{s} \right]^2  \bigg[ \frac{2}{d(d+1)}  - \frac{1}{d^2}\bigg]
	\nonumber\\
	&=  
	\Theta \left[ \left( 1+\frac{d^2}{s^2} \right) dB \right], 
\end{align}
where $(a)$ follows from \eref{eq:statetdesign} in the main text and the fact that 
the state $\hat{\psi}$ is chosen randomly from a 4-design state  ensemble.  
This completes the proof of \lref{lem:Varhat{Z}}. 
\end{proof}

\section{Proof of \lref{lem:all_the_covariances}}\label{sec:covariances}
In preparation for proving \lref{lem:all_the_covariances}, 
it will be convenient to first calculate the second moment of the snapshot $\hat{X}$, which is defined in \eref{eq:defhatX} of the main text.    

\begin{lemma}\label{lem:ExpXTotimesXT}
The expectation of $\hat{X} \otimes \hat{X}$ (over both random state preparation and measurement outcomes) is
\begin{align}
\E \big[ \hat{X} \otimes \hat{X} \big]
=\frac{2 d^2(d+1)^2(d+s)}{s^2(d+s+1)}\left( \frac{1}{\kappa_2} \tilde \Delta_1+\frac{s}{\kappa_3}\tilde \Delta_2
+\frac{s}{\kappa_3}\tilde  \Delta_3+\frac{s(s-1)}{2\kappa_4} \tilde \Delta_4 \right) 
-\frac{d+1+s}{s} \left( \tilde \Delta_5 + \tilde \Delta_6\right) 
-\frac{(d+1+s)^2}{s^2} \tilde \Delta_7, 
\end{align} 
where
\begin{align}
\tilde \Delta_1&:=\left( \sym^{(2)}\right)_{1,3} \otimes \left( \sym^{(2)}\right)_{2,4} , \\
\tilde \Delta_2&:=\left( (I\otimes I\otimes U \otimes I)\left[ I_1\otimes \left( \sym^{(3)}\right)_{2,3,4}\right](I\otimes I \otimes U^{\dag}\otimes I)
          \left[ \left( \sym^{(2)}\right)_{1,3} \otimes I_2 \otimes I_4\right] \right)^{\top_{2,4}} , \\
\tilde \Delta_3&:=\left( (U\otimes I\otimes I \otimes I)\left[ I_3\otimes \left( \sym^{(3)}\right)_{1,2,4}\right] (U^{\dag}\otimes I\otimes I \otimes I)
          \left[ \left( \sym^{(2)}\right)_{1,3} \otimes I_2 \otimes I_4\right] \right)^{\top_{2,4}} , \\
\tilde \Delta_4&:=\left( (U\otimes I\otimes U \otimes I)\, \sym^{(4)}\, (U^{\dag}\otimes I\otimes U^{\dag}\otimes I)
          \right)^{\top_{2,4}} , \\
\tilde \Delta_5&:= \Upsilon_\caU\otimes I\otimes I, \qquad
\tilde \Delta_6 := I\otimes I\otimes \Upsilon_\caU, \qquad
\tilde \Delta_7 := I\otimes I \otimes I\otimes I,
\end{align}
and $\top_{2,4}$ denotes the partial transpose on the second and forth subsystems. 
\end{lemma}

\begin{proof}[Proof of \lref{lem:ExpXTotimesXT}]
We have 
\begin{align}
\E \big[ \hat{X} \otimes \hat{X} \big] 
&= \E \left[ \left( \frac{d(d+1)(d+s)(\hat{\phi}\otimes\hat{\psi}^{\top})-(d+1+s)(I\otimes I)}{s} \right) ^{\otimes2}\right] \nonumber\\
&= \frac{d^2(d+1)^2(d+s)^2}{s^2} \E\left[ \hat{\phi}\otimes\hat{\psi}\otimes\hat{\phi}\otimes\hat{\psi}\right]^{\top_{2,4}} 
-\frac{d(d+1)(d+s)(d+1+s)}{s^2} \E\left[ \hat{\phi}\otimes\hat{\psi}\otimes I\otimes I\right]^{\top_2}  \nonumber\\
&\qquad -\frac{d(d+1)(d+s)(d+1+s)}{s^2} \E\left[ I\otimes I\otimes\hat{\phi}\otimes\hat{\psi} \right]^{\top_4}
+\frac{(d+1+s)^2}{s^2} \left( I\otimes I\otimes I\otimes I\right) \nonumber\\
&\stackrel{(a)}{=} \frac{2 d^2(d+1)^2(d+s)}{s^2(d+s+1)}\left( \frac{1}{\kappa_2} \tilde \Delta_1+\frac{s}{\kappa_3}\tilde \Delta_2
+\frac{s}{\kappa_3}\tilde  \Delta_3+\frac{s(s-1)}{2\kappa_4} \tilde \Delta_4 \right)
+\frac{(d+1+s)^2}{s^2} \left( I\otimes I\otimes I\otimes I\right) \nonumber\\
&\qquad -\frac{d(d+1)(d+s)(d+1+s)}{s^2} \left[ \frac{(d+1+s)(I\otimes I)+s \Upsilon_\caU}{d(d+1)(d+s)} \otimes I\otimes I \right] \nonumber\\
&\qquad -\frac{d(d+1)(d+s)(d+1+s)}{s^2} \left[ I\otimes I \otimes \frac{(d+1+s)(I\otimes I)+s \Upsilon_\caU}{d(d+1)(d+s)} \right] \nonumber\\
&= \frac{2 d^2(d+1)^2(d+s)}{s^2(d+s+1)}\left( \frac{1}{\kappa_2} \tilde \Delta_1+\frac{s}{\kappa_3}\tilde \Delta_2
+\frac{s}{\kappa_3}\tilde  \Delta_3+\frac{s(s-1)}{2\kappa_4} \tilde \Delta_4 \right) \nonumber\\
&\qquad -\frac{d+1+s}{s} \left( I\otimes I\otimes \Upsilon_\caU +  \Upsilon_\caU\otimes I\otimes I \right) 
-\frac{(d+1+s)^2}{s^2} \left( I\otimes I\otimes I\otimes I\right), 
\end{align} 
which confirms \lref{lem:ExpXTotimesXT}. 
Here the equality $(a)$ follows from \lref{lem:first_moment} in the main text, 
\lref{lem:second_moment} in Appendix~\ref{sec:UsefulLemma}, and the 
relation $\left[ (U\otimes U^{\dag})T_{(1,2)}\right]^{\top_2}= \Upsilon_\caU$.
\end{proof}

Now we are ready to bound the covariance term 
$\operatorname{Cov}\left( \Lambda_{i,j},\Lambda_{k,\ell}  \right)$ 
for each combination of the four indices $(i,j,k,\ell)$ with $i \neq j$ and $k \neq \ell$. 
In the rest of this appendix, Appendix~\ref{sec:CalCov1} deals with the case in which exactly one index matches in different 
positions $(i = \ell$ or $j = k)$; 
Appendix~\ref{sec:CalCov2} deals with the case in which exactly one index matches in the same position $(i = k$ or $j = \ell)$; 
Appendix~\ref{sec:CalCov3} deals with the case in which both indices match with order swapped $(i = \ell$ and $j = k)$; 
and Appendix~\ref{sec:CalCov4} deals with the case in which both indices match with the same order $(i = k$ and $j = \ell)$. 
Lemmas~\ref{lem:CalCov1}--\ref{lem:CalCov4} in the following subsections together will confirm \lref{lem:all_the_covariances} in the main text.

\subsection{Exactly one index matches in different positions}\label{sec:CalCov1}
\begin{lemma}\label{lem:CalCov1}
For all distinct $i,j,k$, we have 	
$\operatorname{Cov}\left( \Lambda_{i,j},\Lambda_{j,k}  \right) = \bigo{d^2 \min\{1, B\wp\}}$, where $\wp=\Tr(\rho^2)$. 
\end{lemma}

\begin{proof}[Proof of Lemma \ref{lem:CalCov1}]
First, we have 
\begin{align}\label{eq:Cov1Aa}
\operatorname{Cov}\left( \Lambda_{i,j},\Lambda_{j,k}  \right)
&= \E \left[ \Tr\left( \left( O\otimes \rho^{\top}\right) \hat{X}_i \hat{X}_j \right)
\Tr\left( \left( O\otimes \rho^{\top}\right) \hat{X}_j \hat{X}_k \right)^* \right] 
- \E \left[ \Tr\left( \left( O\otimes \rho^{\top}\right) \hat{X}_i \hat{X}_j \right) \right]
  \E \left[ \Tr\left( \left( O\otimes \rho^{\top}\right) \hat{X}_j \hat{X}_k \right)^* \right]
\nonumber\\
&\!\stackrel{(a)}{=} 
\E \left[ \Tr\left( \left( O\otimes \rho^{\top}\right) \hat{X}_i \hat{X}_j \right)
           \Tr\left( \left( O\otimes \rho^{\top}\right) \hat{X}_k \hat{X}_j \right) \right] 
   - \Tr\left[ \left( O\otimes \rho^{\top}\right) d\,\Upsilon_\caU \right] 
     \Tr\left[ \left( O\otimes \rho^{\top}\right) d\,\Upsilon_\caU \right]^*
\nonumber\\
&= \E \left[ \Tr \left( \left( O\otimes \rho^{\top}\otimes O\otimes \rho^{\top}\right) 
                 \left( \hat{X}_i \hat{X}_j\otimes \hat{X}_k \hat{X}_j \right) \right) \right] 
       - d^2 \left[ \Tr\left( O\,\mathcal{U}(\rho)\right) \right]^2, 
\end{align} 
where $(a)$ holds because $O\otimes\rho^{\top}$, $\hat{X}_j$, $\hat{X}_k$ are Hermitian, 
and 
\begin{align}
\E[\hat{X}_i \hat{X}_j]=\E[\hat{X}_j \hat{X}_k]=\E[\hat{X}]^2= \Upsilon_\caU^2=d\Upsilon_\caU.
\end{align} 
The second moment term in the last line of \eref{eq:Cov1Aa} can be further calculated as
\begin{align}
\E \left[ \Tr \left( \left( O\otimes \rho^{\top}\otimes O\otimes \rho^{\top}\right) 
\left( \hat{X}_i \hat{X}_j\otimes \hat{X}_k \hat{X}_j \right) \right) \right] 
&= \Tr \left( \left( O\otimes \rho^{\top}\otimes O\otimes \rho^{\top}\right) 
\left( \E \big[ \hat{X}_i\big]  \otimes \E \big[\hat{X}_k \big] \right) \E \big[ \hat{X}_j \otimes \hat{X}_j \big] \right) 
\nonumber\\
&= \Tr \left( \left( O\otimes \rho^{\top}\otimes O\otimes \rho^{\top}\right) 
\left( \Upsilon_\caU \otimes \Upsilon_\caU \right) \, \E \big[ \hat{X} \otimes \hat{X} \big] \right). 
\label{eq:Cov1A}
\end{align}

The expression of $\E \big[ \hat{X} \otimes \hat{X} \big]$ in \lref{lem:ExpXTotimesXT} has 7 terms. 
Plugging them into Eq.~\eqref{eq:Cov1A}, respectively, we get
\begin{align}
&\Tr \left[ \left( O\otimes \rho^{\top}\otimes O\otimes \rho^{\top}\right) \left( \Upsilon_\caU \otimes \Upsilon_\caU \right) \tilde \Delta_5 \right] 
= d \left[ \Tr\left( O\,\mathcal{U}(\rho)\right) \right]^2 \geq 0, 
\label{eq:Cov1bound1B}
\\
&\Tr \left[ \left( O\otimes \rho^{\top}\otimes O\otimes \rho^{\top}\right) \left( \Upsilon_\caU \otimes \Upsilon_\caU \right) \tilde \Delta_6 \right] 
= d \left[ \Tr\left( O\,\mathcal{U}(\rho)\right) \right]^2 \geq 0,
\label{eq:Cov1bound2B} 
\\
&\Tr \left[ \left( O\otimes \rho^{\top}\otimes O\otimes \rho^{\top}\right) \left( \Upsilon_\caU \otimes \Upsilon_\caU \right) \tilde \Delta_7 \right] 
= \left[ \Tr\left( O\,\mathcal{U}(\rho)\right) \right]^2 \geq 0 ;
\label{eq:Cov1bound3B}
\end{align}
and 
\begin{align}
\Tr \left[ \left( O\otimes \rho^{\top}\otimes O\otimes \rho^{\top}\right) \left( \Upsilon_\caU \otimes \Upsilon_\caU \right) \tilde \Delta_1 \right] 
&= \bigo{\min\{1, B\wp\}}, 
\label{eq:Cov1bound1}
\\
\Tr \left[ \left( O\otimes \rho^{\top}\otimes O\otimes \rho^{\top}\right) \left( \Upsilon_\caU \otimes \Upsilon_\caU \right) \tilde \Delta_2 \right] 
&= \bigo{ d \min\{1, B\wp\}}, 	
\label{eq:Cov1bound2}
\\
\Tr \left[ \left( O\otimes \rho^{\top}\otimes O\otimes \rho^{\top}\right) \left( \Upsilon_\caU \otimes \Upsilon_\caU \right) \tilde \Delta_3 \right] 
&= \bigo{ d \min\{1, B\wp\}},
\label{eq:Cov1bound3} 	
\\
\Tr \left[ \left( O\otimes \rho^{\top}\otimes O\otimes \rho^{\top}\right) \left( \Upsilon_\caU \otimes \Upsilon_\caU \right) \tilde \Delta_4 \right] 
&= \bigo{ d^2 \min\{1, B\wp\}}.  	
\label{eq:Cov1bound4}
\end{align} 
Here Eqs.~\eqref{eq:Cov1bound1}--\eqref{eq:Cov1bound4} can be derived by using tensor-network diagrams.  
For example, 
\begin{align}
	&\Tr \left[ \left( O\otimes \rho^{\top}\otimes O\otimes \rho^{\top}\right) \left( \Upsilon_\caU \otimes \Upsilon_\caU \right) \tilde \Delta_2 \right]
	\nonumber\\[0.5ex]
	&=\  \begin{tikzpicture}[baseline={([yshift=-.5ex]current bounding box.center)},inner sep=-4mm]
		\node[tensor_blue] (O1) at (1*\xratio, 0) {$O$};
		\node[tensor_blue] (rho1) at (2*\xratio, 0) {$\rho^{\top}$};
		\node[tensor_blue] (O2) at (3*\xratio, 0) {$O$};
		\node[tensor_blue] (rho2) at (4*\xratio, 0) {$\rho^{\top}$};
		\node[tensor_green] (U1) at (1*\xratio, -1) {$U$};
		\node[tensor_green] (U2) at (3*\xratio, -1) {$U$};
		\node[tensor_green] (U3) at (1*\xratio, -2.5) {$U^{\dag}$};
		\node[tensor_green] (U4) at (3*\xratio, -2.5) {$U^{\dag}$};
		\node[tensor_green] (U5) at (3*\xratio, -3.7) {$U$};
		\node[Permute_3] (permute_a) at (3*\xratio, -4.5) {$\sym^{(3)}$};
		\node[tensor_green] (U6) at (3*\xratio, -5.5) {$U^{\dag}$};
		\node[Permute_3] (permute_b) at (2*\xratio, -6.5) {$\sym^{(2)}$};
		\draw[-,draw=black] (O1) -- (U1);
		\draw[-,draw=black] (U1.south) .. controls (1*\xratio, -1.7) and (2*\xratio, -2.1) ..  (rho1.south);
		\draw[-,draw=black] (O2) -- (U2);
		\draw[-,draw=black] (U2.south) .. controls (3*\xratio, -1.7) and (4*\xratio, -2.1) ..  (rho2.south);
		\draw[-,draw=black] (U3.north) .. controls (1*\xratio, -1.8) and (2*\xratio, -1.4) ..  (2*\xratio, -2.8);
		\draw[-,draw=black] (2*\xratio, -2.8) .. controls (2*\xratio, -3.4) and (1.3*\xratio, -3.4) ..  (1.3*\xratio, -4.6);
		\draw[-,draw=black] (1.3*\xratio, -4.6) .. controls (1.3*\xratio, -5.3) and (2*\xratio, -5.3) ..  (2*\xratio, -4.8);
		\draw[-,draw=black] (2*\xratio, -4.2) .. controls (2*\xratio, -3.8) and (1.5*\xratio, -3.8) ..  (1.5*\xratio, -4.4);
		\draw[-,draw=black] (1.5*\xratio, -4.4) .. controls (1.5*\xratio, -5) and (2*\xratio, -5.1) ..  (2*\xratio, -5.5);
		\draw[-<<,draw=black] (2*\xratio, -5.5) -- (2*\xratio, -5.9);
		\draw[-,draw=black] (U4.north) .. controls (3*\xratio, -1.8) and (4*\xratio, -1.4) ..  (4*\xratio, -2.8);
		\draw[-,draw=black] (4*\xratio, -2.8) .. controls (4*\xratio, -3.4) and (4.6*\xratio, -3.4) ..  (4.6*\xratio, -4.6);
		\draw[-,draw=black] (4.6*\xratio, -4.6) .. controls (4.6*\xratio, -5.3) and (4*\xratio, -5.3) ..  (4*\xratio, -4.8);
		\draw[-,draw=black] (4*\xratio, -4.2) .. controls (4*\xratio, -3.8) and (4.5*\xratio, -3.8) ..  (4.5*\xratio, -4.4);
		\draw[-,draw=black] (4.5*\xratio, -4.4) .. controls (4.5*\xratio, -5) and (4*\xratio, -5.1) ..  (4*\xratio, -5.7);
		\draw[-<<,draw=black] (4*\xratio, -5.7) -- (4*\xratio, -7.3);
		\draw[-,draw=black] (U4) -- (U5);
		\draw[-,draw=black] (U5) -- (3*\xratio, -4.2);
		\draw[-,draw=black] (U3) -- (1*\xratio, -6.2);
		\draw[-,draw=black] (3*\xratio, -4.8) -- (U6);
		\draw[-,draw=black] (U6) -- (3*\xratio, -6.2);
		\draw[-<<,draw=black] (O1) -- (1*\xratio, +0.7);
		\draw[-<<,draw=black] (rho1) -- (2*\xratio, +0.7);
		\draw[-<<,draw=black] (O2) -- (3*\xratio, +0.7);
		\draw[-<<,draw=black] (rho2) -- (4*\xratio, +0.7);
		\draw[-<<,draw=black] (1*\xratio, -6.8) -- (1*\xratio, -7.3);
		\draw[-<<,draw=black] (3*\xratio, -6.8) -- (3*\xratio, -7.3);
		\draw[-,densely dotted,thick,draw=red] (0.6*\xratio, -3.0) -- (0.6*\xratio, -0.5);
		\draw[-,densely dotted,thick,draw=red] (0.6*\xratio, -0.5) -- (4.2*\xratio, -0.5);
		\draw[-,densely dotted,thick,draw=red] (4.2*\xratio, -0.5) -- (4.2*\xratio, -3.0);
		\draw[-,densely dotted,thick,draw=red] (4.2*\xratio, -3.0) -- (0.6*\xratio, -3.0);
		\node[red] (Omega) at (4.90*\xratio, -1.6) {$\Upsilon_\caU\otimes\Upsilon_\caU$};
		\draw[-,densely dotted,thick,draw=red] (0.6*\xratio, -3.21) -- (0.6*\xratio, -7.0);
		\draw[-,densely dotted,thick,draw=red] (0.6*\xratio, -7.0) -- (4.7*\xratio, -7.0);
		\draw[-,densely dotted,thick,draw=red] (4.7*\xratio, -7.0) -- (4.7*\xratio, -3.21);
		\draw[-,densely dotted,thick,draw=red] (4.7*\xratio, -3.21) -- (0.6*\xratio, -3.21);
		\node[red] (T) at (4.98*\xratio, -5.2) {$\tilde \Delta_2$};
	\end{tikzpicture}
	\ = \frac{1}{3! \cdot 2!} \sum_{\pi_1\in\mathrm{S}_3} \sum_{\pi_2\in\mathrm{S}_2} \ 
	\begin{tikzpicture}[baseline={([yshift=-.5ex]current bounding box.center)},inner sep=-4mm]
		\node[tensor_blue] (O1) at (2*\xratio, -0.5) {$O$};
		\node[tensor_blue] (rho1) at (2*\xratio, -2.5) {$\rho$};
		\node[tensor_blue] (O2) at (4*\xratio, -0.5) {$O$};
		\node[tensor_blue] (rho2) at (4*\xratio, -2.5) {$\rho$};
		\node[tensor_green] (U1) at (2*\xratio, -1.5) {$U$};
		\node[tensor_green] (U2) at (4*\xratio, -1.5) {$U$};
		\node[tensor_green] (U3) at (2*\xratio, -5) {$U^{\dag}$};
		\node[Permute_3] (permute_a) at (3*\xratio, -4) {$T_{\pi_1}$};
		\node[tensor_green] (U6) at (3*\xratio, -5) {$U^{\dag}$};
		\node[Permute_2] (permute_b) at (2.5*\xratio, -6) {$T_{\pi_2}$};
		\draw[-,draw=black] (O1) -- (U1);
		\draw[-,draw=black] (U1.south) --  (rho1.north);
		\draw[-,draw=black] (O2) -- (U2);
		\draw[-,draw=black] (U2.south) -- (rho2.north);
		\draw[-,draw=black] (U3.north) -- (2*\xratio, -4.3);
		\draw[-,draw=black] (3*\xratio, -3.7) .. controls (3*\xratio, -3.0) and (4.5*\xratio, -3.0) ..  (4.5*\xratio, -3.9);
		\draw[-,draw=black] (4*\xratio, -4.3) .. controls (4*\xratio, -4.8) and (4.5*\xratio, -4.8) ..  (4.5*\xratio, -3.9);
		\draw[-,draw=black] (rho1.south) -- (2*\xratio, -3.7);
		\draw[-,draw=black] (rho2.south) -- (4*\xratio, -3.7);
		\draw[-,draw=black] (U3) -- (2*\xratio, -5.7);
		\draw[-,draw=black] (3*\xratio, -4.3) -- (U6);
		\draw[-,draw=black] (U6) -- (3*\xratio, -5.7);
		\draw[-<<,draw=black] (O1) -- (2*\xratio, +0.2);
		\draw[-<<,draw=black] (O2) -- (4*\xratio, +0.2);
		\draw[-<<,draw=black] (2*\xratio, -6.3) -- (2*\xratio, -6.8);
		\draw[-<<,draw=black] (3*\xratio, -6.3) -- (3*\xratio, -6.8);
	\end{tikzpicture}
	\nonumber\\[1ex]
	&= \frac{1}{12} \left[  4(\Tr(OU\rho U^{\dag}))^2 + 2d(\Tr(OU\rho U^{\dag}))^2 + 4\Tr(OU\rho U^{\dag}OU\rho U^{\dag}) + 2d\Tr(OU\rho U^{\dag}OU\rho U^{\dag}) \right] 
	\nonumber\\ 
	&=\bigo{ d \min\{1, B\wp\}},
\end{align}
where the last equality follows from \lref{lem:UsefulIneqs}.

Therefore, by Eqs.~\eqref{eq:Cov1Aa} and \eqref{eq:Cov1A}, 
\begin{align}
\operatorname{Cov}\left( \Lambda_{i,j},\Lambda_{j,k}  \right)
&\leq \Tr \left( \left( O\otimes \rho^{\top}\otimes O\otimes \rho^{\top}\right) 
\left( \Upsilon_\caU \otimes \Upsilon_\caU \right) \, \E \big[ \hat{X} \otimes \hat{X} \big] \right)
\nonumber\\
&\stackrel{(a)}{\leq}
\Tr \left[ \left( O\otimes \rho^{\top}\otimes O\otimes \rho^{\top}\right) 
\left( \Upsilon_\caU \otimes \Upsilon_\caU \right) \,  
\frac{2 d^2(d+1)^2(d+s)}{s^2(d+s+1)}\left( \frac{1}{\kappa_2} \tilde \Delta_1+\frac{s}{\kappa_3}\tilde \Delta_2
+\frac{s}{\kappa_3}\tilde  \Delta_3+\frac{s(s-1)}{2\kappa_4} \tilde \Delta_4 \right) \right] 
\nonumber\\
& \stackrel{(b)}{=} \mathcal O \left[ \frac{d^4}{s^2} 
\left( \frac{ \min\{1, B\wp\}}{\kappa_2}  +\frac{s d\min\{1, B\wp\}}{\kappa_3} 
+\frac{s(s-1) d^2\min\{1, B\wp\}}{2\kappa_4}  \right)  \right] 
\nonumber\\
&= \mathcal O \left( d^2 \min\{1, B\wp\} \right),  
\end{align} 
where $(a)$ follows from \lref{lem:ExpXTotimesXT} and Eqs.~\eqref{eq:Cov1bound1B}--\eqref{eq:Cov1bound3B}, and 
$(b)$ follows from Eqs.~\eqref{eq:Cov1bound1}--\eqref{eq:Cov1bound4}. 
\end{proof}

\subsection{Exactly one index matches in the same position}\label{sec:CalCov2}

\begin{lemma}\label{lem:CalCov2}
For all distinct $i,j,k$, we have 	
$\operatorname{Cov}\left( \Lambda_{i,j},\Lambda_{k,j}  \right)
= \mathcal O \left( \frac{d^3}{s} \wp + d^2 \min\{1, B\wp\} \right)$,  
where $\wp=\Tr(\rho^2)$.	
\end{lemma}

\begin{proof}[Proof of Lemma \ref{lem:CalCov2}]
First, we have 
\begin{align}
\operatorname{Cov}\left( \Lambda_{i,j},\Lambda_{k,j}  \right)
&= \E \left[ \Tr\left( \left( O\otimes \rho^{\top}\right) \hat{X}_i \hat{X}_j \right)
\Tr\left( \left( O\otimes \rho^{\top}\right) \hat{X}_k \hat{X}_j \right)^* \right] 
- \E \left[ \Tr\left( \left( O\otimes \rho^{\top}\right) \hat{X}_i \hat{X}_j \right) \right]
\E \left[ \Tr\left( \left( O\otimes \rho^{\top}\right) \hat{X}_k \hat{X}_j \right)^* \right]
\nonumber\\
&\!\stackrel{(a)}{=} 
\E \left[ \Tr\left( \left( O\otimes \rho^{\top}\right) \hat{X}_i \hat{X}_j \right)
\Tr\left( \left( O\otimes \rho^{\top}\right) \hat{X}_j \hat{X}_k \right) \right] 
- \Tr\left[ \left( O\otimes \rho^{\top}\right) d\,\Upsilon_\caU \right] 
\Tr\left[ \left( O\otimes \rho^{\top}\right) d\,\Upsilon_\caU \right]^*
\nonumber\\
&= \E \left[ \Tr \left( \left( O\otimes \rho^{\top}\otimes O\otimes \rho^{\top}\right) 
\left(\hat{X}_i\hat{X}_j \otimes \hat{X}_j\hat{X}_k\right) \right)\right] 
- d^2 \left[ \Tr\left( O\,\mathcal{U}(\rho)\right) \right]^2, 
\end{align} 
where $(a)$ holds because $O\otimes\rho^{\top}$, $\hat{X}_j$, $\hat{X}_k$ are Hermitian, 
and $\E[\hat{X}_i \hat{X}_j]=\E[\hat{X}_j \hat{X}_k]=\E[\hat{X}]^2= \Upsilon_\caU^2=d\Upsilon_\caU$. 
The second moment term in the last line can be further calculated as 
\begin{align}
&\E \left[ \Tr \left( \left( O\otimes \rho^{\top}\otimes O\otimes \rho^{\top}\right) 
\left(\hat{X}_i\hat{X}_j \otimes \hat{X}_j\hat{X}_k\right) \right)\right]
\nonumber\\
&= \Tr \left( \left( O\otimes \rho^{\top}\otimes O\otimes \rho^{\top}\right) 
\left( \E \big[ \hat{X}_i\big]  \otimes I\otimes I \right) 
\E \big[ \hat{X}_j \otimes \hat{X}_j \big] 
\left( I\otimes I \otimes \E \big[ \hat{X}_k\big] \right)  \right) 
\nonumber\\
&= \Tr \left( \left( O\otimes \rho^{\top}\otimes O\otimes \rho^{\top}\right) 
\left( \Upsilon_\caU \otimes I\otimes I \right) 
\, \E \big[ \hat{X} \otimes \hat{X} \big] 
\left( I\otimes I \otimes \Upsilon_\caU \right)  \right).\label{eq:Cov1B}
\end{align}

The expression of $\E \big[ \hat{X} \otimes \hat{X} \big]$ in Lemma~\ref{lem:ExpXTotimesXT} has 7 terms. 
Plugging them into Eq.~\eqref{eq:Cov1B}, respectively, we get
\begin{align}
&\Tr \left[ \left( O\otimes \rho^{\top}\otimes O\otimes \rho^{\top}\right) \left( \Upsilon_\caU\otimes I\otimes I \right) \tilde \Delta_5  \left( I\otimes I \otimes \Upsilon_\caU \right)\right]
= d \left[ \Tr\left( O\,\mathcal{U}(\rho)\right) \right]^2 \geq 0, 
\label{eq:Cov2bound1B}
\\
&\Tr \left[ \left( O\otimes \rho^{\top}\otimes O\otimes \rho^{\top}\right) \left( \Upsilon_\caU\otimes I\otimes I \right) \tilde \Delta_6  \left( I\otimes I \otimes \Upsilon_\caU \right)\right]
= d \left[ \Tr\left( O\,\mathcal{U}(\rho)\right) \right]^2 \geq 0, 
\\
&\Tr \left[ \left( O\otimes \rho^{\top}\otimes O\otimes \rho^{\top}\right) \left( \Upsilon_\caU\otimes I\otimes I \right) \tilde \Delta_7 \left( I\otimes I \otimes \Upsilon_\caU \right)\right]
= \left[ \Tr\left( O\,\mathcal{U}(\rho)\right) \right]^2 \geq 0, 
\label{eq:Cov2bound3B}
\end{align} 
and   
\begin{align}
\Tr \left[ \left( O\otimes \rho^{\top}\otimes O\otimes \rho^{\top}\right) \left( \Upsilon_\caU\otimes I\otimes I \right) \tilde \Delta_1 \left( I\otimes I \otimes \Upsilon_\caU \right)\right] &= \bigo{d \wp} ,
\label{eq:Cov2bound1}
\\
\Tr \left[ \left( O\otimes \rho^{\top}\otimes O\otimes \rho^{\top}\right) \left( \Upsilon_\caU\otimes I\otimes I \right) \tilde \Delta_2 \left( I\otimes I \otimes \Upsilon_\caU \right)\right] &= \bigo{ d \min\{1, B\wp\}}, 
\label{eq:Cov2bound2}
\\
\Tr \left[ \left( O\otimes \rho^{\top}\otimes O\otimes \rho^{\top}\right) \left( \Upsilon_\caU\otimes I\otimes I \right) \tilde \Delta_3 \left( I\otimes I \otimes \Upsilon_\caU \right)\right] &= \bigo{d^2 \wp} ,
\label{eq:Cov2bound3}
\\
\Tr \left[ \left( O\otimes \rho^{\top}\otimes O\otimes \rho^{\top}\right) \left( \Upsilon_\caU\otimes I\otimes I \right) \tilde \Delta_4 \left( I\otimes I \otimes \Upsilon_\caU \right)\right] &= \bigo{ d^2 \min\{1, B\wp\}}.
\label{eq:Cov2bound4}
\end{align} 
Here Eqs.~\eqref{eq:Cov2bound1}--\eqref{eq:Cov2bound4} can be derived by using tensor-network diagrams.  
For example, 
\begin{align}
&\!\!\!\!\!\!\!\!\!\!\!\!\!\!\!\!\!\!
\Tr \left[ \left( O\otimes \rho^{\top}\otimes O\otimes \rho^{\top}\right) \left( \Omega\otimes I\otimes I \right) \tilde \Delta_4 \left( I\otimes I \otimes \Omega \right)\right] 
\nonumber\\[0.5ex]
&=\  \begin{tikzpicture}[baseline={([yshift=-.5ex]current bounding box.center)},inner sep=-4mm]
\node[tensor_blue] (O1) at (1*\xratio, -0.3) {$O$};
\node[tensor_blue] (rho1) at (2*\xratio, -0.3) {$\rho^{\top}$};
\node[tensor_blue] (O2) at (3*\xratio, -0.3) {$O$};
\node[tensor_blue] (rho2) at (4*\xratio, -0.3) {$\rho^{\top}$};
\draw[-<<,draw=black] (O1) -- (1*\xratio, +0.4);
\draw[-<<,draw=black] (rho1) -- (2*\xratio, +0.4);
\draw[-<<,draw=black] (O2) -- (3*\xratio, +0.4);
\draw[-<<,draw=black] (rho2) -- (4*\xratio, +0.4);
\node[tensor_green] (U1) at (1*\xratio, -1.2) {$U$};
\node[tensor_green] (U2) at (1*\xratio, -2.5) {$U^{\dag}$};
\node[tensor_green] (U3) at (1*\xratio, -3.7) {$U$};
\node[tensor_green] (U4) at (3*\xratio, -3.7) {$U$};
\node[Permute_4] (permute_a) at (2.5*\xratio, -4.6) {$\sym^{(4)}$};
\node[tensor_green] (U5) at (1*\xratio, -5.5) {$U^{\dag}$};
\node[tensor_green] (U6) at (3*\xratio, -5.5) {$U^{\dag}$};
\node[tensor_green] (U7) at (3*\xratio, -6.7) {$U$};
\node[tensor_green] (U8) at (3*\xratio, -8) {$U^{\dag}$};
\draw[-,draw=black] (U1.south) .. controls (1*\xratio, -1.9) and (2*\xratio, -1.9) ..  (2*\xratio, -1.1);
\draw[-,draw=black] (2*\xratio, -1.1) -- (rho1.south);
\draw[-,draw=black] (U2.north) .. controls (1*\xratio, -1.9) and (2*\xratio, -1.4) ..  (2*\xratio, -2.8); 
\draw[-,draw=black] (O1.south) -- (U1.north);
\draw[-,draw=black] (U2.south) -- (U3.north);
\draw[-,draw=black] (U3.south) -- (1*\xratio, -4.3);
\draw[-,draw=black] (1*\xratio, -4.9) -- (U5.north);
\draw[-,draw=black] (O2.south) -- (U4.north);
\draw[-,draw=black] (U4.south) -- (3*\xratio, -4.3);
\draw[-,draw=black] (3*\xratio, -4.9) -- (U6.north);
\draw[-,draw=black] (U6.south) -- (U7.north); 
\draw[-,draw=black] (4*\xratio, -2.5) .. controls (4*\xratio, -3.4) and (4.6*\xratio, -3.3) ..  (4.6*\xratio, -4.7);
\draw[-,draw=black] (4.6*\xratio, -4.7) .. controls (4.6*\xratio, -5.3) and (4*\xratio, -5.3) ..  (4*\xratio, -4.9);
\draw[-,draw=black] (4*\xratio, -4.3) .. controls (4*\xratio, -3.9) and (4.5*\xratio, -3.8) ..  (4.5*\xratio, -4.6);
\draw[-,draw=black] (4.5*\xratio, -4.6) .. controls (4.5*\xratio, -5.2) and (4*\xratio, -5.1) ..  (4*\xratio, -6.2);
\draw[-,draw=black] (rho2.south) -- (4*\xratio, -2.5);
\draw[-,draw=black] (U7.south).. controls (3*\xratio, -7.4) and (4*\xratio, -7.4) ..  (4*\xratio, -6.7);
\draw[-,draw=black] (4*\xratio, -6.2) -- (4*\xratio, -6.7);
\draw[-,draw=black] (U8.north).. controls (3*\xratio, -7.3) and (4*\xratio, -7.3) ..  (4*\xratio, -8);
\draw[-<<,draw=black] (U8.south) -- (3*\xratio, -8.8);
\draw[-<<,draw=black] (4*\xratio, -8) -- (4*\xratio, -8.8);
\draw[-<<,draw=black] (U5.south) -- (1*\xratio, -8.8);
\draw[-<<,draw=black] (2*\xratio, -4.3) .. controls (2*\xratio, -3.1) and (2.4*\xratio, -3.1) ..  (2.4*\xratio, -4.0);
\draw[-|||,draw=black](2*\xratio, -2.8) .. controls (2*\xratio, -3.1) and (1.6*\xratio, -3.8) ..  (1.6*\xratio, -4.2);
\draw[-|||,draw=black](2*\xratio, -4.9) .. controls (2*\xratio, -5.6) and (1.6*\xratio, -5.6) ..  (1.6*\xratio, -5.0);
\draw[-,densely dotted,thick,draw=red] (0.6*\xratio, -6.0) -- (0.6*\xratio, -3.2);
\draw[-,densely dotted,thick,draw=red] (0.6*\xratio, -3.2) -- (4.7*\xratio, -3.2);
\draw[-,densely dotted,thick,draw=red] (4.7*\xratio, -3.2) -- (4.7*\xratio, -6.0);
\draw[-,densely dotted,thick,draw=red] (4.7*\xratio, -6.0) -- (0.6*\xratio, -6.0);
\node[red] (T) at (4.98*\xratio, -5.3) {$\tilde \Delta_4$};
\draw[-,densely dotted,thick,draw=red] (0.6*\xratio, -3.0) -- (0.6*\xratio, -0.75);
\draw[-,densely dotted,thick,draw=red] (0.6*\xratio, -0.75) -- (2.2*\xratio, -0.75);
\draw[-,densely dotted,thick,draw=red] (2.2*\xratio, -0.75) -- (2.2*\xratio, -3.0);
\draw[-,densely dotted,thick,draw=red] (2.2*\xratio, -3.0) -- (0.6*\xratio, -3.0);
\node[red] (Omega) at (2.52*\xratio, -1.8) {$\Upsilon_\caU$};
\draw[-,densely dotted,thick,draw=red] (2.6*\xratio, -8.5) -- (2.6*\xratio, -6.2);
\draw[-,densely dotted,thick,draw=red] (2.6*\xratio, -6.2) -- (4.2*\xratio, -6.2);
\draw[-,densely dotted,thick,draw=red] (4.2*\xratio, -6.2) -- (4.2*\xratio, -8.5);
\draw[-,densely dotted,thick,draw=red] (4.2*\xratio, -8.5) -- (2.6*\xratio, -8.5);
\node[red] (Omega) at (4.52*\xratio, -7.3) {$\Upsilon_\caU$};
\end{tikzpicture}
\ = \frac{1}{4!} \sum_{\pi\in\mathrm{S}_4} 
\begin{tikzpicture}[baseline={([yshift=-.5ex]current bounding box.center)},inner sep=-4mm]
\draw[-<<,draw=black] (O1) -- (1*\xratio, +0.8);
\node[tensor_blue] (O1) at (1*\xratio, -0) {$O$};
\node[tensor_green] (U1) at (1*\xratio, -0.9) {$U$};
\node[tensor_blue] (rho1) at (1*\xratio, -1.8) {$\rho$};
\node[Permute_4] (permute_a) at (2.5*\xratio, -3.6) {$T_{\pi}$};
\node[tensor_green] (U2) at (1*\xratio, -5.1) {$U^{\dag}$};
\draw[-<<,draw=black] (U2) -- (1*\xratio, -5.9);
\draw[-,draw=black] (O1) -- (U1);
\draw[-,draw=black] (U1) -- (rho1);
\draw[-,draw=black] (rho1.south) .. controls (1*\xratio, -2.7) and (2*\xratio, -2.7) .. (2*\xratio, -3.3);
\draw[-,draw=black] (1*\xratio, -3.9) -- (U2);
\draw[-,draw=black] (0.5*\xratio, -3.7) .. controls (0.5*\xratio, -4.6) and (2*\xratio, -4.6) .. (2*\xratio, -3.9);
\draw[-,draw=black] (0.5*\xratio, -3.7) .. controls (0.5*\xratio, -2.7) and (1*\xratio, -2.7) .. (1*\xratio, -3.3);
\node[tensor_blue] (rho2) at (3*\xratio, -0) {$\rho$};
\node[tensor_green] (U3) at (3*\xratio, -0.9) {$U^{\dag}$};
\node[tensor_blue] (O2) at (3*\xratio, -1.8) {$O$};
\node[tensor_green] (U4) at (3*\xratio, -2.7) {$U$};
\draw[-<<,draw=black]  (rho2) --(3*\xratio, 0.8);
\draw[-,draw=black] (rho2) -- (U3);
\draw[-,draw=black] (U3) -- (O2);
\draw[-,draw=black] (O2) -- (U4);
\draw[-,draw=black] (U4) -- (3*\xratio, -3.3);
\draw[-,draw=black] (4.5*\xratio, -3.7) .. controls (4.5*\xratio, -4.6) and (3*\xratio, -4.6) .. (3*\xratio, -3.9);
\draw[-,draw=black] (4.5*\xratio, -3.7) .. controls (4.5*\xratio, -2.7) and (4*\xratio, -2.7) .. (4*\xratio, -3.3);
\draw[-<<,draw=black] (4*\xratio, -3.9) .. controls (4*\xratio, -5.5) and (3*\xratio, -4.3) .. (3*\xratio, -5.9);
\end{tikzpicture}
\nonumber\\[1ex] 
&= \frac{1}{24} \left[  \left( d^2+5d+6\right)  (\Tr(OU\rho U^{\dag}))^2 + \left( d^2+5d+6\right) \Tr(O^2 U\rho^2 U^{\dag}) \right] 
\nonumber\\ 
&=\bigo{ d^2 \min\{1, B\wp\}},   
\end{align}
where the last inequality follows from \lref{lem:UsefulIneqs}.

Therefore, 
\begin{align}
\operatorname{Cov}\left( \Lambda_{i,j},\Lambda_{k,j}  \right)
&\leq \Tr \left( \left( O\otimes \rho^{\top}\otimes O\otimes \rho^{\top}\right) 
              \left( \Upsilon_\caU \otimes I\otimes I \right) 
                   \, \E \big[ \hat{X} \otimes \hat{X} \big] 
                         \left( I\otimes I \otimes \Upsilon_\caU \right)  \right)
\nonumber\\
&\stackrel{(a)}{\leq} 
\Tr \!\bigg[ \left( O\otimes \rho^{\top}\otimes O\otimes \rho^{\top}\right) 
\left( \Upsilon_\caU \otimes I\otimes I \right) \, 
\nonumber\\
&\qquad\qquad \times \frac{2 d^2(d+1)^2(d+s)}{s^2(d+s+1)} \left( \frac{1}{\kappa_2} \tilde \Delta_1+\frac{s}{\kappa_3}\tilde \Delta_2
+\frac{s}{\kappa_3}\tilde  \Delta_3+\frac{s(s-1)}{2\kappa_4} \tilde \Delta_4 \right) \left( I\otimes I \otimes \Upsilon_\caU \right) \bigg] 
\nonumber\\
& \stackrel{(b)}{=} \mathcal O \left[ \frac{d^4}{s^2} 
\left( \frac{d\wp}{\kappa_2}  +\frac{s d\min\{1, B\wp\}}{\kappa_3} +\frac{s d^2\wp }{\kappa_3} 
+\frac{s(s-1) d^2\min\{1, B\wp\}}{2\kappa_4}  \right)  \right] 
\nonumber\\
&= \mathcal O \left( \frac{d^3}{s}  \wp + d^2 \min\{1, B\wp\} \right), 
\end{align} 
where $(a)$ follows from \lref{lem:ExpXTotimesXT} and Eqs.~\eqref{eq:Cov2bound1B}--\eqref{eq:Cov2bound3B},
and $(b)$ follows from Eqs.~\eqref{eq:Cov2bound1}--\eqref{eq:Cov2bound4}. 
\end{proof}

\subsection{Both indices match with order swapped}\label{sec:CalCov3}
\begin{lemma}\label{lem:CalCov3}
For distinct $i,j$, we have 	
\begin{align}
\operatorname{Cov}\left( \Lambda_{i,j},\Lambda_{j,i}  \right)
&=\mathcal O  \left[ \left( 
\frac{d^3}{s^4}
+ \frac{d^2}{s^3} 
+ \frac{d\sqrt{d\wp}}{s^2} 
+ \frac{\sqrt{d\wp}}{s} 
+ \wp \right) d^2B \right]. 
\end{align} 	
\end{lemma}

\begin{proof}[Proof of \lref{lem:CalCov3}]
For distinct $i,j$, we have 
\begin{align}
\operatorname{Cov}\left( \Lambda_{i,j},\Lambda_{j,i}  \right)
&= \E \left[ \Tr\left( \left( O\otimes \rho^{\top}\right) \hat{X}_i \hat{X}_j \right)
\Tr\left( \left( O\otimes \rho^{\top}\right) \hat{X}_i \hat{X}_j \right) \right] 
- \E \left[ \Tr\left( \left( O\otimes \rho^{\top}\right) \hat{X}_i \hat{X}_j \right) \right]
\E \left[ \Tr\left( \left( O\otimes \rho^{\top}\right) \hat{X}_j \hat{X}_i \right)^* \right]
\nonumber\\
&= \E \left[ \Tr \left( \left( O\otimes \rho^{\top}\otimes O\otimes \rho^{\top}\right) 
\left(\hat{X}_i\hat{X}_j \otimes \hat{X}_i\hat{X}_j\right) \right)\right] 
- d^2 \left[ \Tr\left( O\,\mathcal{U}(\rho)\right) \right]^2, 
\nonumber\\
&\leq \Tr \left( \left( O\otimes \rho^{\top}\otimes O\otimes \rho^{\top}\right) 
\E \big[ \hat{X} \otimes \hat{X} \big] \E \big[ \hat{X} \otimes \hat{X} \big] \right).  
\label{eq:Cov3A}
\end{align} 

Define 
\begin{align}\label{eq:define_xi}
\xi(\mu,\nu):= 
\left| \Tr \left( \left( O\otimes \rho^{\top}\otimes O\otimes \rho^{\top}\right) \tilde \Delta_\mu \, \tilde \Delta_\nu  \right)\right| , 
\qquad \mu,\nu\in\{1,2,\dots,7\}. 
\end{align}
Then, according to \lref{lem:ExpXTotimesXT}, we can expand the last line of Eq.~\eqref{eq:Cov3A} as 
\begin{align} \label{eq:Cov3B}
&\Tr \left( \left( O\otimes \rho^{\top}\otimes O\otimes \rho^{\top}\right) 
\E \big[ \hat{X} \otimes \hat{X} \big] \E \big[ \hat{X} \otimes \hat{X} \big] \right)
\nonumber\\
&= \mathcal O \bigg(       \frac{d^4}{s^4}\xi(1,1) + \frac{d^3}{s^3}[\xi(1,2)+\xi(1,3)+\xi(2,1)+\xi(3,1)] 
                         + \frac{d}{s}[\xi(2,4)+\xi(3,4)+\xi(4,2)+\xi(4,3)]  
\nonumber\\&\qquad\quad  + \frac{d^2}{s^2}[\xi(2,2)+\xi(2,3)+\xi(1,4)+\xi(3,2)+\xi(3,3)+\xi(4,1)] + \xi(4,4)
                         + \frac{(d+s)^4}{s^4}\xi(7,7) 
\nonumber\\&\qquad\quad  + \frac{d(d+s)}{s^2}[\xi(3,5)+\xi(3,6)+\xi(2,5)+\xi(2,6)+\xi(5,2)+\xi(5,3)+\xi(6,2)+\xi(6,3)] 
\nonumber\\&\qquad\quad  + \frac{d^2(d+s)}{s^3}[\xi(1,5)+\xi(1,6)+\xi(5,1)+\xi(6,1)] 
                         + \frac{(d+s)}{s}[\xi(4,5)+\xi(4,6)+\xi(5,4)+\xi(6,4)]  
\nonumber\\&\qquad\quad  + \frac{(d+s)^3}{s^3}[\xi(5,7)+\xi(6,7)+\xi(7,5)+\xi(7,6)] 
                         + \frac{d(d+s)^2}{s^3}[\xi(2,7)+\xi(3,7)+\xi(7,2)+\xi(7,3)]
\nonumber\\&\qquad\quad  + \frac{(d+s)^2}{s^2}[\xi(5,5)+\xi(5,6)+\xi(6,5)+\xi(6,6)+\xi(4,7)+\xi(7,4)] 
                         + \frac{d^2(d+s)^2}{s^4}[\xi(1,7)+\xi(7,1)]
\bigg) .
\end{align}

To bound the covariance, we derive bounds for all $\xi(\mu,\nu)$ using tensor-network diagrams, as summarized in Table~\ref{tab:xiUB}. 
For example, 
\begin{align}\label{eq:xi13}
&\Tr \left[ \left( O\otimes \rho^{\top}\otimes O\otimes \rho^{\top}\right) \tilde \Delta_1 \, \tilde \Delta_3 \right]  
\nonumber\\[0.4ex]
&\qquad=\  \begin{tikzpicture}[baseline={([yshift=-.5ex]current bounding box.center)},inner sep=-4mm]
\node[tensor_blue] (O1) at (1*\xratio, -0.2) {$O$};
\node[tensor_blue] (rho1) at (2*\xratio, -0.2) {$\rho^{\top}$};
\node[tensor_blue] (O2) at (3*\xratio, -0.2) {$O$};
\node[tensor_blue] (rho2) at (4*\xratio, -0.2) {$\rho^{\top}$};
\draw[-<<,draw=black] (O1) -- (1*\xratio, +0.5);
\draw[-<<,draw=black] (rho1) -- (2*\xratio, +0.5);
\draw[-<<,draw=black] (O2) -- (3*\xratio, +0.5);
\draw[-<<,draw=black] (rho2) -- (4*\xratio, +0.5);
\node[Permute_3] (permute_a) at (2*\xratio, -1.3) {$\sym^{(2)}$};
\node[Permute_3] (permute_b) at (3*\xratio, -2.3) {$\sym^{(2)}$};
\node[tensor_green] (U1) at (1*\xratio, -3.6) {$U$};
\node[Permute_4] (permute_c) at (2.5*\xratio, -4.5) {$\sym^{(3)}$};
\node[tensor_green] (U2) at (1*\xratio, -5.4) {$U^{\dag}$};
\node[Permute_3] (permute_d) at (2*\xratio, -6.3) {$\sym^{(2)}$};
\draw[-<<,draw=black] (1*\xratio, -6.6) -- (1*\xratio, -7.1);
\draw[-<<,draw=black] (3*\xratio, -6.6) -- (3*\xratio, -7.1);
\draw[-,draw=black] (O1.south) -- (1*\xratio, -1);
\draw[-,draw=black] (1*\xratio, -1.6) -- (U1.north);
\draw[-,draw=black] (U1.south) -- (1*\xratio, -4.2);
\draw[-,draw=black] (1*\xratio, -4.8) -- (U2.north);
\draw[-,draw=black] (U2.south) -- (1*\xratio, -6.0);
\draw[-|||,draw=black] (rho1.south) -- (2*\xratio, -0.9);
\draw[-|||,draw=black] (2*\xratio, -2) -- (2*\xratio, -1.7);
\draw[-,draw=black] (O2.south) -- (3*\xratio, -1);
\draw[-|||,draw=black] (3*\xratio, -1.6) -- (3*\xratio, -1.9);
\draw[-,draw=black] (rho2.south) -- (4*\xratio, -2);
\draw[|||-|||,draw=black] (3*\xratio, -2.7) -- (3*\xratio, -4.1);
\draw[-|||,draw=black] (3*\xratio, -6.0) -- (3*\xratio, -4.9);
\draw[-,draw=black] (4*\xratio, -2.6) .. controls (4*\xratio, -3.3) and (4.5*\xratio, -3.2) ..  (4.5*\xratio, -4.6);
\draw[-,draw=black] (4.5*\xratio, -4.6) .. controls (4.5*\xratio, -5.2) and (4*\xratio, -5.2) ..  (4*\xratio, -4.8);
\draw[-,draw=black] (4*\xratio, -4.2) .. controls (4*\xratio, -3.8) and (4.4*\xratio, -3.7) ..  (4.4*\xratio, -4.5);
\draw[-<<,draw=black] (4.4*\xratio, -4.5) .. controls (4.4*\xratio, -5.1) and (4*\xratio, -5.0) ..  (4*\xratio, -7.1);
\draw[-<<,draw=black] (2*\xratio, -4.2) .. controls (2*\xratio, -3.1) and (2.4*\xratio, -3.1) ..  (2.4*\xratio, -4.0);
\draw[-|||,draw=black](2*\xratio, -2.6) .. controls (2*\xratio, -3.5) and (1.6*\xratio, -3.4) ..  (1.6*\xratio, -4.1);
\draw[-|||,draw=black](2*\xratio, -4.8) .. controls (2*\xratio, -5.4) and (1.6*\xratio, -5.4) ..  (1.6*\xratio, -4.9);
\draw[-,densely dotted,thick,draw=red] (0.6*\xratio, -6.8) -- (0.6*\xratio, -3.1);
\draw[-,densely dotted,thick,draw=red] (0.6*\xratio, -3.1) -- (4.63*\xratio, -3.1);
\draw[-,densely dotted,thick,draw=red] (4.63*\xratio, -3.1) -- (4.63*\xratio, -6.8);
\draw[-,densely dotted,thick,draw=red] (4.63*\xratio, -6.8) -- (0.6*\xratio, -6.8);
\node[red] (T) at (4.9*\xratio, -5) {$\tilde \Delta_3$};
\draw[-,densely dotted,thick,draw=red] (0.6*\xratio, -2.9) -- (0.6*\xratio, -0.7);
\draw[-,densely dotted,thick,draw=red] (0.6*\xratio, -0.7) -- (4.4*\xratio, -0.7);
\draw[-,densely dotted,thick,draw=red] (4.4*\xratio, -0.7) -- (4.4*\xratio, -2.9);
\draw[-,densely dotted,thick,draw=red] (4.4*\xratio, -2.9) -- (0.6*\xratio, -2.9);
\node[red] (T) at (4.67*\xratio, -1.8) {$\tilde \Delta_1$};
\end{tikzpicture}
\ =\frac{1}{2!\cdot 3!\cdot 2!} \sum_{\substack{\pi_1,\pi_2\in\mathrm{S}_2 \\ \pi_3\in\mathrm{S}_3}} 
\begin{tikzpicture}[baseline={([yshift=-.5ex]current bounding box.center)},inner sep=-4mm]
\node[tensor_blue] (O1) at (1*\xratio, -0.2) {$O$};
\node[tensor_blue] (O2) at (2*\xratio, -0.2) {$O$};
\draw[-<<,draw=black] (O1) -- (1*\xratio, +0.6);
\draw[-<<,draw=black] (O2) -- (2*\xratio, +0.6);
\node[Permute_2] (permute_a) at (1.5*\xratio, -1.2) {$T_{\pi_1}$};
\node[tensor_green] (U1) at (2*\xratio, -2.2) {$U$};
\node[tensor_blue] (rho1) at (3*\xratio, -2.2) {$\rho$};
\node[tensor_blue] (rho2) at (4*\xratio, -2.2) {$\rho$};
\node[Permute_3] (permute_b) at (3*\xratio, -3.2) {$T_{\pi_3}$};
\node[Permute_2] (permute_c) at (5.4*\xratio, -3.2) {$T_{\pi_2}$};
\node[tensor_green] (U2) at (2*\xratio, -4.2) {$U^\dag$};
\draw[-<<,draw=black] (1*\xratio, -1.5) -- (1*\xratio, -5);
\draw[-<<,draw=black] (U2.south) -- (2*\xratio, -5);
\draw[-,draw=black] (O1.south) -- (1*\xratio, -0.9);
\draw[-,draw=black] (O2.south) -- (2*\xratio, -0.9);
\draw[-,draw=black] (U1.north) -- (2*\xratio, -1.5);
\draw[-,draw=black] (U1.south) -- (2*\xratio, -2.9);
\draw[-,draw=black] (rho1.south) -- (3*\xratio, -2.9);
\draw[-,draw=black] (rho2.south) -- (4*\xratio, -2.9);
\draw[-,draw=black] (U2.north) -- (2*\xratio, -3.5);
\draw[-,draw=black] (3*\xratio, -3.5) .. controls (3*\xratio, -5.2) and (5.9*\xratio, -5.2) ..  (5.9*\xratio, -3.5);
\draw[-,draw=black] (4*\xratio, -3.5) .. controls (4*\xratio, -4.2) and (4.9*\xratio, -4.2) ..  (4.9*\xratio, -3.5);
\draw[-,draw=black] (3*\xratio, -1.9) .. controls (3*\xratio, -0.4) and (5.9*\xratio, +0.1) ..  (5.9*\xratio, -2.9);
\draw[-,draw=black] (4*\xratio, -1.9) .. controls (4*\xratio, -1.4) and (4.9*\xratio, -0.9) ..  (4.9*\xratio, -2.9);
\end{tikzpicture}
\nonumber\\[1ex]
&\qquad= \frac{1}{12} \,\Big[ (\Tr(O))^2 (\Tr(\rho))^2 + 2\Tr(O)\Tr(\rho)\Tr(OU\rho U^\dag) + (\Tr(O))^2\Tr(\rho^2) + 2\Tr(O)\Tr(O U\rho^2 U^\dag) \nonumber\\
&\qquad\qquad\quad + \Tr(O^2)(\Tr(\rho))^2 + 2\Tr(\rho)\Tr(O^2U\rho U^\dag) + \Tr(O^2)\Tr(\rho^2) + 2\Tr(O^2 U\rho^2 U^\dag) \Big]. 
\end{align}
This equation and \lref{lem:UsefulIneqs} together imply that 
\begin{align}
\xi(1,3)
=\left| \Tr \left[ \left( O\otimes \rho^{\top}\otimes O\otimes \rho^{\top}\right) \tilde \Delta_1 \, \tilde \Delta_3 \right]\right|
\leq \frac{dB + 2\sqrt{dB} + B + 2}{6} 
=\bigo{dB} . 
\end{align}
The other bounds in Table~\ref{tab:xiUB} can be derived in a similar way.

\begin{table}
	\caption{\label{tab:xiUB}
		Upper bounds for $\xi(\mu,\nu)$ derived by using tensor-network diagrams.
		Here $\wp=\Tr(\rho^2)$.}
	\begin{math} 
		\begin{array}{c|ccccccc}
			\hline\hline
			& \quad\ \  \nu=1 \quad\ \ 
			& \quad\ \  \nu=2 \quad\ \ 
			& \quad\ \  \nu=3 \quad\ \ 
			& \quad\ \  \nu=4 \quad\ \ 
			& \quad\ \  \nu=5 \quad\ \ 
			& \quad\ \  \nu=6 \quad\ \ 
			& \quad\ \  \nu=7 \quad\ \ 
			\\[0.5ex]
			\hline
			\mu=1
			&\bigo{dB}   &\bigo{dB}  &\bigo{dB}  &\bigo{dB}  &\mathcal O\big(\sqrt{d\wp}B\big)  &\mathcal O\big(\sqrt{d\wp}B\big) &\bigo{dB}   \\[0.5ex]
			\hline
			\mu=2
			&\bigo{dB}   &\mathcal O\big(d\sqrt{d\wp}B\big)  &\mathcal O\big(d\sqrt{d\wp}B\big) 
			&\mathcal O\big(d\sqrt{d\wp}B\big)  &\mathcal O\big(dB \sqrt{\wp}\big)  &\mathcal O\big(d\sqrt{d\wp}B\big)  &\bigo{dB}    \\[0.5ex]
			\hline
			\mu=3
			&\bigo{dB}   &\mathcal O\big(d\sqrt{d\wp}B\big)  &\mathcal O\big(d\sqrt{d\wp}B\big) 
			&\mathcal O\big(d\sqrt{d\wp}B\big) &\mathcal O\big(d\sqrt{d\wp}B\big) &\mathcal O\big(dB \sqrt{\wp}\big)  &\bigo{dB}    \\[0.5ex]
			\hline
			\mu=4
			&\bigo{dB}   &\mathcal O\big(d\sqrt{d\wp}B\big)  &\mathcal O\big(d\sqrt{d\wp}B\big) 
			&\mathcal O\big(d^2B\wp\big) &\mathcal O\big(d\sqrt{d\wp}B\big)  &\mathcal O\big(d\sqrt{d\wp}B\big)  &\bigo{dB}    \\[0.5ex]
			\hline
			\mu=5
			&\mathcal O\big(\sqrt{d\wp}B\big)  &\mathcal O\big(\sqrt{d\wp}B\big) &\mathcal O\big(d\sqrt{d\wp}B\big) 
			&\mathcal O\big(d\sqrt{d\wp}B\big) 
			&\mathcal O\big(d\sqrt{d\wp}B\big)  &\mathcal O\big(B\wp\big)  &\mathcal O\big(\sqrt{d\wp}B\big)    \\[0.5ex]
			\hline
			\mu=6
			&\mathcal O\big(\sqrt{d\wp}B\big)  &\mathcal O\big(d\sqrt{d\wp}B\big) &\mathcal O\big(\sqrt{d\wp}B\big)
			&\mathcal O\big(d\sqrt{d\wp}B\big) 
			&\mathcal O\big(B\wp\big)  &\mathcal O\big(d\sqrt{d\wp}B\big)  &\mathcal O\big(\sqrt{d\wp}B\big)    \\[0.5ex]
			\hline
			\mu=7
			&\bigo{dB}  &\bigo{dB}   &\bigo{dB} &\bigo{dB} &\mathcal O\big(\sqrt{d\wp}B\big)  &\mathcal O\big(\sqrt{d\wp}B\big)  &\bigo{dB}    \\[0.5ex]
			\hline\hline
		\end{array}	
	\end{math}
\end{table}

By combining Eqs.~\eqref{eq:Cov3A}, \eqref{eq:Cov3B}, and Table~\ref{tab:xiUB}, we have 
\begin{align} 
&\operatorname{Cov}\left(\Lambda_{i,j},\Lambda_{j,i}\right)
\nonumber\\
&\leq \mathcal O \bigg(       \frac{d^4}{s^4}dB + \frac{d^3}{s^3}dB + \frac{d}{s}dB\sqrt{d\wp}
                         + \frac{d^2}{s^2}dB\sqrt{d\wp} + d^2B\wp + \frac{(d+s)^4}{s^4}dB + \frac{d(d+s)}{s^2}dB\sqrt{d\wp}
                         + \frac{d^2(d+s)}{s^3}B\sqrt{d\wp}
\nonumber\\&\qquad\quad  + \frac{(d+s)}{s}dB\sqrt{d\wp}   
                         + \frac{(d+s)^3}{s^3}B\sqrt{d\wp} + \frac{d(d+s)^2}{s^3}dB 
                         + \frac{(d+s)^2}{s^2}dB\sqrt{d\wp} + \frac{d^2(d+s)^2}{s^4}dB
\bigg) 
\nonumber\\
&= \mathcal O  \left[ \left( 
\frac{d^3}{s^4}
+ \frac{d^2}{s^3} 
+ \frac{d\sqrt{d\wp}}{s^2} 
+ \frac{\sqrt{d\wp}}{s} 
+ \wp \right) d^2B \right] ,  
\end{align} 
which confirms \lref{lem:CalCov3}. 
\end{proof}

\subsection{Both indices match with the same order}\label{sec:CalCov4}
\begin{lemma}\label{lem:CalCov4}
For distinct $i,j$, we have 	
\begin{align}
&\operatorname{Cov}\left( \Lambda_{i,j},\Lambda_{i,j}  \right)
= \mathcal O  \left[  \left( \frac{d^4}{s^4} +\frac{d^3}{s^3} +\frac{d^2}{s^2} +\frac{d}{s} + 1 \right) d^2B \wp \right] 
\end{align} 	
\end{lemma}

\begin{proof}[Proof of \lref{lem:CalCov4}]
For distinct $i,j$, we have 
\begin{align}
\operatorname{Cov}\left( \Lambda_{i,j},\Lambda_{i,j}  \right)
&= \E \left[ \Tr\left( \left( O\otimes \rho^{\top}\right) \hat{X}_i \hat{X}_j \right)
\Tr\left( \left( O\otimes \rho^{\top}\right) \hat{X}_j \hat{X}_i \right) \right] 
- d^2 \left[ \Tr\left( O\,\mathcal{U}(\rho)\right) \right]^2
\nonumber\\
&\leq \E \left[ \Tr \left( \left( O\otimes \rho^{\top}\otimes O\otimes \rho^{\top}\right) 
\left(\hat{X}_i\hat{X}_j \otimes \hat{X}_j\hat{X}_i\right) \right)\right] , 
\label{eq:Cov4A}
\end{align}
Note that the trace inside the last line of Eq.~\eqref{eq:Cov4A} can be further written as 
\begin{align}
&\Tr \left( \left( O\otimes \rho^{\top}\otimes O\otimes \rho^{\top}\right) 
\left(\hat{X}_i\hat{X}_j \otimes \hat{X}_j\hat{X}_i\right) \right) 
\nonumber\\
&=\Tr \left( 
\left( O\otimes \rho^{\top}\otimes O\otimes \rho^{\top} \otimes I \otimes I\right) 
\left( I\otimes I\otimes \hat{X}_i \otimes \hat{X}_i \right) 
\left( \hat{X}_j \otimes \hat{X}_j \otimes I\otimes I\right) 
T_{(1,5)(2,6)(3)(4)}
\right) .
\end{align}  
So the covariance satisfies
\begin{align}
&\operatorname{Cov}\left[ \Tr\left( \left( O\otimes \rho^{\top}\right) \hat{X}_i \hat{X}_j \right),
\Tr\left( \left( O\otimes \rho^{\top}\right) \hat{X}_i \hat{X}_j \right) \right] 
\nonumber\\
&\leq\Tr \left[
\left( O\otimes \rho^{\top}\otimes O\otimes \rho^{\top} \otimes I \otimes I\right) 
\left( I\otimes I\otimes \E \big[ \hat{X} \otimes \hat{X} \big] \right) 
\left( \E \big[ \hat{X} \otimes \hat{X} \big] \otimes I\otimes I\right)
T_{(1,5)(2,6)(3)(4)} 
\right]. \label{eq:Cov4B}
\end{align}

Define 
\begin{align}\label{eq:define_zeta}
\zeta(\mu,\nu):= 
\left| \Tr \left[
\left( O\otimes \rho^{\top}\otimes O\otimes \rho^{\top} \otimes I \otimes I\right) 
\left( I\otimes I\otimes  \tilde \Delta_\mu \right) 
\left(  \tilde \Delta_\nu \otimes I\otimes I\right)
T_{(1,5)(2,6)(3)(4)} 
\right]\right| 
\quad \mu,\nu\in\{1,2,\dots,7\}. 
\end{align}
According to \lref{lem:ExpXTotimesXT}, we can then expand the RHS of \eref{eq:Cov4B} as 
\begin{align}\label{eq:Cov4C}
&\text{RHS of \eref{eq:Cov4B}} 
\nonumber\\ 
&= \mathcal O \bigg(       \frac{d^4}{s^4}\zeta(1,1) + \frac{d^3}{s^3}[\zeta(1,2)+\zeta(1,3)+\zeta(2,1)+\zeta(3,1)] 
+ \frac{d}{s}[\zeta(2,4)+\zeta(3,4)+\zeta(4,2)+\zeta(4,3)]  
\nonumber\\&\qquad\quad  + \frac{d^2}{s^2}[\zeta(2,2)+\zeta(2,3)+\zeta(1,4)+\zeta(3,2)+\zeta(3,3)+\zeta(4,1)] + \zeta(4,4)
+ \frac{(d+s)^4}{s^4}\zeta(7,7) 
\nonumber\\&\qquad\quad  + \frac{d(d+s)}{s^2}[\zeta(3,5)+\zeta(3,6)+\zeta(2,5)+\zeta(2,6)+\zeta(5,2)+\zeta(5,3)+\zeta(6,2)+\zeta(6,3)] 
\nonumber\\&\qquad\quad  + \frac{d^2(d+s)}{s^3}[\zeta(1,5)+\zeta(1,6)+\zeta(5,1)+\zeta(6,1)] 
+ \frac{(d+s)}{s}[\zeta(4,5)+\zeta(4,6)+\zeta(5,4)+\zeta(6,4)]  
\nonumber\\&\qquad\quad  + \frac{(d+s)^3}{s^3}[\zeta(5,7)+\zeta(6,7)+\zeta(7,5)+\zeta(7,6)] 
+ \frac{d(d+s)^2}{s^3}[\zeta(2,7)+\zeta(3,7)+\zeta(7,2)+\zeta(7,3)]
\nonumber\\&\qquad\quad  + \frac{(d+s)^2}{s^2}[\zeta(5,5)+\zeta(5,6)+\zeta(6,5)+\zeta(6,6)+\zeta(4,7)+\zeta(7,4)] 
+ \frac{d^2(d+s)^2}{s^4}[\zeta(1,7)+\zeta(7,1)]
\bigg) . 
\end{align} 
To bound the covariance, we derive bounds for all $\zeta(\mu,\nu)$ using tensor-network diagrams, as summarized in Table~\ref{tab:zetaUB}. 
For example, 
\begin{align}
&\!\!\! 
\Tr \left[
\left( O\otimes \rho^{\top}\otimes O\otimes \rho^{\top} \otimes I \otimes I\right) 
\left( I\otimes I\otimes \tilde \Delta_4 \right) 
\left( \tilde \Delta_3 \otimes I\otimes I\right)
T_{(1,5)(2,6)(3)(4)} \right] 
\nonumber\\[0.5ex]
&=\, 
\begin{tikzpicture}[baseline={([yshift=-.5ex]current bounding box.center)},inner sep=-4mm]
\node[tensor_blue] (O1) at (1*\xratio, -0.2) {$O$};
\node[tensor_blue] (rho1) at (2*\xratio, -0.2) {$\rho^{\top}$};
\node[tensor_blue] (O2) at (3*\xratio, -0.2) {$O$};
\node[tensor_blue] (rho2) at (4*\xratio, -0.2) {$\rho^{\top}$};
\draw[-<<,draw=black] (O1) -- (1*\xratio, +0.5);
\draw[-<<,draw=black] (rho1) -- (2*\xratio, +0.5);
\draw[-<<,draw=black] (O2) -- (3*\xratio, +0.5);
\draw[-<<,draw=black] (rho2) -- (4*\xratio, +0.5);
\node[tensor_green] (U1) at (3*\xratio, -1.2) {$U$};
\node[tensor_green] (U2) at (5*\xratio, -1.2) {$U$};
\node[Permute_4] (permute_a) at (4.5*\xratio, -2.1) {$\sym^{(4)}$};
\node[tensor_green] (U3) at (3*\xratio, -3) {$U^{\dag}$};
\node[tensor_green] (U4) at (5*\xratio, -3) {$U^{\dag}$};
\node[tensor_green] (U5) at (1*\xratio, -4.2) {$U$};
\node[Permute_4] (permute_b) at (2.5*\xratio, -5.1) {$\sym^{(3)}$};
\node[tensor_green] (U6) at (1*\xratio, -6) {$U^{\dag}$};
\node[Permute_3] (permute_c) at (2*\xratio, -6.9) {$\sym^{(2)}$};
\draw[>>-,draw=black] (6*\xratio, +0.5) .. controls (6*\xratio, -0.9) and (6.5*\xratio, -0.8) ..  (6.5*\xratio, -2.2);
\draw[-,draw=black] (6.5*\xratio, -2.2) .. controls (6.5*\xratio, -2.8) and (6*\xratio, -2.8) ..  (6*\xratio, -2.4);
\draw[-,draw=black] (6*\xratio, -1.8) .. controls (6*\xratio, -1.4) and (6.4*\xratio, -1.3) ..  (6.4*\xratio, -2.1);
\draw[>>-,draw=black] (5*\xratio, +0.5) -- (U2);
\draw[-,draw=black] (O1) -- (U5);
\draw[-,draw=black] (O2) -- (U1);
\draw[-|||,draw=black](4*\xratio, -1.8) .. controls (4*\xratio, -1.1) and (4.4*\xratio, -1.1) ..  (4.4*\xratio, -1.7);
\draw[-|||,draw=black](4*\xratio, -0.5) .. controls (4*\xratio, -1.1) and (3.6*\xratio, -1.2) ..  (3.6*\xratio, -1.7);
\draw[-|||,draw=black](4*\xratio, -2.4) .. controls (4*\xratio, -3) and (3.6*\xratio, -3) ..  (3.6*\xratio, -2.5);
\draw[-|||,draw=black](2*\xratio, -4.8) .. controls (2*\xratio, -4.1) and (2.4*\xratio, -4.1) ..  (2.4*\xratio, -4.7);
\draw[-|||,draw=black](2*\xratio, -3) .. controls (2*\xratio, -4.1) and (1.6*\xratio, -4.2) ..  (1.6*\xratio, -4.7);
\draw[-|||,draw=black](2*\xratio, -5.4) .. controls (2*\xratio, -6) and (1.6*\xratio, -6) ..  (1.6*\xratio, -5.5);
\draw[-,draw=black] (2*\xratio, -3) -- (rho1);
\draw[-|||,draw=black] (4*\xratio, -3.6) .. controls (4*\xratio, -3) and (4.4*\xratio, -3.2) ..  (4.4*\xratio, -2.5);
\draw[-,draw=black] (4*\xratio, -3.6) .. controls (4*\xratio, -4.6) and (4.5*\xratio, -3.8) ..  (4.5*\xratio, -5.2);
\draw[-,draw=black] (4.5*\xratio, -5.2) .. controls (4.5*\xratio, -5.8) and (4*\xratio, -5.8) ..  (4*\xratio, -5.4);
\draw[-,draw=black] (4*\xratio, -4.8) .. controls (4*\xratio, -4.4) and (4.4*\xratio, -4.3) ..  (4.4*\xratio, -5.1);
\draw[-,draw=black] (U1) -- (3*\xratio, -1.8);
\draw[-,draw=black] (U3) -- (3*\xratio, -2.4);
\draw[-|||,draw=black] (U3) -- (3*\xratio, -4.7);
\draw[-,draw=black] (U2) -- (5*\xratio, -1.8);
\draw[-,draw=black] (U4) -- (5*\xratio, -2.4);
\draw[-,draw=black] (U5) -- (1*\xratio, -4.8);
\draw[-,draw=black] (U6) -- (1*\xratio, -5.4);
\draw[-,draw=black] (U6) -- (1*\xratio, -6.6);
\draw[-|||,draw=black] (3*\xratio, -6.6) -- (3*\xratio, -5.5);
\draw[|||-|||,draw=black] (2*\xratio, -6.5) .. controls (2*\xratio, -5.8) and (2.4*\xratio, -6.2) .. (2.4*\xratio, -5.5);
\draw[|||-<<,draw=black] (2*\xratio, -7.3) .. controls (2*\xratio, -8.5) and (6*\xratio, -8.2) ..  (6*\xratio, -9.9);
\draw[-<<,draw=black]    (1*\xratio, -7.2) .. controls (1*\xratio, -8.5) and (5*\xratio, -8.2) ..  (5*\xratio, -9.9);
\draw[-<<,draw=black]    (3*\xratio, -7.2) -- (3*\xratio, -9.9);
\draw[-<<,draw=black]    (4.4*\xratio, -5.1) .. controls (4.4*\xratio, -6) and (4*\xratio, -5) ..  (4*\xratio, -9.9);
\draw[-,draw=black] (U4) -- (5*\xratio, -7.2);
\draw[-,draw=black] (6.4*\xratio, -2.1) .. controls (6.4*\xratio, -2.6) and (6*\xratio, -2.8) .. (6*\xratio, -3.8);
\draw[-,draw=black] (6*\xratio, -3.8) -- (6*\xratio, -7.2);
\draw[-<<,draw=black]    (6*\xratio, -7.2) .. controls (6*\xratio, -8.2) and (2*\xratio, -8.5) ..  (2*\xratio, -9.9);
\draw[-<<,draw=black]    (5*\xratio, -7.2) .. controls (5*\xratio, -8.2) and (1*\xratio, -8.5) ..  (1*\xratio, -9.9);
\draw[-,densely dotted,thick,draw=red] (0.63*\xratio, -7.5) -- (0.63*\xratio, -3.72);
\draw[-,densely dotted,thick,draw=red] (0.63*\xratio, -3.72) -- (4.58*\xratio, -3.72);
\draw[-,densely dotted,thick,draw=red] (4.58*\xratio, -3.72) -- (4.58*\xratio, -7.5);
\draw[-,densely dotted,thick,draw=red] (4.58*\xratio, -7.5) -- (0.63*\xratio, -7.5);
\node[red] (T) at (4.8*\xratio, -5.6) {$\tilde \Delta_3$};
\draw[-,densely dotted,thick,draw=red] (2.63*\xratio, -3.5) -- (2.63*\xratio, -0.7);
\draw[-,densely dotted,thick,draw=red] (2.63*\xratio, -0.7) -- (6.58*\xratio, -0.7);
\draw[-,densely dotted,thick,draw=red] (6.58*\xratio, -0.7) -- (6.58*\xratio, -3.5);
\draw[-,densely dotted,thick,draw=red] (6.58*\xratio, -3.5) -- (2.63*\xratio, -3.5);
\node[red] (T) at (2.4*\xratio, -2.1) {$\tilde \Delta_4$};
\end{tikzpicture}
\!=\frac{1}{2!\cdot3!\cdot4!} \sum_{\substack{\pi_1\in\mathrm{S}_3 \\ \pi_2\in\mathrm{S}_4 \\ \pi_3\in\mathrm{S}_2}} \,
\begin{tikzpicture}[baseline={([yshift=-.5ex]current bounding box.center)},inner sep=-4mm]
\node[tensor_blue] (O2) at (1*\xratio, +0.2) {$O$};
\node[tensor_green] (U2) at (1*\xratio, -0.8) {$U$};
\node[tensor_blue] (rho1) at (3*\xratio, -0.8) {$\rho$};
\node[tensor_green] (U3) at (4*\xratio, -0.8) {$U^{\dag}$};
\node[tensor_green] (U1) at (4*\xratio, +1.2) {$U$};
\node[tensor_blue] (O1) at (4*\xratio, +2.2) {$O$};
\node[tensor_green] (U4) at (5*\xratio, -0.8) {$U$};
\node[Permute_4] (permute_b) at (2.5*\xratio, -2.3) {$T_{\pi_2}$};
\node[tensor_green] (U5) at (1*\xratio, -3.3) {$U^{\dag}$};
\node[tensor_blue] (rho2) at (2*\xratio, -3.3) {$\rho$};
\node[tensor_green] (U6) at (4*\xratio, -3.3) {$U^{\dag}$};
\node[Permute_5] (permute_c) at (3*\xratio, -4.7) {$T_{\pi_3}$};
\node[Permute_3] (permute_a) at (3*\xratio, +0.2) {$T_{\pi_1}$};
\draw[-,draw=black] (O1.south) -- (U1.north);
\draw[-,draw=black] (U1.south) -- (4*\xratio, +0.5);
\draw[-,draw=black] (U3.north) -- (4*\xratio, -0.1);
\draw[-,draw=black] (O2.south) -- (U2.north);
\draw[-,draw=black] (1*\xratio, -2) -- (U2.south);
\draw[-,draw=black] (3*\xratio, -2) -- (rho1.south);
\draw[-,draw=black] (3*\xratio, -0.1) -- (rho1.north);
\draw[-,draw=black] (2*\xratio, -0.1) -- (2*\xratio, -2);
\draw[-,draw=black] (1*\xratio, -4.4) -- (U5.south);
\draw[-,draw=black] (1*\xratio, -2.6) -- (U5.north);
\draw[-,draw=black] (2*\xratio, -2.6) -- (rho2.north);
\draw[-,draw=black] (4*\xratio, -2.6) -- (U6.north);
\draw[-<<,draw=black] (O2) -- (1*\xratio, +1);
\draw[-<<,draw=black] (U4) -- (5*\xratio, +1);
\draw[-<<,draw=black] (2*\xratio, +0.5) -- (2*\xratio, +1);
\draw[-<<,draw=black] (3*\xratio, +0.5) -- (3*\xratio, +1);
\draw[-<<,draw=black] (O1) -- (4*\xratio, +3);
\draw[-<<,draw=black] (rho2) -- (2*\xratio, -4.1);
\draw[-<<,draw=black] (3*\xratio, -2.6) -- (3*\xratio, -4.1);
\draw[-<<,draw=black] (U6) -- (4*\xratio, -4.1);
\draw[-<<,draw=black] (1*\xratio, -5) -- (1*\xratio, -5.5);
\draw[-<<,draw=black] (5*\xratio, -5) -- (5*\xratio, -5.5);
\draw[-,draw=black] (5*\xratio, -1.1) .. controls (5*\xratio, -1.5) and (4*\xratio, -1.6) ..  (4*\xratio, -2);
\draw[-,draw=black] (4*\xratio, -1.1) .. controls (4*\xratio, -1.5) and (5*\xratio, -1.6) ..  (5*\xratio, -2.5);
\draw[-,draw=black] (5*\xratio, -4.4) -- (5*\xratio, -2.5);
\end{tikzpicture}
\nonumber\\[1ex]
&=  (C_1 d^2 + C_2 d + C_3) \Tr(O_U^2)\Tr(\rho^2) 
  + (C_4 d^2 + C_5 d + C_6) (\Tr(O_U\rho))^2 
  + (C_7 d^2 + C_8 d + C_9) \Tr(O_U^2 \rho^2) 
\nonumber\\&\quad\ 
  + (C_{10} d + C_{11}) \Tr(O_U)\Tr(O_U\rho)   
  + (C_{12} d + C_{13}) (\Tr(O_U))^2 \Tr(\rho^2)
  + (C_{14} d + C_{15}) \Tr(O_U)\Tr(O_U\rho^2) 
\nonumber\\&\quad\ 
  + (C_{16} d + C_{17}) \Tr(O_U^2\rho) 
  + (C_{18} d + C_{19}) \Tr(O_U\rho O_U \rho) 
  + (C_{20} d + C_{21}) \Tr(O^2) 
  + (C_{22} d + C_{23}) (\Tr(O))^2 , 
\end{align}
where $C_1,C_2,\dots,C_{23}$ are constants, and $O_U=U^{\dag}OU$. 
This equation and \lref{lem:UsefulIneqs} together imply that 
\begin{align}
\zeta(4,3)
= \left| \Tr \left[
\left( O\otimes \rho^{\top}\otimes O\otimes \rho^{\top} \otimes I \otimes I\right) 
\left( I\otimes I\otimes \tilde \Delta_4 \right) 
\left( \tilde \Delta_3 \otimes I\otimes I\right)
T_{(1,5)(2,6)(3)(4)} \right] \right|  
= \bigo{d^2B\wp}. 
\end{align}
The other bounds in Table~\ref{tab:zetaUB} can be derived in a similar way.

\begin{table}
\caption{\label{tab:zetaUB}
Upper bounds for $\zeta(\mu,\nu)$ derived by using tensor-network diagrams.
Here $\wp=\Tr(\rho^2)$.}
\begin{math} 
\begin{array}{c|ccccccc}
\hline\hline
& \quad\ \  \nu=1 \quad\ \ 
& \quad\ \  \nu=2 \quad\ \ 
& \quad\ \  \nu=3 \quad\ \ 
& \quad\ \  \nu=4 \quad\ \ 
& \quad\ \  \nu=5 \quad\ \ 
& \quad\ \  \nu=6 \quad\ \ 
& \quad\ \  \nu=7 \quad\ \ 
\\[0.5ex]
\hline
\mu=1
&\mathcal O\big(d^2B\wp\big)  &\mathcal O\big(d^2B\wp\big)  
&\mathcal O\big(d^2B\wp\big)  &\mathcal O\big(d^2B\wp\big)  
&\mathcal O\big(\sqrt{d\wp}B\big)  &\mathcal O\big(\sqrt{d\wp}B\big)  &\bigo{dB} \\[0.5ex]
\hline
\mu=2
&\mathcal O\big(d^2B\wp\big)  &\mathcal O\big(d^2B\wp\big)  
&\mathcal O\big(d^2B\wp\big)  &\mathcal O\big(d^2B\wp\big)  
&\mathcal O\big(d\sqrt{d\wp}B\big)  &\mathcal O\big(dB\sqrt{\wp}\big)  &\bigo{dB}  \\[0.5ex]
\hline
\mu=3
&\mathcal O\big(d^2B\wp\big)  &\mathcal O\big(d^2B\wp\big)  
&\mathcal O\big(d^2B\wp\big)  &\mathcal O\big(d^2B\wp\big)  
&\mathcal O\big(\sqrt{d\wp}B\big)  &\mathcal O\big(d\sqrt{d\wp}B\big)  &\bigo{dB}     \\[0.5ex]
\hline
\mu=4
&\mathcal O\big(d^2B\wp\big)  &\mathcal O\big(d^2B\wp\big)  
&\mathcal O\big(d^2B\wp\big)  &\mathcal O\big(d^2B\wp\big) 
&\mathcal O\big(d\sqrt{d\wp}B\big)  &\mathcal O\big(d\sqrt{d\wp}B\big)  &\bigo{dB}   \\[0.5ex]
\hline
\mu=5
&\mathcal O\big(\sqrt{d\wp}B\big)  &\mathcal O\big(d\sqrt{d\wp}B\big) &\mathcal O\big(\sqrt{d\wp}B\big)
&\mathcal O\big(d\sqrt{d\wp}B\big) 
&\mathcal O\big(B\wp\big)  &\mathcal O\big(d\sqrt{d\wp}B\big)  &\mathcal O\big(\sqrt{d\wp}B\big)     \\[0.5ex]
\hline
\mu=6
&\mathcal O\big(\sqrt{d\wp}B\big)  &\mathcal O\big(dB \sqrt{\wp}\big) &\mathcal O\big(d\sqrt{d\wp}B\big) 
&\mathcal O\big(d\sqrt{d\wp}B\big) &\mathcal O\big(d\sqrt{d\wp}B\big)  
&\mathcal O\big(B\wp\big)   &\mathcal O\big(\sqrt{d\wp}B\big)   \\[0.5ex]
\hline
\mu=7
&\bigo{dB}  &\bigo{dB}   &\bigo{dB} &\bigo{dB} &\mathcal O\big(\sqrt{d\wp}B\big)  &\mathcal O\big(\sqrt{d\wp}B\big)  &\bigo{dB}     \\[0.5ex]
\hline\hline
\end{array}	
\end{math}
\end{table}

By combining Eqs.~\eqref{eq:Cov4B}, \eqref{eq:Cov4C}, and Table~\ref{tab:zetaUB}, we have 
\begin{align}
&\operatorname{Cov}\left( \Lambda_{i,j},\Lambda_{i,j}  \right)
\nonumber\\
&\leq \mathcal O \bigg(    \frac{d^4}{s^4}d^2B\wp + \frac{d^3}{s^3}d^2B\wp + \frac{d}{s}d^2B\wp  
                         + \frac{d^2}{s^2}d^2B\wp + d^2B\wp  + \frac{(d+s)^4}{s^4}dB + \frac{d(d+s)}{s^2}dB\sqrt{d\wp} 
                         + \frac{d^2(d+s)}{s^3}B\sqrt{d\wp} 
\nonumber\\&\qquad\quad  + \frac{(d+s)}{s}dB\sqrt{d\wp}  
                         + \frac{(d+s)^3}{s^3}B\sqrt{d\wp} + \frac{d(d+s)^2}{s^3}dB
                         + \frac{(d+s)^2}{s^2}dB\sqrt{d\wp} + \frac{d^2(d+s)^2}{s^4}dB
\bigg) 
\nonumber\\
&= \mathcal O  \left[  \left( \frac{d^4}{s^4} +\frac{d^3}{s^3} +\frac{d^2}{s^2} +\frac{d}{s} + 1 \right) d^2B \wp \right] ,
\end{align} 
which confirms \lref{lem:CalCov4}. 
\end{proof}

\section{Proof of \tref{thm:jointUpperBound}} \label{sec:ProofThm1}
To prove \tref{thm:jointUpperBound}, it suffices to prove the following theorem, 
which is a stronger version of \tref{thm:jointUpperBound}.  

\begin{theorem}\label{thm:UpperBoundFormal}
Suppose $s\geq 1$ is the number of systems collectively measured in our learning phase; 
accuracy parameters $0<\delta,\epsilon<1$;  
$\rho_1,\dots, \rho_M$ is a collection of quantum states chosen from $\caD(\caH)$, 
with $\max_l\Tr(\rho_l^2)=\mathscr{P}$; 
and $O_1,\dots, O_M$ is a collection of observables chosen from $\obs(B)$, with $1\leq B\leq d$. 
To estimate all $M$ properties $\Tr\left( O_l \,\mathcal{U}(\rho_l)\right)$
within additive error $\epsilon$ and failure probability at most $\delta$,
the number of queries required by our protocol for CSEU is upper bounded by 
\begin{align}
\bigo{ \left( \frac{d\mathscr{P}+s\min\{1,B\mathscr{P}\}}{\epsilon^2} + \frac{\max\{d^2,s^2\}\sqrt{B\mathscr{P}}}{s\epsilon} \right) \log\left( \frac{M}{\delta}\right)}.
\end{align} 
\end{theorem}

\begin{proof}[Proof of \tref{thm:UpperBoundFormal}]
First, we consider the case of $M=1$, namely, the goal is to accurately predict 
$\Tr\left( O \,\mathcal{U}(\rho)\right)$ for an observable $O\in\obs(B)$ and a quantum state $\rho\in\caD(\caH)$ that has purity $\wp=\Tr(\rho^2)$. Suppose we want the probability in \eref{eq:ChebyY} to be at most $1/4$, i.e., 
\begin{align}\label{eq:Pr<1/4}
\operatorname{Pr}\left\{ \left| \hat{Z}_{(r)}(O,\rho) - \Tr\left( O\,\mathcal{U}(\rho)\right) \right| \geq \epsilon\right\} \leq \frac{1}{4}. 
\end{align} 
Thanks to \pref{prop:VarMain}, it is sufficient to have
\begin{align}
\frac{1}{q} \left( \frac{d \wp}{s} + \min\left\lbrace 1, B\wp\right\rbrace  \right)  \leq \epsilon^2  
\quad \text{and} \quad
\frac{1}{q^2 } \left( \frac{d^4}{s^4} + 1 \right) B \wp \leq \epsilon^2, 
\end{align} 
where we ignored constant coefficients. These two conditions  can be achieved by choosing  
\begin{align}\label{eq:M1q}
q=  C \bigg( \frac{1}{\epsilon^2}  \max \left\lbrace\frac{d\wp}{s},\min\left\lbrace 1, B\wp\right\rbrace\right\rbrace    
  + \frac{\sqrt{B\wp}}{\epsilon}  \max \left\lbrace\frac{d^2}{s^2},1\right\rbrace \bigg) 
\end{align} 
for some constant $C$.

Recall that our final estimate is the median value $\hat{E}(O,\rho)$ given in \eref{eq:definehatE}. 
Assume that the number of batches $R$ is an odd number, then $\hat{E}(O,\rho)$ is actually some $\hat{Z}_{(r)}(O,\rho)$. 
If $\hat{E}(O,\rho)$ is a bad estimate, i.e.,
$\big| \hat{E}(O,\rho)-\Tr\left( O\,\mathcal{U}(\rho)\right) \big|\geq \epsilon$, 
then at least half of the $R$ estimates
$\hat{Z}_{(r)}(O,\rho)$ are bad. 
Hence, when \eref{eq:Pr<1/4} holds for all $r=1,2,\dots,R$, 
the failure probability of $\hat{E}(O,\rho)$ can be bounded by 
\begin{align}\label{eq:TotalFailProb}
&\Pr\left\{  \left| \hat{E}(O,\rho)-\Tr\left( O\,\mathcal{U}(\rho)\right) \right| \geq \epsilon \right\}
\leq \Pr\left\{\#\left\{r:\left| \hat{Z}_{(r)}(O,\rho) - \Tr\left( O\,\mathcal{U}(\rho)\right) \right| \geq \epsilon\right\} \geq \frac{R}{2}\right\}
\nonumber\\
&\quad \stackrel{(a)}{\leq} \sum_{j=0}^{(R-1)/2} \binom{R}{j} \left(\frac{3}{4}\right)^j \left(\frac{1}{4}\right)^{R-j}
\stackrel{(b)}{\leq} \exp\left[-R \,D\left(\left.\frac{(R-1)/2}{R} \, \right\|\, 3/4\right)\right]
\stackrel{(c)}{\leq} \mathrm{e}^{-R \, D\left(1/2 \|3/4\right)},  
\end{align}
where $(a)$ holds because the $R$ estimates
$\hat{Z}_{(r)}(O,\rho)$ are independent; $D(x\|y):=x \ln \frac{x}{y}+(1-x) \ln \frac{1-x}{1-y}$ is the Kullback-Leibler divergence; $(b)$ follows from the Chernoff bound in \lref{lem:Chernoff} below; and $(c)$ holds because $D(x\|y)$ is strictly decreasing in $x$ when $0<x<y<1$. 
Therefore, to ensure that the failure probability in \eref{eq:TotalFailProb} is at most $\delta$, 
it suffices to choose 
\begin{align}\label{eq:R=logdelta}
R= \left\lceil \frac{\ln\delta^{-1}}{D\left(1/2 \|3/4\right)} \right\rceil  
= \mathcal O\left( \log\delta^{-1}\right) .  
\end{align}

By Eqs.~\eqref{eq:M1q} and \eqref{eq:R=logdelta}, 
the total number of queries used by our protocol is 
\begin{align} \label{eq:complexityM=1}
s \cdot q \cdot R = 
\bigo{ \left( \frac{d\wp+s\min\{1,B\wp\}}{\epsilon^2} + \frac{\max\{d^2,s^2\}\sqrt{B\wp}}{s\epsilon} \right) \log\left( \delta^{-1}\right)  }, 
\end{align}
which confirms \tref{thm:UpperBoundFormal} in the special case $M=1$. 

To complete the proof, now we consider the general case of $M\geq 1$. 
According to the union bound, if our protocol can accurately predict one expectation value $\Tr\left( O \,\mathcal{U}(\rho)\right)$ with failure probability at most $\tilde\delta$, then it can also accurately predict $M$ expectation values 
$\Tr\left( O_l \,\mathcal{U}(\rho_l)\right)$ for $l=1,2,\dots,M$ simultaneously, 
with failure probability at most $M\tilde\delta$. 
This fact and the query number in \eref{eq:complexityM=1} together confirm \tref{thm:UpperBoundFormal} in the general case $M\geq 1$. 
\end{proof}

\end{widetext}
	
\begin{lemma}[Chernoff bound]\label{lem:Chernoff}
Suppose $0<p<1$, integers $k,z\geq 0$, and $k \leq p z$. Then we have 
\begin{align}
\sum_{j=0}^k\binom{z}{j} p^j(1-p)^{z-j} 
\leq 
\mathrm{e}^{-z D\left(\left. \!\frac{k}{z} \right\| p\right)}. 
\end{align}
\end{lemma}

\section{Proof of \lref{lem:AverageGoal}} \label{sec:AverageGoal}
Note that $\Pr\left\{ \Xi(\caT,O)\geq \epsilon \right\} \leq\delta$ is satisfied if the following relations hold:
\begin{align}\label{eq:AveGoalproof1}
{\E} [ \Xi(\caT,O)]\leq \frac{\epsilon}{2}, 
\ 
\operatorname{Pr}\left\{ \left|\Xi(\caT,O) - {\E} \left[ \Xi(\caT,O)\right] \right| \geq \frac{\epsilon}{2} \right\}  \leq \delta.  
\end{align}
Here and afterwards,  the expectation $\E$ is over the randomness in our protocol.

First, we have
\begin{align}\label{eq:Xileqsqrt}
\Xi(\caT,O)
&:= \underset{\rho \sim \caT}{\rm{Ave}} \, \left|E(O,\rho)-\Tr\left( O\,\mathcal{U}(\rho)\right) \right|
\nonumber\\
&\leq \sqrt{ \underset{\rho \sim \caT}{\rm{Ave}} \left( \left|\hat{E}(O,\rho)-\Tr\left( O\,\mathcal{U}(\rho)\right) \right|^2 \right) }. 
\end{align}
It follows that 
\begin{align}\label{eq:AveGoalproof2}
{\E} \left[ \Xi(\caT,O)\right] 
&\leq {\E} 
\sqrt{ \underset{\rho \sim \caT}{\rm{Ave}} \left( \left|\hat{E}(O,\rho)-\Tr\left( O\,\mathcal{U}(\rho)\right) \right|^2 \right) }
\nonumber\\
&\leq \sqrt{{\E} \left[ \underset{\rho \sim \caT}{\rm{Ave}} \left( \left|\hat{E}(O,\rho)-\Tr\left( O\,\mathcal{U}(\rho)\right) \right|^2 \right) \right]} .
\end{align}
Second, Chebyshev’s inequality implies that 
\begin{align}
&\operatorname{Pr}\left\{ \left|\Xi(\caT,O) - {\E} \left[ \Xi(\caT,O)\right] \right| \geq \frac{\epsilon}{2} \right\} 
\nonumber\\
&\leq 
\frac{4}{\epsilon^2} {\Var} \left[ \Xi(\caT,O)\right]
\leq 
\frac{4}{\epsilon^2} {\E} \left[ \Xi(\caT,O)^2 \right] 
\nonumber\\
&\leq 
\frac{4}{\epsilon^2} {\E} \left[ \underset{\rho \sim \caT}{\rm{Ave}} \left( \left|\hat{E}(O,\rho)-\Tr\left( O\,\mathcal{U}(\rho)\right) \right|^2 \right) \right], \label{eq:AveGoalproof3}
\end{align}
where the last inequality follows from \eref{eq:Xileqsqrt} again. 

According to equations \eqref{eq:AveGoalproof1}, \eqref{eq:AveGoalproof2}, and \eqref{eq:AveGoalproof3}, 
the relation $\Pr\left\{ \Xi(\caT,O)\geq \epsilon \right\} \leq\delta$ is satisfied if  
\begin{align}\label{eq:AveGoalproof4}
{\E} \left[ \underset{\rho \sim \caT}{\rm{Ave}} \left(  \left|\hat{E}(O,\rho)-\Tr\left( O\,\mathcal{U}(\rho)\right) \right|^2 \right) \right]
\leq \frac{\delta\epsilon^2}{4}. 
\end{align}
In addition, we have 
\begin{align}\label{eq:AveGoalproof5}
&\,{\E} \left[ \underset{\rho \sim \caT}{\rm{Ave}} \left(  \left|\hat{E}(O,\rho)-\Tr\left( O\,\mathcal{U}(\rho)\right) \right|^2 \right) \right]
\nonumber\\
&=
\underset{\rho \sim \caT}{\rm{Ave}} \left[ {\E} \left(  \left|\hat{E}(O,\rho)-\Tr\left( O\,\mathcal{U}(\rho)\right) \right|^2 \right) \right] 
\nonumber\\
&=
\underset{\rho \sim \caT}{\rm{Ave}} \Big[ {\E} 
\Big( \left|\hat{E}(O,\rho)\right|^2-\hat{E}(O,\rho)^*\Tr\left( O\,\mathcal{U}(\rho)\right)
\nonumber\\
&\qquad\quad -\hat{E}(O,\rho)\Tr\left( O\,\mathcal{U}(\rho)\right) + [\Tr\left( O\,\mathcal{U}(\rho)\right)]^2 \Big) \Big] 
\nonumber\\
&\!\stackrel{(a)}{=} \underset{\rho \sim \caT}{\rm{Ave}} 
\left[ {\E} \left( \left|\hat{E}(O,\rho)\right|^2 \right) - \left| {\E} \left[\hat{E}(O,\rho)\right] \right|^2 \right] 
\nonumber\\
&=\underset{\rho \sim \caT}{\rm{Ave}} \; {\Var} \left[ \hat{E}(O,\rho) \right] ,  
\end{align}
where $(a)$ holds because ${\E}[\hat{E}(O,\rho)]=\Tr\left( O\,\mathcal{U}(\rho)\right)$.
Equations \eqref{eq:AveGoalproof4} and \eqref{eq:AveGoalproof5} together confirm \lref{lem:AverageGoal}.

\section{Additional results on Hamiltonian learning}\label{sec:appHamiltonian}

\subsection{Performance of previous protocols for learning an arbitrary Hamiltonian}

To illustrate the advantage of our Hamiltonian learning protocol, here we provide more details 
on the performance of previous protocols for learning an $n$-qubit unknown Hamiltonian $H$, including those in Refs.~\cite{caro2022learning,ma2024learning,zhao2024learning,hu2025ansatz,castaneda2023hamiltonian,stilck2024efficient,gu2024practical}.
These protocols are highly efficient when $H$ has a  sparse or local structure.  
However, we show that none of them can learn an arbitrarily complex $H$ with respect to the normalized Frobenius norm (NFN) using  only 
$\tilde{\mathcal{O}}(d^3)$ Hamiltonian queries or total evolution time, even with access to quantum memory.  
Actually, most previous works did not consider the problem of learning fully arbitrary Hamiltonians, and 
only a few existing protocols can be modified to learn a completely unknown Hamiltonian under our framework.

The Hamiltonian $H$ of interest can be expressed as
\begin{align}
	H=\sum_{\bfk\in\{0,1,2,3\}^n} \mu(\bfk) \sigma_\bfk, 
\end{align}  
where $\sigma_\bfk$ are Pauli operators and $\mu(\bfk)\in \mathbb{R}$ are Pauli coefficients. 
$H$ is called a \emph{$k$-body Hamiltonian with $M$ terms}
if $\mu(\bfk)\ne0$ only for $M$ Pauli terms $\sigma_\bfk$ with ${\rm wt}(\sigma_\bfk)\leq k$,  
where the weight ${\rm wt}(\sigma_\bfk)$ denotes the number of qubits on which $\sigma_\bfk$ acts nontrivially. 
Most previous Hamiltonian learning protocols focus on estimating all non-zero Pauli coefficients of $H$. 
That is, they aim to output accurate estimates $\hat{\mu}(\bfk)$ for $\mu(\bfk)$ that satisfy  
$\|\hat{\boldsymbol{\mu}}-\boldsymbol{\mu}\|_\infty \leq \epsilon$ with high success probability. Here 
$\boldsymbol{\mu}:=(\mu(\bfk))_{\bfk\in \{0,1,2,3\}^n}$ and $\hat{\boldsymbol{\mu}}:=(\hat{\mu}(\bfk))_{\bfk\in \{0,1,2,3\}^n}$
are the vectors of true coefficients and estimated  coefficients, respectively, and
\begin{align}
	\|\hat{\boldsymbol{\mu}}-\boldsymbol{\mu}\|_p 
	= \left(\sum_{\bfk\in\{0,1,2,3\}^n} |\hat{\mu}(\bfk)-\mu(\bfk)|^{p} \right)^{1/p}
\end{align} 
denotes the $\ell^p$ norm of $\hat{\boldsymbol{\mu}}-\boldsymbol{\mu}$. 
Let $\hat{H}=\sum_{\bfk} \hat\mu(\bfk) \sigma_\bfk$ be the learned description of $H$, 
then we have 
\begin{align}
	\frac{1}{\sqrt{d}}\|\hat{H}-H\|_2
	=\|\hat{\boldsymbol{\mu}}-\boldsymbol{\mu}\|_2 
	\leq  d\|\hat{\boldsymbol{\mu}}-\boldsymbol{\mu}\|_\infty,   
\end{align} 
where the inequality is saturated when $|\hat{\mu}(\bfk)-\mu(\bfk)|$ are equal for all $\bfk\in \{0,1,2,3\}^n$. 
Therefore, estimating all Pauli coefficients of $H$ within additive error $\epsilon$
only implies learning $H$ within error $d\epsilon$ in NFN. 

Reference~\cite{caro2022learning} proposed an efficient protocol for learning the Pauli coefficients of an arbitrary traceless Hamiltonian $H$. 
It uses parallel queries to the time evolution under $H$ in experiments, and thus needs quantum memory. 
To achieve $\|\hat{\boldsymbol{\mu}}-\boldsymbol{\mu}\|_\infty \leq \epsilon$ with high probability, this protocol requires $\tilde{\mathcal{O}}(n \|H\|^4 \epsilon^{-4})$ queries to $H$ and a total evolution time of $\tilde{\mathcal{O}}(n \|H\|^3 \epsilon^{-4})$. 
A similar performance is achieved by the protocol in Ref.~\cite{zhao2024learning}, which uses sequential queries to $H$. 
According to the above discussions, if the goal is to learn $H$ within error $\epsilon$ in NFN, then   
the query number and the evolution time required by these protocols are both $\tilde{\mathcal{O}}(d^4 \epsilon^{-4})$, assuming that $\|H\|=\mathcal{O}(\poly(n))$. 

A recent work Ref.~\cite{ma2024learning} proposed a protocol 
for learning a $k$-body Hamiltonian $H$ with $M$ terms, which is robust to state preparation and measurement (SPAM) errors. 
It employs sequential queries to $H$ in each experiment and achieves the following performance: 
By using $N=\tilde{\mathcal{O}}(3^k M)$ independent experiments and a total evolution 
time $T=\tilde{\mathcal{O}}(9^k M^{1/p+1/2}\epsilon^{-1})$, it can output estimate $\hat{\mu}(\bfk)$ for every $\mu(\bfk)$ such that $\|\hat{\boldsymbol{\mu}}-\boldsymbol{\mu}\|_p\leq \epsilon$ with high probability, where $p=1$ or 2.  
This protocol is highly efficient when $k$ is a constant. 
However, if we do not make prior assumptions about the structure of $H$, then $k$ and $M$ can take values up to $k=n$ and $M=d^2$.
For this worst-case scenario, learning $H$ within error $\epsilon$ in NFN 
(that is, $\|\hat{\boldsymbol{\mu}}-\boldsymbol{\mu}\|_2\leq \epsilon$) requires a
$T=\tilde{\mathcal{O}}(9^n d^2\epsilon^{-1})\approx\tilde{\mathcal{O}}(d^{5.170}\epsilon^{-1})$ total evolution time
and $N=\tilde{\mathcal{O}}(3^n d^2)\approx\tilde{\mathcal{O}}(d^{3.585})$ experiments.
Another work \cite{hu2025ansatz} showed that one can achieve 
$\|\hat{\boldsymbol{\mu}}-\boldsymbol{\mu}\|_\infty\leq \epsilon$ from 
$T=\tilde{\mathcal{O}}(M^{2}\epsilon^{-1})$ total evolution time. 
From the above discussions, this result implies a  $T=\tilde{\mathcal{O}}(d^5\epsilon^{-1})$ time complexity 
for learning an arbitrary Hamiltonian within error $\epsilon$ in NFN. 

Reference~\cite{castaneda2023hamiltonian} proposed an efficient protocol for learning an unknown 
Hamiltonian $H$ with $M$ Pauli terms.
In order to learn $H$ within error $\epsilon$ in NFN, it needs 
$\tilde{\mathcal{O}}(M\|H\|^2\epsilon^{-2})$ Hamiltonian queries,  which reduces to $\tilde{\mathcal{O}}(d^2\epsilon^{-2})$
if $H$ is an arbitrary unknown Hamiltonian satisfying $\|H\|=\mathcal{O}(\poly(n))$. 
However, this protocol requires not only the time evolution unitary $\rme^{-\rmi Ht}$, but also the backwards evolution unitary $\rme^{\rmi Ht}$, which is physically unrealizable.

There are also Hamiltonian learning protocols that do not use quantum memory and look beyond Hamiltonians with sparse interactions \cite{stilck2024efficient,gu2024practical}. 
To learn all coefficients of a Hamiltonian $H$ within additive error $\epsilon$, 
the protocol of Ref.~\cite{gu2024practical} requires $\tilde{\mathcal{O}}(M^4\epsilon^{-2})$ queries, 
and thus is unsatisfactory when $H$ contains exponentially many terms; 
the protocol of Ref.~\cite{stilck2024efficient} requires $\tilde{\mathcal{O}}(9^k\epsilon^{-2})$ queries,
and thus is unsatisfactory when $H$ contains high-weight terms. 
From the previous discussions, when the goal is to learn an arbitrary Hamiltonian within error $\epsilon$ in NFN, the query costs of these protocols are more than $\tilde{\mathcal{O}}(d^5\epsilon^{-2})$.

Compared with these existing protocols, our Hamiltonian learning protocol demonstrates superior performance 
by simultaneously (i) operating without the requirement of quantum memory, and 
(ii) achieving a smaller query number and total evolution time (see \tref{thm:Hamiltonian} in the main text), 
making it favorable for real-world implementation.

\subsection{Details of our Hamiltonian learning protocol and its performance}
Given time evolution access to an arbitrary $n$-qubit traceless Hamiltonian $H$ with $\|H\|=\mathcal{O}(\poly(n))$. 
Our goal is to learn a classical description $\hat{H}\in \mathcal{L}_{\rm H}(\mathcal{H})$ of the Hamiltonian, 
such that the following condition holds: 
\begin{equation}\label{eq:HamiltonianGoalapp}
	\Pr\left\{ \frac{1}{\sqrt{2^n}}\left\|\hat{H}-H\right\|_2\leq\epsilon\right\} \geq 1-\delta 
\end{equation}
for accuracy $0< \delta,\epsilon<1/3$. 
Next, we present our learning protocol for this task in Appendix~\ref{sec:Hprotocol} and 
show its performance in Appendix~\ref{sec:Hperformance}.

\subsubsection{Hamiltonian learning protocol}\label{sec:Hprotocol}

We follow the Hamiltonian learning protocol sketched in Sec.~\ref{sec:HamiltProtocol} of the main text 
with the following parameter choices:
Our evolution time in each experiment is set to be $t_j = \tfrac{T}{2}(1+z_j)$ for $j=1,2,\dots, L$, where 
\begin{align}
&T= \frac{1}{\norm{H}}, \qquad
z_j = -\cos\left(\frac{2j-1}{2L} \pi\right), 
\nonumber\\
&L=\left\lceil 2\log\left(\frac{8\norm{H}}{\sqrt{2\pi}\ln(2)\varepsilon}\right)\right\rceil, \qquad
\varepsilon=\frac{\epsilon}{d}.  
\end{align}

In the first step of our protocol, for each time $t_j$ ($1\leq j\leq L$), we treat $\caU_{t_j}$\! as an unknown unitary channel, and implement our CSEU protocol in Sec.~\ref{sec:OurProtocol} to estimate the expectation values  $\<\sigma_{\bfk'}\>_{\rho_{\bfk}(t_j)}=\Tr\left[\sigma_{\bfk'}\,\caU_{t_j}\!(\rho_{\bfk})\right]$, 
where the input state  $\rho_{\bfk}$ and the measured observable $\sigma_{\bfk'}$ are chosen as follows 
(see also Ref.~\cite[Lemma~6.1]{caro2022learning}):  
Let $1\leq i\leq n$ be such that $\bfk_i\ne0$; let $\bfk'_i\in\{1,2,3\} \backslash \{\bfk_i\}$ and set $\bfk'_j=0$ for all $j\ne i$; let 
\begin{equation}
\rho_{\bfk}:=\frac{1}{d}\left( I+\frac{\rmi}{2} [\sigma_{\bfk},\sigma_{\bfk'}]\right) .
\end{equation}
We denote by $2\hat{\alpha}_{t_j}^{(\bfk)}\in \mathbb{R}$ our estimate of $\Tr\left[\sigma_{\bfk'}\,\caU_{t_j}(\rho_{\bfk})\right]$, which should achieve the following accuracy: 
\begin{align}\label{eq:hatalphaaccurate1}
\Pr\Big\{\,&\abs{2\hat{\alpha}_{t_j}^{(\bfk)}-\Tr\left[\sigma_{\bfk'}\,\caU_{t_j}(\rho_{\bfk})\right]}\leq\tilde{\varepsilon}, 
\nonumber\\
 &\qquad\quad  \forall\, \bfk\in\{0,1,2,3\}^n \backslash \{0^n\} \Big\}
\geq 1-\frac{\delta}{L}, 
\end{align}
where
\begin{align}
	\tilde{\varepsilon}:=\frac{3T\varepsilon}{4(L-1)L(2L-1)}. 
\end{align}
Let $1\leq s\leq d$ be the number of systems collectively measured by our CSEU protocol in each experiment. 
As explained in Sec.~\ref{sec:OurProtocol}, if $s>1$, then quantum memory is required. 
If instead $s=1$, then the CSEU protocol does not need quantum memory and ancillary qubits;  
in each round of experiments, 
it requires only the preparation of a computational-basis state $|0\>^{\otimes n}$, the application of a random Clifford unitary to a single system, a measurement in the computational basis, 
and one application of $\caU_{t_j}$\! [\,i.e., one real $\mathcal{O}(1/\|H\|)$-time evolution under $H$].

In the second step of our Hamiltonian learning protocol, we perform Chebyshev polynomial interpolation \cite{gu2024practical,caro2022learning} to estimate the Pauli coefficients $\mu(\bfk)$ by using $\hat{\alpha}_{t_j}^{(\bfk)}$ learned in the first step. 
The Chebyshev coefficients are estimated as 
\begin{equation}
\begin{aligned}
\hat{b}_{0}^{(\bfk)} &= \frac{1}{L} \sum_{j=1}^L \hat{\alpha}_{t_j}^{(\bfk)}, 
\\
\hat{b}_{l}^{(\bfk)} &= \frac{2}{L} \sum_{j=1}^L \hat{\alpha}_{t_j}^{(\bfk)} \, T_l (z_j), 
\quad 
1\leq l\leq L-1, 
\end{aligned} 
\end{equation}
where $T_l$ is the $l$th Chebyshev polynomial.
They in turn give rise to the first-order Taylor coefficients \cite{gu2024practical,caro2022learning}
\begin{equation}\label{eq:hatmuAdefine}
	\hat{\mu}(\bfk) := -\frac{2}{T} \sum_{l=0}^{L-1} (-1)^l \, \hat{b}_{l}^{(\bfk)}\, l^2,
	\quad 
	\bfk\in\{0,1,2,3\}^n \backslash \{0^n\}.  
\end{equation} 
Here $\hat{\mu}(\bfk)$ serves as our estimate of $\mu(\bfk)$. 
Finally, we use these estimated values to construct our learned description:
\begin{equation}\label{eq:learnedHapp}
	\hat{H}=\sum_{\bfk\in\{0,1,2,3\}^n \backslash \{0^n\}} \hat\mu(\bfk) \sigma_\bfk . 
\end{equation}

\subsubsection{Performance of the protocol}\label{sec:Hperformance}

The performance of our Hamiltonian learning protocol in the previous subsection
is stated in the following theorem, 
which is a stronger version of \tref{thm:Hamiltonian} in the main text.

\begin{theorem}\label{thm:HamiltonianMore}
Suppose the unknown $n$-qubit Hamiltonian $H$ to be learned is traceless and satisfies $\|H\|=\mathcal{O}(\poly(n))$; 
parameters $0< \delta,\epsilon<1/3$;
and $s$ is the number of systems collectively measured  in each experiment.  
With the accuracy guaranteed in \eref{eq:hatalphaaccurate1}, 
the learned description $\hat{H}$ in \eref{eq:learnedHapp} satisfies \eref{eq:HamiltonianGoalapp}. 
When $s=1$, our Hamiltonian learning protocol is ancilla-free and dose not use quantum memory;   
its query number to $H$ and total evolution time both scale as
\begin{align}\label{eq:s1HnumberApp}
	\tilde{ \mathcal{O}} 
	\left( \bigg( \frac{d^2}{\epsilon^2} + \frac{d^3}{\epsilon} \bigg) \log(\delta^{-1})\right).                       
\end{align} 
When $s=\Theta(\sqrt{d})$, our Hamiltonian learning protocol uses quantum memory; its query number to $H$ and total evolution time both scale as
\begin{align}\label{eq:s>1Hnumber}
	\tilde{ \mathcal{O}} \left( \frac{d^{2.5}}{\epsilon^2} \log(\delta^{-1}) \right).                     
\end{align}  
\end{theorem}

\begin{proof}[Proof of \tref{thm:HamiltonianMore}]
We begin by proving that $\hat{H}$ given in \eref{eq:learnedHapp} indeed satisfies \eref{eq:HamiltonianGoalapp}. 
First, by the union bound, the accuracy guaranteed in \eref{eq:hatalphaaccurate1} implies that 
the following condition holds with probability at least $1-\delta$:
\begin{align}\label{eq:hatalphaaccurate2}
&\abs{2\hat{\alpha}_{t_j}^{(\bfk)}-\Tr\left[\sigma_{\bfk'}\,\caU_{t_j}(\rho_{\bfk})\right]}\leq\tilde{\varepsilon} 
\nonumber\\  
&\qquad \forall\, \bfk\in\{0,1,2,3\}^n \backslash \{0^n\} \ \ \text{and} \ \ 1\leq j\leq L.  
\end{align}
Second, according to Ref.~\cite[Appendix~D]{caro2022learning}, 
if \eref{eq:hatalphaaccurate2} holds with $\tilde{\varepsilon}=\frac{3T\varepsilon}{4(L-1)L(2L-1)}$, then 
$\hat{\mu}(\bfk)$ constructed in \eref{eq:hatmuAdefine} is an $\varepsilon$-accurate estimate of $\mu(\bfk)$
for all $\bfk\ne 0^n$, i.e., 
\begin{equation}\label{eq:hatmuAaccurate}
	\abs{\hat{\mu}(\bfk)-\mu(\bfk)}\leq \varepsilon
	\quad \forall\, \bfk\in\{0,1,2,3\}^n \backslash \{0^n\} . 
\end{equation}
Third, as explained in Sec.~\ref{sec:HamiltProtocol}, if \eref{eq:hatmuAaccurate} holds with $\varepsilon=\epsilon/d$, then $\hat{H}$ in \eref{eq:learnedHapp} satisfies 
\begin{equation}\label{eq:H-hatH}
\frac{1}{\sqrt{d}}\left\|\hat{H}-H\right\|_2\leq\epsilon. 
\end{equation}
In conclusion, our protocol learns the unknown Hamiltonian within error $\epsilon$ in NFN, with success probability at least $1-\delta$.

Next, we calculate the number of queries to $H$ and total evolution time required by our protocol. 
Notice that the Hamiltonian $H$ is queried only in the first step of our protocol, 
in which we use the CSEU method to estimate linear properties of the time-evolution channel $\caU_{t_j}$. 
Let $1\leq s\leq d$ be the number of systems collectively measured by our CSEU protocol in each experiment. 
According to \tref{thm:UpperBoundFormal} in Appendix~\ref{sec:ProofThm1}, for each $1\leq j\leq L$, 
in order to guarantee the relation \eqref{eq:hatalphaaccurate1}, our CSEU protocol  
needs to query $\caU_{t_j}$\! (\,i.e., the Hamiltonian $H$) the following number of times: 
\begin{align}\label{eq:N1H2} 
N_1&= \mathcal{O}\left[ \left( \frac{s}{\tilde{\varepsilon}^2} + \frac{d^2}{s\tilde{\varepsilon} } \right) \log\left( \frac{4^n-1}{\delta/L}\right) \right] 
\nonumber\\
&= \tilde{ \mathcal{O}} \left[ \left( \frac{s d^2}{\epsilon^2} + \frac{d^3}{s \epsilon} \right) \log\left(\delta^{-1}\right)  \right],   
\end{align} 
where we have used the relations $\Tr(\rho_{\bfk}^2)=\bigo{1/d}$ and $\Tr(\sigma_{\bfk'}^2)=d$, 
and the assumption $\|H\|=\mathcal{O}(\poly(n))$. 
In the following, we consider two cases depending on the value of $s$.

\begin{itemize}
\item[1.] First, when $s=1$, our protocol does not use quantum memory and ancillary qubits.
In this case, the total number of Hamiltonian queries is given by 
\begin{align}\label{eq:Ntots1} 
N_{\rm total}
= L N_1  = \tilde{ \mathcal{O}} \left[ \left( \frac{d^2}{\epsilon^2} + \frac{d^3}{\epsilon} \right) \log\left(\delta^{-1}\right)  \right],     
\end{align} 
and the total evolution time under $H$ reads 
\begin{align}\label{eq:Ttots1} 
T_{\rm total}
&= N_1 \sum_{j=1}^L t_j
\leq N_1 L T
\nonumber\\
&=\tilde{ \mathcal{O}} \left[ \left( \frac{d^2}{\epsilon^2} + \frac{d^3}{\epsilon} \right) \log\left(\delta^{-1}\right)  \right].    
\end{align} 
Equations \eqref{eq:Ntots1} and \eqref{eq:Ttots1} confirm \eref{eq:s1HnumberApp} in \tref{thm:HamiltonianMore}.

\item[2.] Second, when $s=\Theta(\sqrt{d})$, our protocol uses quantum memory. 
In this case, the total number of queries to $H$ is 
\begin{align}\label{eq:Ntots>1} 
N_{\rm total}
&= L N_1  
\nonumber\\
&= \tilde{ \mathcal{O}} \left[ \left( \frac{d^2 \, \Theta(\sqrt{d})}{\epsilon^2} + \frac{d^3}{\Theta(\sqrt{d})\, \epsilon} \right) \log\left(\delta^{-1}\right)  \right]  
\nonumber\\
&= \tilde{ \mathcal{O}} \left( \frac{d^{2.5}}{\epsilon^2} \log(\delta^{-1}) \right), 
\end{align} 
and the total evolution time reads 
\begin{align}\label{eq:Ttots>1} 
T_{\rm total}
= N_1 \sum_{j=1}^L t_j
\leq N_1 L T
= \tilde{ \mathcal{O}} \left( \frac{d^{2.5}}{\epsilon^2} \log(\delta^{-1}) \right). 
\end{align} 
Equations \eqref{eq:Ntots>1} and \eqref{eq:Ttots>1} confirm \eref{eq:s>1Hnumber} in \tref{thm:HamiltonianMore}. 
\end{itemize}
\end{proof}

\section{Proof of \eref{eq:VarOTOC}}\label{sec:proofVarOTOC}

\subsection{Main proof}
When $\rho$ is the infinite temperature thermal state $I/d$, and $W$, $V$ are unitary, traceless, and Hermitian operators, 
the estimator $\hat{C}(U)$ defined in \eref{eq:hatCU} can be rewritten as 
\begin{align}
\hat{C}(U)
&:= 
\frac{1}{m(m-1)d}\sum_{i\ne j} \hat{D}(i,j), 
\nonumber\\
\hat{D}(i,j) 
&:=\Tr \left[ \left( \hat{X}_i\otimes\hat{X}_j\right)  
                          \left( W\otimes V^{\top}\otimes W\otimes V^{\top}\right)  T_{(1,3)}\right] . 
\end{align}
So the variance of $\hat{C}(U)$ can be expanded as 
\begin{align}
\Var \left[ \hat{C}(U) \right] 
&= \frac{\sum_{i\ne j}\sum_{k\ne\ell} \Cov \left[ \hat{D}(i,j),\hat{D}(k,\ell) \right]}{m^2(m-1)^2 d^2}  .  \label{eq:VarOTOCdef}
\end{align} 
We need to bound all of these covariance terms to bound the variance. 
When all indices $i,j,k,\ell$ are distinct, the covariance is 0 by independence.
For the other cases, the covariance can be bounded by 
\begin{align}
&\Cov \left[ \hat{D}(i,j),\hat{D}(k,\ell) \right]
\nonumber\\
&=
\E\left[ \hat{D}(i,j)\hat{D}(k,\ell)^*\right] -\E\left[\hat{D}(i,j)\right] \E\left[\hat{D}(k,\ell)^*\right] 
\nonumber\\
&= 
\E\left[ \hat{D}(i,j)\hat{D}(k,\ell)\right] -\left| \E\big[\hat{D}(i,j)\big]\right| ^2  
\nonumber\\
&\leq \E\left[ \hat{D}(i,j)\hat{D}(k,\ell)\right].  \label{eq:UBotocCov} 
\end{align}
Here the second equality holds because $\hat{D}(k,\ell)$ is real. To see this, note that 
\begin{align}
\hat{D}(k,\ell)^*
&= \Tr \left[  T_{(1,3)}^\dag \left( W\otimes V^{\top}\otimes W\otimes V^{\top}\right) ^\dag \big( \hat{X}_k\otimes\hat{X}_\ell\big) ^\dag \right] 
\nonumber\\
&= \Tr \left[ \left( \hat{X}_k\otimes\hat{X}_\ell\right)  T_{(1,3)} \left( W\otimes V^{\top}\otimes W\otimes V^{\top}\right)  \right] 
\nonumber\\
&= \hat{D}(k,\ell) ,  
\end{align}
where the second equality holds because $\hat{X}_k$, $\hat{X}_\ell$, $W$, $V$, and $T_{(1,3)}$ are all Hermitian operators;  
and the last equality holds because $W\otimes V^{\top}\otimes W\otimes V^{\top}$ and $T_{(1,3)}$ commute.

The following lemma summarizes the upper bounds for $\E\big[\hat{D}(i,j)\hat{D}(k,\ell)\big]$ in different cases. 

\begin{widetext}
\begin{lemma} 
\label{lem:allOTOCcovs}
For each combination of $i,j,k,\ell$ with $i \ne j$ and $k \ne \ell$, 
the term $\E\big[\hat{D}(i,j)\hat{D}(k,\ell)\big]$ is  
\begin{enumerate}
\item $\mathcal O \Big( \big( \frac{d^2}{s^2}  +\frac{d}{s} + 1 \big) d^2 \Big)$
when exactly one index matches $(|\{ i, j \} \cap \{ k, \ell \}| = 1)$. 

\item $\mathcal O \Big( \big( \frac{d^2}{s^2}  +\frac{d}{s} + 1 \big)^2 d^4 \Big)$
when both indices match $(|\{ i, j \} \cap \{ k, \ell \}| = 2)$. 
\end{enumerate}
\end{lemma}

For all combinations of $i,j,k,\ell\in\{1,2,\dots,m\}$ with $i \ne j$ and $k \ne \ell$, there are $\mathcal O (m^4)$ combinations where all four indices $i,j,k,\ell$ are distinct;
$\mathcal O (m^3)$ combinations where exactly one index matches $(|\{ i, j \} \cap \{ k, \ell \}| = 1)$; 
and $\mathcal O (m^2)$ combinations where both indices match $(|\{ i, j \} \cap \{ k, \ell \}| = 2)$.
These facts, together with Eqs.~\eqref{eq:VarOTOCdef}, \eqref{eq:UBotocCov} and Lemma~\ref{lem:allOTOCcovs}, imply that 
\begin{align}\label{eq:VarOt2}
\Var \left[ \hat{C}(U) \right] 
\leq \frac{1}{m^2(m-1)^2d^2} \, 
\mathcal O \left[  m^3 \left( \frac{d^2}{s^2}  +\frac{d}{s} + 1 \right) d^2 
+ m^2 \left( \frac{d^2}{s^2}  +\frac{d}{s} + 1 \right)^2 d^4 \right]
= \mathcal O \left[\, \frac{1}{m} \left( \frac{d^2}{s^2} + 1 \right)
+ \frac{d^2}{m^2} \left( \frac{d^4}{s^4} + 1 \right) \right] ,                            
\end{align} 
which confirms \eref{eq:VarOTOC} in the main text.

\subsection{Proof of the auxiliary \lref{lem:allOTOCcovs}}
To complete the proof of \eref{eq:VarOTOC}, in the following we prove \lref{lem:allOTOCcovs},
which bounds the term $\E\big[\hat{D}(i,j)\hat{D}(k,\ell)\big]$ for all combinations of $(i,j,k,\ell)$ with $i \neq j$ and $k \neq \ell$. 
In particular, Appendix~\ref{sec:OTOCCalCov1} deals with the case in which exactly 
one index matches $(|\{ i, j \} \cap \{ k, \ell \}| = 1)$; 
Appendix~\ref{sec:OTOCCalCov2} deals with the case in which both indices match $(|\{ i, j \} \cap \{ k, \ell \}| = 2)$.

\subsubsection{Proof of the first statement in \lref{lem:allOTOCcovs}}\label{sec:OTOCCalCov1}
Suppose indices $i,j,k \in\{1,\dots,m\}$ are all distinct. Then we have 
\begin{align}
\E\left[ \hat{D}(i,j)\hat{D}(i,k)\right] 
&= \E \Big\lbrace \! \Tr \left[  (\hat{X}_i\otimes\hat{X}_j)(W\otimes V^{\top}\otimes W\otimes V^{\top}) T_{(1,3)}\right] 
\Tr \left[  (\hat{X}_i\otimes\hat{X}_k)(W\otimes V^{\top}\otimes W\otimes V^{\top}) T_{(1,3)}\right] \Big\rbrace 
\nonumber\\
&= \E \Big\lbrace \left[ \Tr \left((\hat{X}\otimes\Upsilon_\caU)(W\otimes V^{\top}\otimes W\otimes V^{\top}) T_{(1,3)}\right)\right]^2 \Big\rbrace , 
\end{align} 
where the second equality holds because $\E[\hat{X}_j]=\E[\hat{X}_k]=\Upsilon_\caU$ [see \eref{eq:EhatX} in the main text] 
and that $\hat{X}_i$, $\hat{X}_j$, $\hat{X}_k$ are independent snapshots. 
Drawing tensor-network diagrams, we have  
\begin{align}
&\Tr \left[  (\hat{X}\otimes\Upsilon_\caU)(W\otimes V^{\top}\otimes W\otimes V^{\top}) T_{(1,3)}\right]
\nonumber\\[0.5ex]
&\qquad=\  
\begin{tikzpicture}[baseline={([yshift=-.5ex]current bounding box.center)},inner sep=-4mm]
\node[Permute_2] (X) at (1.5*\xratio, -1.75) {$\hat{X}$};
\node[tensor_green] (U1) at (3*\xratio, -1) {$U$};
\node[tensor_green] (U2) at (3*\xratio, -2.5) {$U^{\dag}$};
\node[tensor_blue] (W1) at (1*\xratio, -3.5) {$W$};
\node[tensor_blue] (V1) at (2*\xratio, -3.5) {$V^{\top}$};
\node[tensor_blue] (W2) at (3*\xratio, -3.5) {$W$};
\node[tensor_blue] (V2) at (4*\xratio, -3.5) {$V^{\top}$};
\draw[>>-,draw=black] (1*\xratio, -0.3) -- (1*\xratio, -1.45);
\draw[>>-,draw=black] (2*\xratio, -0.3) -- (2*\xratio, -1.45);
\draw[>>-,draw=black] (3*\xratio, -0.3) -- (U1);
\draw[-<<,draw=black] (U1.south) .. controls (3*\xratio, -1.7) and (4*\xratio, -2.1) ..  (4*\xratio, -0.3);
\draw[-,draw=black]   (U2.north) .. controls (3*\xratio, -1.8) and (4*\xratio, -1.4) ..  (V2.north);
\draw[-,draw=black] (W1) -- (1*\xratio, -2.05);
\draw[-,draw=black] (V1) -- (2*\xratio, -2.05);
\draw[-,draw=black] (W2) -- (U2);
\draw[-<<,draw=black] (W1.south) .. controls (1*\xratio, -4.7) and (3*\xratio, -4.2) ..  (3*\xratio, -5.2);
\draw[-<<,draw=black] (W2.south) .. controls (3*\xratio, -4.7) and (1*\xratio, -4.2) ..  (1*\xratio, -5.2);
\draw[-<<,draw=black] (V1.south) -- (2*\xratio, -5.2);
\draw[-<<,draw=black] (V2.south) -- (4*\xratio, -5.2);
\end{tikzpicture}
\ = \ 
\begin{tikzpicture}[baseline={([yshift=-.5ex]current bounding box.center)},inner sep=-4mm]
\node[Permute_2] (X) at (1.5*\xratio, -1.75) {$\hat{X}$};
\node[UVU_6] (WV) at (0.6*\xratio, -3.5) {$WUVU^{\dag}W$};
\node[tensor_blue] (V1) at (2*\xratio, -3.5) {$V^{\top}$};
\draw[>>-,draw=black] (1*\xratio, -0.3) -- (1*\xratio, -1.45);
\draw[>>-,draw=black] (2*\xratio, -0.3) -- (2*\xratio, -1.45);
\draw[-,draw=black] (W1) -- (1*\xratio, -2.05);
\draw[-,draw=black] (V1) -- (2*\xratio, -2.05);
\draw[-<<,draw=black] (W1.south) -- (1*\xratio, -5.2);
\draw[-<<,draw=black] (V1.south) -- (2*\xratio, -5.2);
\end{tikzpicture}
\ = \Tr \left[  \hat{X} \left( WUVU^{\dag}W\otimes V^{\top}\right)  \right] . 
\end{align}
It follows that 
\begin{align}\label{eq:1matchOTOCa}
	\E\left[ \hat{D}(i,j)\hat{D}(i,k)\right] 
	&= \E \Big\lbrace \left( \Tr \!\big[ \hat{X} (WUVU^{\dag}W\otimes V^{\top})\big]\right)^2 \Big\rbrace
	\nonumber\\
	&= \Tr\left(  \E\left[ \hat{X}\otimes\hat{X}\right]  \,\left( WUVU^{\dag}W\otimes V^{\top}\otimes WUVU^{\dag}W\otimes V^{\top}\right) \right). 
\end{align}

The expression of $\E \big[\hat{X} \otimes \hat{X}\big]$ in Lemma~\ref{lem:ExpXTotimesXT} has 7 terms. 
Plugging them into Eq.~\eqref{eq:1matchOTOCa}, we get
\begin{align}
	\Tr\left[ \tilde \Delta_\ell\, (WUVU^{\dag}W\otimes V^{\top}\otimes WUVU^{\dag}W\otimes V^{\top})\right]
	=
	\begin{cases}
		\bigo{d^2}  &\ell=1,2,3,4,\\
		0           &\ell=5,6,7, 
	\end{cases}
	\label{eq:1matchOTOCb}
\end{align} 
which can be derived using tensor-network diagrams. 
For example, 
\begin{align}
	&\!\!\!\!
	\Tr\left[ \tilde \Delta_2\, (WUVU^{\dag}W\otimes V^{\top}\otimes WUVU^{\dag}W\otimes V^{\top})\right]
	\nonumber\\[0.5ex]
	&=\  \begin{tikzpicture}[baseline={([yshift=-.5ex]current bounding box.center)},inner sep=-4mm]
		\node[UVU_6] (WV1) at (1*\xratio, -1.5) {$WUVU^{\dag}W$};
		\node[UVU_6] (WV2) at (3*\xratio, -1.5) {$WUVU^{\dag}W$};
		\node[tensor_blue] (V1) at (2*\xratio, -2.5) {$V^{\top}$};
		\node[tensor_blue] (V2) at (4*\xratio, -2.5) {$V^{\top}$};
		\node[tensor_green] (U5) at (3*\xratio, -3.6) {$U$};
		\node[Permute_3] (permute_a) at (3*\xratio, -4.5) {$\sym^{(3)}$};
		\node[tensor_green] (U6) at (3*\xratio, -5.5) {$U^{\dag}$};
		\node[Permute_3] (permute_b) at (2*\xratio, -6.5) {$\sym^{(2)}$};
		\draw[-,draw=black] (2*\xratio, -2.8) .. controls (2*\xratio, -3.4) and (1.3*\xratio, -3.4) ..  (1.3*\xratio, -4.6);
		\draw[-,draw=black] (1.3*\xratio, -4.6) .. controls (1.3*\xratio, -5.3) and (2*\xratio, -5.3) ..  (2*\xratio, -4.8);
		\draw[-,draw=black] (2*\xratio, -4.2) .. controls (2*\xratio, -3.8) and (1.5*\xratio, -3.8) ..  (1.5*\xratio, -4.4);
		\draw[-,draw=black] (1.5*\xratio, -4.4) .. controls (1.5*\xratio, -5) and (2*\xratio, -5.1) ..  (2*\xratio, -5.5);
		\draw[-<<,draw=black] (2*\xratio, -5.5) -- (2*\xratio, -5.9);
		\draw[-,draw=black] (4*\xratio, -2.8) .. controls (4*\xratio, -3.4) and (4.6*\xratio, -3.4) ..  (4.6*\xratio, -4.6);
		\draw[-,draw=black] (4.6*\xratio, -4.6) .. controls (4.6*\xratio, -5.3) and (4*\xratio, -5.3) ..  (4*\xratio, -4.8);
		\draw[-,draw=black] (4*\xratio, -4.2) .. controls (4*\xratio, -3.8) and (4.5*\xratio, -3.8) ..  (4.5*\xratio, -4.4);
		\draw[-,draw=black] (4.5*\xratio, -4.4) .. controls (4.5*\xratio, -5) and (4*\xratio, -5.1) ..  (4*\xratio, -5.7);
		\draw[-<<,draw=black] (4*\xratio, -5.7) -- (4*\xratio, -7.3);
		\draw[-,draw=black] (U5) -- (3*\xratio, -4.2);
		\draw[-,draw=black] (3*\xratio, -4.8) -- (U6);
		\draw[-,draw=black] (U6) -- (3*\xratio, -6.2);
		\draw[-,draw=black] (WV1) -- (1*\xratio, -6.2);
		\draw[-,draw=black] (WV2) -- (U5);
		\draw[-<<,draw=black] (1*\xratio, -6.8) -- (1*\xratio, -7.3);
		\draw[-<<,draw=black] (3*\xratio, -6.8) -- (3*\xratio, -7.3);
		\draw[-,densely dotted,thick,draw=red] (0.6*\xratio, -3.11) -- (0.6*\xratio, -7.0);
		\draw[-,densely dotted,thick,draw=red] (0.6*\xratio, -7.0) -- (4.7*\xratio, -7.0);
		\draw[-,densely dotted,thick,draw=red] (4.7*\xratio, -7.0) -- (4.7*\xratio, -3.11);
		\draw[-,densely dotted,thick,draw=red] (4.7*\xratio, -3.11) -- (0.6*\xratio, -3.11);
		\node[red] (T) at (4.98*\xratio, -5.2) {$\tilde \Delta_2$};
		\draw[>>-,draw=black] (1*\xratio, -0.75) -- (WV1);
		\draw[>>-,draw=black] (2*\xratio, -0.75) -- (V1);
		\draw[>>-,draw=black] (3*\xratio, -0.75) -- (WV2);
		\draw[>>-,draw=black] (4*\xratio, -0.75) -- (V2);
	\end{tikzpicture}
	=
	\frac{1}{3! \cdot 2!} \sum_{\pi_1\in\mathrm{S}_3} \sum_{\pi_2\in\mathrm{S}_2} \!
	\begin{tikzpicture}[baseline={([yshift=-.5ex]current bounding box.center)},inner sep=-4mm]
		\node[UVU_6] (WV1) at (1*\xratio, -1.5) {$WUVU^{\dag}W$};
		\node[UVU_6] (WV2) at (3*\xratio, -1.5) {$WUVU^{\dag}W$};
		\node[tensor_blue] (V1) at (2*\xratio, -3.5) {$V$};
		\node[tensor_blue] (V2) at (4*\xratio, -3.5) {$V$};
		\node[tensor_green] (U5) at (3*\xratio, -3.5) {$U$};
		\node[Permute_3] (permute_a) at (3*\xratio, -4.5) {$T_{\pi_1}$};
		\node[tensor_green] (U6) at (3*\xratio, -5.5) {$U^{\dag}$};
		\node[Permute_3] (permute_b) at (2*\xratio, -6.5) {$T_{\pi_2}$};
		\draw[-,draw=black] (U5) -- (3*\xratio, -4.2);
		\draw[-,draw=black] (3*\xratio, -4.8) -- (U6);
		\draw[-,draw=black] (U6) -- (3*\xratio, -6.2);
		\draw[-,draw=black] (WV1) -- (1*\xratio, -6.2);
		\draw[-,draw=black] (WV2) -- (U5);
		\draw[-,draw=black] (V1.south) -- (2*\xratio, -4.2);
		\draw[-,draw=black] (V2.south) -- (4*\xratio, -4.2);
		\draw[-<<,draw=black] (1*\xratio, -6.8) -- (1*\xratio, -7.3);
		\draw[-<<,draw=black] (3*\xratio, -6.8) -- (3*\xratio, -7.3);
		\draw[-<<,draw=black] (2*\xratio, -4.8) -- (2*\xratio, -5.4);
		\draw[-<<,draw=black] (4*\xratio, -4.8) -- (4*\xratio, -5.4);
		\draw[>>-,draw=black] (1*\xratio, -0.75) -- (WV1);
		\draw[>>-,draw=black] (2*\xratio, -2.6) -- (V1);
		\draw[>>-,draw=black] (3*\xratio, -0.75) -- (WV2);
		\draw[>>-,draw=black] (4*\xratio, -2.6) -- (V2);
	\end{tikzpicture}
	\nonumber\\[1ex]
	&\stackrel{(a)}{=} \frac{1}{12} \left( \Tr(V^2) \Tr[(WUVU^{\dag}W)^2] + 2\Tr[(WUVU^{\dag}W)^2 U V^2 U^{\dag}] \right) 
	\nonumber\\
	&\stackrel{(b)}{=} \frac{1}{12} \left[ (\Tr(I))^2 + 2\Tr(I) \right]
	= \frac{ d^2 + 2d}{12} 
	= \bigo{d^2},  
\end{align}
where $(a)$ holds because $\tr{V}=\Tr(WUVU^{\dag}W)=0$, 
and $(b)$ holds because $(WUVU^{\dag}W)^2=V^2=I$.

Equations~\eqref{eq:1matchOTOCa}, \eqref{eq:1matchOTOCb}, and Lemma~\ref{lem:ExpXTotimesXT} together imply that 
\begin{align}\label{eq:OTOCmatch1a}
\E\left[ \hat{D}(i,j)\hat{D}(i,k)\right] 
&= \frac{2 d^2(d+1)^2(d+s)}{s^2(d+s+1)}\left( \frac{1}{\kappa_2} +\frac{s}{\kappa_3}
+\frac{s}{\kappa_3}+\frac{s(s-1)}{2\kappa_4} \right) \bigo{d^2}
=\mathcal O \left[  \left( \frac{d^2}{s^2}  +\frac{d}{s} + 1 \right) d^2 \right].  
\end{align} 
In addition, we have 
\begin{align}\label{eq:OTOCmatch1b}
 \E\left[ \hat{D}(i,j)\hat{D}(k,i)\right] 	
=\E\left[ \hat{D}(i,j)\hat{D}(i,k)\right] 	
=\mathcal O \left[  \left( \frac{d^2}{s^2}  +\frac{d}{s} + 1 \right) d^2 \right] , 
\end{align} 
where the first equality holds because $\hat{D}(k,i)=\hat{D}(i,k)$ by symmetry. 
Equations \eqref{eq:OTOCmatch1a} and \eqref{eq:OTOCmatch1b} together confirm the first statement in \lref{lem:allOTOCcovs}.

\subsubsection{Proof of the second statement in \lref{lem:allOTOCcovs}}\label{sec:OTOCCalCov2}
Suppose indices $i,j\in\{1,\dots,m\}$ are distinct. Then we have 
\begin{align}
\E\left[ \hat{D}(i,j)\hat{D}(i,j)\right] 
&= \E \Big\lbrace \left[ \Tr\left( (\hat{X}_i\otimes\hat{X}_j)(W\otimes V^{\top}\otimes W\otimes V^{\top}) T_{(1,3)}\right)\right]^2 \Big\rbrace 
\nonumber\\
&= \Tr \left[  \left(\E\left[ \hat{X}\otimes\hat{X}\right] \otimes\E\left[ \hat{X}\otimes\hat{X}\right] \right) 
             \left( W\otimes V^{\top}\otimes W\otimes V^{\top}\otimes W\otimes V^{\top}\otimes W\otimes V^{\top}\right) T_{(1,5)(3,7)}\right] .
\label{eq:2matchOTOCa}
\end{align}
Define 
\begin{align}\label{eq:define_theta}
\gamma(\mu,\nu)
:= \Tr \Big[ \left( \tilde \Delta_\mu\otimes\tilde \Delta_\nu \right) \left( W\otimes V^{\top}\otimes W\otimes V^{\top}\otimes W\otimes V^{\top}\otimes W\otimes V^{\top}\right)  T_{(1,5)(3,7)}\Big], 
\end{align}
which can be bounded via tensor-network diagrams: 
\begin{align}\label{eq:OTOCmatch2b}
\gamma(\mu,\nu)=
\begin{cases}
\bigo{d^4}  &\mu,\nu\in\{1,2,3,4\}, \\ 
0           &\text{otherwise}. 
\end{cases}
\end{align}
The proof of \eref{eq:OTOCmatch2b} is a simple analog of \eref{eq:1matchOTOCb} and is omitted. 

According to Lemma~\ref{lem:ExpXTotimesXT}, we can then expand the RHS of \eref{eq:2matchOTOCa} as 
\begin{align}\label{eq:OTOCmatch2c}
&\text{RHS of \eref{eq:2matchOTOCa}} 
\nonumber\\
&= \mathcal O \bigg(       \frac{d^4}{s^4}\gamma(1,1) + \frac{d^3}{s^3}[\gamma(1,2)+\gamma(1,3)+\gamma(2,1)+\gamma(3,1)] 
+ \frac{d}{s}[\gamma(2,4)+\gamma(3,4)+\gamma(4,2)+\gamma(4,3)]  
\nonumber\\&\qquad\quad  + \frac{d^2}{s^2}[\gamma(2,2)+\gamma(2,3)+\gamma(1,4)+\gamma(3,2)+\gamma(3,3)+\gamma(4,1)] + \gamma(4,4)
+ \frac{(d+s)^4}{s^4}\gamma(7,7) 
\nonumber\\&\qquad\quad  + \frac{d(d+s)}{s^2}[\gamma(3,5)+\gamma(3,6)+\gamma(2,5)+\gamma(2,6)+\gamma(5,2)+\gamma(5,3)+\gamma(6,2)+\gamma(6,3)] 
\nonumber\\&\qquad\quad  + \frac{d^2(d+s)}{s^3}[\gamma(1,5)+\gamma(1,6)+\gamma(5,1)+\gamma(6,1)] 
+ \frac{(d+s)}{s}[\gamma(4,5)+\gamma(4,6)+\gamma(5,4)+\gamma(6,4)]  
\nonumber\\&\qquad\quad  + \frac{(d+s)^3}{s^3}[\gamma(5,7)+\gamma(6,7)+\gamma(7,5)+\gamma(7,6)] 
+ \frac{d(d+s)^2}{s^3}[\gamma(2,7)+\gamma(3,7)+\gamma(7,2)+\gamma(7,3)]
\nonumber\\&\qquad\quad  + \frac{(d+s)^2}{s^2}[\gamma(5,5)+\gamma(5,6)+\gamma(6,5)+\gamma(6,6)+\gamma(4,7)+\gamma(7,4)] 
+ \frac{d^2(d+s)^2}{s^4}[\gamma(1,7)+\gamma(7,1)]
\bigg) 
\nonumber\\
&= \mathcal O \left[ \left(  \frac{d^2}{s^2}  +\frac{d}{s} + 1 \right) ^2 d^4 \right],  
\end{align} 
where the last equality follows from \eref{eq:OTOCmatch2b}. 
Therefore, we have 
\begin{align}\label{eq:OTOCmatch2d}
\E\left[ \hat{D}(i,j)\hat{D}(j,i)\right] 	
=\E\left[ \hat{D}(i,j)\hat{D}(i,j)\right] 	
= \mathcal O \left[ \left(  \frac{d^2}{s^2}  +\frac{d}{s} + 1 \right) ^2 d^4 \right], 
\end{align} 
where the first equality holds because $\hat{D}(j,i)=\hat{D}(i,j)$. 
This equation confirms the second statement in \lref{lem:allOTOCcovs}.
\end{widetext}

\end{appendix}


\begin{thebibliography}{83}%
	\makeatletter
	\providecommand \@ifxundefined [1]{%
		\@ifx{#1\undefined}
	}%
	\providecommand \@ifnum [1]{%
		\ifnum #1\expandafter \@firstoftwo
		\else \expandafter \@secondoftwo
		\fi
	}%
	\providecommand \@ifx [1]{%
		\ifx #1\expandafter \@firstoftwo
		\else \expandafter \@secondoftwo
		\fi
	}%
	\providecommand \natexlab [1]{#1}%
	\providecommand \enquote  [1]{``#1''}%
	\providecommand \bibnamefont  [1]{#1}%
	\providecommand \bibfnamefont [1]{#1}%
	\providecommand \citenamefont [1]{#1}%
	\providecommand \href@noop [0]{\@secondoftwo}%
	\providecommand \href [0]{\begingroup \@sanitize@url \@href}%
	\providecommand \@href[1]{\@@startlink{#1}\@@href}%
	\providecommand \@@href[1]{\endgroup#1\@@endlink}%
	\providecommand \@sanitize@url [0]{\catcode `\\12\catcode `\$12\catcode `\&12\catcode `\#12\catcode `\^12\catcode `\_12\catcode `\%12\relax}%
	\providecommand \@@startlink[1]{}%
	\providecommand \@@endlink[0]{}%
	\providecommand \url  [0]{\begingroup\@sanitize@url \@url }%
	\providecommand \@url [1]{\endgroup\@href {#1}{\urlprefix }}%
	\providecommand \urlprefix  [0]{URL }%
	\providecommand \Eprint [0]{\href }%
	\providecommand \doibase [0]{https://doi.org/}%
	\providecommand \selectlanguage [0]{\@gobble}%
	\providecommand \bibinfo  [0]{\@secondoftwo}%
	\providecommand \bibfield  [0]{\@secondoftwo}%
	\providecommand \translation [1]{[#1]}%
	\providecommand \BibitemOpen [0]{}%
	\providecommand \bibitemStop [0]{}%
	\providecommand \bibitemNoStop [0]{.\EOS\space}%
	\providecommand \EOS [0]{\spacefactor3000\relax}%
	\providecommand \BibitemShut  [1]{\csname bibitem#1\endcsname}%
	\let\auto@bib@innerbib\@empty
	\bibitem [{\citenamefont {Gebhart}\ \emph {et~al.}(2023)\citenamefont {Gebhart}, \citenamefont {Santagati}, \citenamefont {Gentile}, \citenamefont {Gauger}, \citenamefont {Craig}, \citenamefont {Ares}, \citenamefont {Banchi}, \citenamefont {Marquardt}, \citenamefont {Pezz{\`e}},\ and\ \citenamefont {Bonato}}]{gebhart2023learning}%
	\BibitemOpen
	\bibfield  {author} {\bibinfo {author} {\bibfnamefont {V.}~\bibnamefont {Gebhart}}, \bibinfo {author} {\bibfnamefont {R.}~\bibnamefont {Santagati}}, \bibinfo {author} {\bibfnamefont {A.~A.}\ \bibnamefont {Gentile}}, \bibinfo {author} {\bibfnamefont {E.~M.}\ \bibnamefont {Gauger}}, \bibinfo {author} {\bibfnamefont {D.}~\bibnamefont {Craig}}, \bibinfo {author} {\bibfnamefont {N.}~\bibnamefont {Ares}}, \bibinfo {author} {\bibfnamefont {L.}~\bibnamefont {Banchi}}, \bibinfo {author} {\bibfnamefont {F.}~\bibnamefont {Marquardt}}, \bibinfo {author} {\bibfnamefont {L.}~\bibnamefont {Pezz{\`e}}},\ and\ \bibinfo {author} {\bibfnamefont {C.}~\bibnamefont {Bonato}},\ }\bibfield  {title} {\bibinfo {title} {Learning quantum systems},\ }\href {https://www.nature.com/articles/s42254-022-00552-1} {\bibfield  {journal} {\bibinfo  {journal} {Nat. Rev. Phys.}\ }\textbf {\bibinfo {volume} {5}},\ \bibinfo {pages} {141} (\bibinfo {year} {2023})}\BibitemShut {NoStop}%
	\bibitem [{\citenamefont {Carrasco}\ \emph {et~al.}(2021)\citenamefont {Carrasco}, \citenamefont {Elben}, \citenamefont {Kokail}, \citenamefont {Kraus},\ and\ \citenamefont {Zoller}}]{PRXQuantum.2.010102}%
	\BibitemOpen
	\bibfield  {author} {\bibinfo {author} {\bibfnamefont {J.}~\bibnamefont {Carrasco}}, \bibinfo {author} {\bibfnamefont {A.}~\bibnamefont {Elben}}, \bibinfo {author} {\bibfnamefont {C.}~\bibnamefont {Kokail}}, \bibinfo {author} {\bibfnamefont {B.}~\bibnamefont {Kraus}},\ and\ \bibinfo {author} {\bibfnamefont {P.}~\bibnamefont {Zoller}},\ }\bibfield  {title} {\bibinfo {title} {Theoretical and experimental perspectives of quantum verification},\ }\href {https://doi.org/10.1103/PRXQuantum.2.010102} {\bibfield  {journal} {\bibinfo  {journal} {PRX Quantum}\ }\textbf {\bibinfo {volume} {2}},\ \bibinfo {pages} {010102} (\bibinfo {year} {2021})}\BibitemShut {NoStop}%
	\bibitem [{\citenamefont {Kliesch}\ and\ \citenamefont {Roth}(2021)}]{Kliesch2021Certification}%
	\BibitemOpen
	\bibfield  {author} {\bibinfo {author} {\bibfnamefont {M.}~\bibnamefont {Kliesch}}\ and\ \bibinfo {author} {\bibfnamefont {I.}~\bibnamefont {Roth}},\ }\bibfield  {title} {\bibinfo {title} {Theory of quantum system certification},\ }\href {https://doi.org/10.1103/PRXQuantum.2.010201} {\bibfield  {journal} {\bibinfo  {journal} {PRX Quantum}\ }\textbf {\bibinfo {volume} {2}},\ \bibinfo {pages} {010201} (\bibinfo {year} {2021})}\BibitemShut {NoStop}%
	\bibitem [{\citenamefont {Mohseni}\ and\ \citenamefont {Lidar}(2006)}]{PhysRevLett.97.170501}%
	\BibitemOpen
	\bibfield  {author} {\bibinfo {author} {\bibfnamefont {M.}~\bibnamefont {Mohseni}}\ and\ \bibinfo {author} {\bibfnamefont {D.~A.}\ \bibnamefont {Lidar}},\ }\bibfield  {title} {\bibinfo {title} {Direct characterization of quantum dynamics},\ }\href {https://doi.org/10.1103/PhysRevLett.97.170501} {\bibfield  {journal} {\bibinfo  {journal} {Phys. Rev. Lett.}\ }\textbf {\bibinfo {volume} {97}},\ \bibinfo {pages} {170501} (\bibinfo {year} {2006})}\BibitemShut {NoStop}%
	\bibitem [{\citenamefont {Caro}\ \emph {et~al.}(2023)\citenamefont {Caro}, \citenamefont {Huang}, \citenamefont {Ezzell}, \citenamefont {Gibbs}, \citenamefont {Sornborger}, \citenamefont {Cincio}, \citenamefont {Coles},\ and\ \citenamefont {Holmes}}]{caro2023out}%
	\BibitemOpen
	\bibfield  {author} {\bibinfo {author} {\bibfnamefont {M.~C.}\ \bibnamefont {Caro}}, \bibinfo {author} {\bibfnamefont {H.-Y.}\ \bibnamefont {Huang}}, \bibinfo {author} {\bibfnamefont {N.}~\bibnamefont {Ezzell}}, \bibinfo {author} {\bibfnamefont {J.}~\bibnamefont {Gibbs}}, \bibinfo {author} {\bibfnamefont {A.~T.}\ \bibnamefont {Sornborger}}, \bibinfo {author} {\bibfnamefont {L.}~\bibnamefont {Cincio}}, \bibinfo {author} {\bibfnamefont {P.~J.}\ \bibnamefont {Coles}},\ and\ \bibinfo {author} {\bibfnamefont {Z.}~\bibnamefont {Holmes}},\ }\bibfield  {title} {\bibinfo {title} {Out-of-distribution generalization for learning quantum dynamics},\ }\href {https://www.nature.com/articles/s41467-023-39381-w} {\bibfield  {journal} {\bibinfo  {journal} {Nat. Commun.}\ }\textbf {\bibinfo {volume} {14}},\ \bibinfo {pages} {3751} (\bibinfo {year} {2023})}\BibitemShut {NoStop}%
	\bibitem [{\citenamefont {Haah}\ \emph {et~al.}(2023)\citenamefont {Haah}, \citenamefont {Kothari}, \citenamefont {O’Donnell},\ and\ \citenamefont {Tang}}]{haah2023query}%
	\BibitemOpen
	\bibfield  {author} {\bibinfo {author} {\bibfnamefont {J.}~\bibnamefont {Haah}}, \bibinfo {author} {\bibfnamefont {R.}~\bibnamefont {Kothari}}, \bibinfo {author} {\bibfnamefont {R.}~\bibnamefont {O’Donnell}},\ and\ \bibinfo {author} {\bibfnamefont {E.}~\bibnamefont {Tang}},\ }\bibfield  {title} {\bibinfo {title} {Query-optimal estimation of unitary channels in diamond distance},\ }in\ \href {https://doi.org/10.1109/FOCS57990.2023.00028} {\emph {\bibinfo {booktitle} {2023 IEEE 64th Annual Symposium on Foundations of Computer Science (FOCS)}}}\ (\bibinfo {organization} {IEEE},\ \bibinfo {year} {2023})\ pp.\ \bibinfo {pages} {363--390}\BibitemShut {NoStop}%
	\bibitem [{\citenamefont {Aaronson}(2018)}]{aaronson2018shadow}%
	\BibitemOpen
	\bibfield  {author} {\bibinfo {author} {\bibfnamefont {S.}~\bibnamefont {Aaronson}},\ }\bibfield  {title} {\bibinfo {title} {Shadow tomography of quantum states},\ }in\ \href {https://doi.org/10.1145/3188745.3188802} {\emph {\bibinfo {booktitle} {Proceedings of the 50th annual ACM SIGACT symposium on theory of computing}}}\ (\bibinfo {year} {2018})\ pp.\ \bibinfo {pages} {325--338}\BibitemShut {NoStop}%
	\bibitem [{\citenamefont {B{\u{a}}descu}\ and\ \citenamefont {O'Donnell}(2021)}]{buadescu2021improved}%
	\BibitemOpen
	\bibfield  {author} {\bibinfo {author} {\bibfnamefont {C.}~\bibnamefont {B{\u{a}}descu}}\ and\ \bibinfo {author} {\bibfnamefont {R.}~\bibnamefont {O'Donnell}},\ }\bibfield  {title} {\bibinfo {title} {Improved quantum data analysis},\ }in\ \href {https://doi.org/10.1145/3406325.3451109} {\emph {\bibinfo {booktitle} {Proceedings of the 53rd Annual ACM SIGACT Symposium on Theory of Computing}}}\ (\bibinfo {year} {2021})\ pp.\ \bibinfo {pages} {1398--1411}\BibitemShut {NoStop}%
	\bibitem [{\citenamefont {Huang}\ \emph {et~al.}(2020)\citenamefont {Huang}, \citenamefont {Kueng},\ and\ \citenamefont {Preskill}}]{HKPshadow20}%
	\BibitemOpen
	\bibfield  {author} {\bibinfo {author} {\bibfnamefont {H.-Y.}\ \bibnamefont {Huang}}, \bibinfo {author} {\bibfnamefont {R.}~\bibnamefont {Kueng}},\ and\ \bibinfo {author} {\bibfnamefont {J.}~\bibnamefont {Preskill}},\ }\bibfield  {title} {\bibinfo {title} {Predicting many properties of a quantum system from very few measurements},\ }\href {https://doi.org/10.1038/s41567-020-0932-7} {\bibfield  {journal} {\bibinfo  {journal} {Nat. Phys.}\ }\textbf {\bibinfo {volume} {16}},\ \bibinfo {pages} {1050} (\bibinfo {year} {2020})}\BibitemShut {NoStop}%
	\bibitem [{\citenamefont {Bertoni}\ \emph {et~al.}(2024)\citenamefont {Bertoni}, \citenamefont {Haferkamp}, \citenamefont {Hinsche}, \citenamefont {Ioannou}, \citenamefont {Eisert},\ and\ \citenamefont {Pashayan}}]{PhysRevLett.133.020602}%
	\BibitemOpen
	\bibfield  {author} {\bibinfo {author} {\bibfnamefont {C.}~\bibnamefont {Bertoni}}, \bibinfo {author} {\bibfnamefont {J.}~\bibnamefont {Haferkamp}}, \bibinfo {author} {\bibfnamefont {M.}~\bibnamefont {Hinsche}}, \bibinfo {author} {\bibfnamefont {M.}~\bibnamefont {Ioannou}}, \bibinfo {author} {\bibfnamefont {J.}~\bibnamefont {Eisert}},\ and\ \bibinfo {author} {\bibfnamefont {H.}~\bibnamefont {Pashayan}},\ }\bibfield  {title} {\bibinfo {title} {Shallow shadows: Expectation estimation using low-depth random Clifford circuits},\ }\href {https://doi.org/10.1103/PhysRevLett.133.020602} {\bibfield  {journal} {\bibinfo  {journal} {Phys. Rev. Lett.}\ }\textbf {\bibinfo {volume} {133}},\ \bibinfo {pages} {020602} (\bibinfo {year} {2024})}\BibitemShut {NoStop}%
	\bibitem [{\citenamefont {Grier}\ \emph {et~al.}(2024)\citenamefont {Grier}, \citenamefont {Pashayan},\ and\ \citenamefont {Schaeffer}}]{Grier22}%
	\BibitemOpen
	\bibfield  {author} {\bibinfo {author} {\bibfnamefont {D.}~\bibnamefont {Grier}}, \bibinfo {author} {\bibfnamefont {H.}~\bibnamefont {Pashayan}},\ and\ \bibinfo {author} {\bibfnamefont {L.}~\bibnamefont {Schaeffer}},\ }\bibfield  {title} {\bibinfo {title} {Sample-optimal classical shadows for pure states},\ }\href {https://doi.org/10.22331/q-2024-06-17-1373} {\bibfield  {journal} {\bibinfo  {journal} {{Quantum}}\ }\textbf {\bibinfo {volume} {8}},\ \bibinfo {pages} {1373} (\bibinfo {year} {2024})}\BibitemShut {NoStop}%
	\bibitem [{\citenamefont {Hu}\ \emph {et~al.}(2023)\citenamefont {Hu}, \citenamefont {Choi},\ and\ \citenamefont {You}}]{PhysRevResearch.5.023027}%
	\BibitemOpen
	\bibfield  {author} {\bibinfo {author} {\bibfnamefont {H.-Y.}\ \bibnamefont {Hu}}, \bibinfo {author} {\bibfnamefont {S.}~\bibnamefont {Choi}},\ and\ \bibinfo {author} {\bibfnamefont {Y.-Z.}\ \bibnamefont {You}},\ }\bibfield  {title} {\bibinfo {title} {Classical shadow tomography with locally scrambled quantum dynamics},\ }\href {https://doi.org/10.1103/PhysRevResearch.5.023027} {\bibfield  {journal} {\bibinfo  {journal} {Phys. Rev. Res.}\ }\textbf {\bibinfo {volume} {5}},\ \bibinfo {pages} {023027} (\bibinfo {year} {2023})}\BibitemShut {NoStop}%
	\bibitem [{\citenamefont {Vermersch}\ \emph {et~al.}(2024)\citenamefont {Vermersch}, \citenamefont {Rath}, \citenamefont {Sundar}, \citenamefont {Branciard}, \citenamefont {Preskill},\ and\ \citenamefont {Elben}}]{PRXQuantum.5.010352}%
	\BibitemOpen
	\bibfield  {author} {\bibinfo {author} {\bibfnamefont {B.}~\bibnamefont {Vermersch}}, \bibinfo {author} {\bibfnamefont {A.}~\bibnamefont {Rath}}, \bibinfo {author} {\bibfnamefont {B.}~\bibnamefont {Sundar}}, \bibinfo {author} {\bibfnamefont {C.}~\bibnamefont {Branciard}}, \bibinfo {author} {\bibfnamefont {J.}~\bibnamefont {Preskill}},\ and\ \bibinfo {author} {\bibfnamefont {A.}~\bibnamefont {Elben}},\ }\bibfield  {title} {\bibinfo {title} {Enhanced estimation of quantum properties with common randomized measurements},\ }\href {https://doi.org/10.1103/PRXQuantum.5.010352} {\bibfield  {journal} {\bibinfo  {journal} {PRX Quantum}\ }\textbf {\bibinfo {volume} {5}},\ \bibinfo {pages} {010352} (\bibinfo {year} {2024})}\BibitemShut {NoStop}%
	\bibitem [{\citenamefont {Helsen}\ and\ \citenamefont {Walter}(2023)}]{PhysRevLett.131.240602}%
	\BibitemOpen
	\bibfield  {author} {\bibinfo {author} {\bibfnamefont {J.}~\bibnamefont {Helsen}}\ and\ \bibinfo {author} {\bibfnamefont {M.}~\bibnamefont {Walter}},\ }\bibfield  {title} {\bibinfo {title} {Thrifty shadow estimation: Reusing quantum circuits and bounding tails},\ }\href {https://doi.org/10.1103/PhysRevLett.131.240602} {\bibfield  {journal} {\bibinfo  {journal} {Phys. Rev. Lett.}\ }\textbf {\bibinfo {volume} {131}},\ \bibinfo {pages} {240602} (\bibinfo {year} {2023})}\BibitemShut {NoStop}%
	\bibitem [{\citenamefont {Zhou}\ and\ \citenamefont {Liu}(2023)}]{Zhou2023perform}%
	\BibitemOpen
	\bibfield  {author} {\bibinfo {author} {\bibfnamefont {Y.}~\bibnamefont {Zhou}}\ and\ \bibinfo {author} {\bibfnamefont {Q.}~\bibnamefont {Liu}},\ }\bibfield  {title} {\bibinfo {title} {Performance analysis of multi-shot shadow estimation},\ }\href {https://doi.org/10.22331/q-2023-06-29-1044} {\bibfield  {journal} {\bibinfo  {journal} {{Quantum}}\ }\textbf {\bibinfo {volume} {7}},\ \bibinfo {pages} {1044} (\bibinfo {year} {2023})}\BibitemShut {NoStop}%
	\bibitem [{\citenamefont {Arienzo}\ \emph {et~al.}(2023)\citenamefont {Arienzo}, \citenamefont {Heinrich}, \citenamefont {Roth},\ and\ \citenamefont {Kliesch}}]{arienzo2023closed}%
	\BibitemOpen
	\bibfield  {author} {\bibinfo {author} {\bibfnamefont {M.}~\bibnamefont {Arienzo}}, \bibinfo {author} {\bibfnamefont {M.}~\bibnamefont {Heinrich}}, \bibinfo {author} {\bibfnamefont {I.}~\bibnamefont {Roth}},\ and\ \bibinfo {author} {\bibfnamefont {M.}~\bibnamefont {Kliesch}},\ }\bibfield  {title} {\bibinfo {title} {Closed-form analytic expressions for shadow estimation with brickwork circuits},\ }\href {https://doi.org/10.26421/QIC23.11-12-5} {\bibfield  {journal} {\bibinfo  {journal} {Quant. Inf. Comp.}\ }\textbf {\bibinfo {volume} {23}},\ \bibinfo {pages} {961} (\bibinfo {year} {2023})}\BibitemShut {NoStop}%
	\bibitem [{\citenamefont {Zhou}\ and\ \citenamefont {Liu}(2024)}]{zhou2024hybrid}%
	\BibitemOpen
	\bibfield  {author} {\bibinfo {author} {\bibfnamefont {Y.}~\bibnamefont {Zhou}}\ and\ \bibinfo {author} {\bibfnamefont {Z.}~\bibnamefont {Liu}},\ }\bibfield  {title} {\bibinfo {title} {A hybrid framework for estimating nonlinear functions of quantum states},\ }\href {https://www.nature.com/articles/s41534-024-00846-5} {\bibfield  {journal} {\bibinfo  {journal} {npj Quantum Inf.}\ }\textbf {\bibinfo {volume} {10}},\ \bibinfo {pages} {62} (\bibinfo {year} {2024})}\BibitemShut {NoStop}%
	\bibitem [{\citenamefont {Zhang}\ \emph {et~al.}(2024)\citenamefont {Zhang}, \citenamefont {Liu},\ and\ \citenamefont {Zhou}}]{zhang2024minimal}%
	\BibitemOpen
	\bibfield  {author} {\bibinfo {author} {\bibfnamefont {Q.}~\bibnamefont {Zhang}}, \bibinfo {author} {\bibfnamefont {Q.}~\bibnamefont {Liu}},\ and\ \bibinfo {author} {\bibfnamefont {Y.}~\bibnamefont {Zhou}},\ }\bibfield  {title} {\bibinfo {title} {Minimal-Clifford shadow estimation by mutually unbiased bases},\ }\href {https://journals.aps.org/prapplied/abstract/10.1103/PhysRevApplied.21.064001} {\bibfield  {journal} {\bibinfo  {journal} {Phys. Rev. Appl.}\ }\textbf {\bibinfo {volume} {21}},\ \bibinfo {pages} {064001} (\bibinfo {year} {2024})}\BibitemShut {NoStop}%
	\bibitem [{\citenamefont {Liu}\ \emph {et~al.}(2024)\citenamefont {Liu}, \citenamefont {Li}, \citenamefont {Yuan}, \citenamefont {Zhu},\ and\ \citenamefont {Zhou}}]{liu2024auxiliary}%
	\BibitemOpen
	\bibfield  {author} {\bibinfo {author} {\bibfnamefont {Q.}~\bibnamefont {Liu}}, \bibinfo {author} {\bibfnamefont {Z.}~\bibnamefont {Li}}, \bibinfo {author} {\bibfnamefont {X.}~\bibnamefont {Yuan}}, \bibinfo {author} {\bibfnamefont {H.}~\bibnamefont {Zhu}},\ and\ \bibinfo {author} {\bibfnamefont {Y.}~\bibnamefont {Zhou}},\ }\bibfield  {title} {\bibinfo {title} {Auxiliary-free replica shadow estimation},\ }\href {https://arxiv.org/abs/2407.20865} {\bibfield  {journal} {\bibinfo  {journal} {arXiv:2407.20865}\ } (\bibinfo {year} {2024})}\BibitemShut {NoStop}%
	\bibitem [{\citenamefont {Elben}\ \emph {et~al.}(2023)\citenamefont {Elben}, \citenamefont {Flammia}, \citenamefont {Huang}, \citenamefont {Kueng}, \citenamefont {Preskill}, \citenamefont {Vermersch},\ and\ \citenamefont {Zoller}}]{elben2023randomized}%
	\BibitemOpen
	\bibfield  {author} {\bibinfo {author} {\bibfnamefont {A.}~\bibnamefont {Elben}}, \bibinfo {author} {\bibfnamefont {S.~T.}\ \bibnamefont {Flammia}}, \bibinfo {author} {\bibfnamefont {H.-Y.}\ \bibnamefont {Huang}}, \bibinfo {author} {\bibfnamefont {R.}~\bibnamefont {Kueng}}, \bibinfo {author} {\bibfnamefont {J.}~\bibnamefont {Preskill}}, \bibinfo {author} {\bibfnamefont {B.}~\bibnamefont {Vermersch}},\ and\ \bibinfo {author} {\bibfnamefont {P.}~\bibnamefont {Zoller}},\ }\bibfield  {title} {\bibinfo {title} {The randomized measurement toolbox},\ }\href {https://www.nature.com/articles/s42254-022-00535-2} {\bibfield  {journal} {\bibinfo  {journal} {Nat. Rev. Phys.}\ }\textbf {\bibinfo {volume} {5}},\ \bibinfo {pages} {9} (\bibinfo {year} {2023})}\BibitemShut {NoStop}%
	\bibitem [{\citenamefont {Kunjummen}\ \emph {et~al.}(2023)\citenamefont {Kunjummen}, \citenamefont {Tran}, \citenamefont {Carney},\ and\ \citenamefont {Taylor}}]{Kunjummen23}%
	\BibitemOpen
	\bibfield  {author} {\bibinfo {author} {\bibfnamefont {J.}~\bibnamefont {Kunjummen}}, \bibinfo {author} {\bibfnamefont {M.~C.}\ \bibnamefont {Tran}}, \bibinfo {author} {\bibfnamefont {D.}~\bibnamefont {Carney}},\ and\ \bibinfo {author} {\bibfnamefont {J.~M.}\ \bibnamefont {Taylor}},\ }\bibfield  {title} {\bibinfo {title} {Shadow process tomography of quantum channels},\ }\href {https://doi.org/10.1103/PhysRevA.107.042403} {\bibfield  {journal} {\bibinfo  {journal} {Phys. Rev. A}\ }\textbf {\bibinfo {volume} {107}},\ \bibinfo {pages} {042403} (\bibinfo {year} {2023})}\BibitemShut {NoStop}%
	\bibitem [{\citenamefont {Levy}\ \emph {et~al.}(2024)\citenamefont {Levy}, \citenamefont {Luo},\ and\ \citenamefont {Clark}}]{PhysRevResearch.6.013029}%
	\BibitemOpen
	\bibfield  {author} {\bibinfo {author} {\bibfnamefont {R.}~\bibnamefont {Levy}}, \bibinfo {author} {\bibfnamefont {D.}~\bibnamefont {Luo}},\ and\ \bibinfo {author} {\bibfnamefont {B.~K.}\ \bibnamefont {Clark}},\ }\bibfield  {title} {\bibinfo {title} {Classical shadows for quantum process tomography on near-term quantum computers},\ }\href {https://doi.org/10.1103/PhysRevResearch.6.013029} {\bibfield  {journal} {\bibinfo  {journal} {Phys. Rev. Res.}\ }\textbf {\bibinfo {volume} {6}},\ \bibinfo {pages} {013029} (\bibinfo {year} {2024})}\BibitemShut {NoStop}%
	\bibitem [{\citenamefont {Huang}\ \emph {et~al.}(2023{\natexlab{a}})\citenamefont {Huang}, \citenamefont {Chen},\ and\ \citenamefont {Preskill}}]{HCP23}%
	\BibitemOpen
	\bibfield  {author} {\bibinfo {author} {\bibfnamefont {H.-Y.}\ \bibnamefont {Huang}}, \bibinfo {author} {\bibfnamefont {S.}~\bibnamefont {Chen}},\ and\ \bibinfo {author} {\bibfnamefont {J.}~\bibnamefont {Preskill}},\ }\bibfield  {title} {\bibinfo {title} {Learning to predict arbitrary quantum processes},\ }\href {https://doi.org/10.1103/PRXQuantum.4.040337} {\bibfield  {journal} {\bibinfo  {journal} {PRX Quantum}\ }\textbf {\bibinfo {volume} {4}},\ \bibinfo {pages} {040337} (\bibinfo {year} {2023}{\natexlab{a}})}\BibitemShut {NoStop}%
	\bibitem [{\citenamefont {Caro}(2024)}]{caro2022learning}%
	\BibitemOpen
	\bibfield  {author} {\bibinfo {author} {\bibfnamefont {M.~C.}\ \bibnamefont {Caro}},\ }\bibfield  {title} {\bibinfo {title} {Learning quantum processes and {H}amiltonians via the {P}auli transfer matrix},\ }\href {https://doi.org/10.1145/3670418} {\bibfield  {journal} {\bibinfo  {journal} {ACM Trans. Quantum Comput.}\ }\textbf {\bibinfo {volume} {5}},\ \bibinfo {pages} {1} (\bibinfo {year} {2024})}\BibitemShut {NoStop}%
	\bibitem [{\citenamefont {Zhao}\ \emph {et~al.}(2024)\citenamefont {Zhao}, \citenamefont {Lewis}, \citenamefont {Kannan}, \citenamefont {Quek}, \citenamefont {Huang},\ and\ \citenamefont {Caro}}]{zhao2023learning}%
	\BibitemOpen
	\bibfield  {author} {\bibinfo {author} {\bibfnamefont {H.}~\bibnamefont {Zhao}}, \bibinfo {author} {\bibfnamefont {L.}~\bibnamefont {Lewis}}, \bibinfo {author} {\bibfnamefont {I.}~\bibnamefont {Kannan}}, \bibinfo {author} {\bibfnamefont {Y.}~\bibnamefont {Quek}}, \bibinfo {author} {\bibfnamefont {H.-Y.}\ \bibnamefont {Huang}},\ and\ \bibinfo {author} {\bibfnamefont {M.~C.}\ \bibnamefont {Caro}},\ }\bibfield  {title} {\bibinfo {title} {Learning quantum states and unitaries of bounded gate complexity},\ }\href {https://doi.org/10.1103/PRXQuantum.5.040306} {\bibfield  {journal} {\bibinfo  {journal} {PRX Quantum}\ }\textbf {\bibinfo {volume} {5}},\ \bibinfo {pages} {040306} (\bibinfo {year} {2024})}\BibitemShut {NoStop}%
	\bibitem [{\citenamefont {Huang}\ \emph {et~al.}(2024)\citenamefont {Huang}, \citenamefont {Liu}, \citenamefont {Broughton}, \citenamefont {Kim}, \citenamefont {Anshu}, \citenamefont {Landau},\ and\ \citenamefont {McClean}}]{huang2024learning}%
	\BibitemOpen
	\bibfield  {author} {\bibinfo {author} {\bibfnamefont {H.-Y.}\ \bibnamefont {Huang}}, \bibinfo {author} {\bibfnamefont {Y.}~\bibnamefont {Liu}}, \bibinfo {author} {\bibfnamefont {M.}~\bibnamefont {Broughton}}, \bibinfo {author} {\bibfnamefont {I.}~\bibnamefont {Kim}}, \bibinfo {author} {\bibfnamefont {A.}~\bibnamefont {Anshu}}, \bibinfo {author} {\bibfnamefont {Z.}~\bibnamefont {Landau}},\ and\ \bibinfo {author} {\bibfnamefont {J.~R.}\ \bibnamefont {McClean}},\ }\bibfield  {title} {\bibinfo {title} {Learning shallow quantum circuits},\ }in\ \href {https://doi.org/10.1145/3618260.3649722} {\emph {\bibinfo {booktitle} {Proceedings of the 56th Annual ACM Symposium on Theory of Computing}}}\ (\bibinfo {year} {2024})\ pp.\ \bibinfo {pages} {1343--1351}\BibitemShut {NoStop}%
	\bibitem [{\citenamefont {Huang}\ \emph {et~al.}(2021)\citenamefont {Huang}, \citenamefont {Kueng},\ and\ \citenamefont {Preskill}}]{HKPITB21}%
	\BibitemOpen
	\bibfield  {author} {\bibinfo {author} {\bibfnamefont {H.-Y.}\ \bibnamefont {Huang}}, \bibinfo {author} {\bibfnamefont {R.}~\bibnamefont {Kueng}},\ and\ \bibinfo {author} {\bibfnamefont {J.}~\bibnamefont {Preskill}},\ }\bibfield  {title} {\bibinfo {title} {Information-theoretic bounds on quantum advantage in machine learning},\ }\href {https://doi.org/10.1103/PhysRevLett.126.190505} {\bibfield  {journal} {\bibinfo  {journal} {Phys. Rev. Lett.}\ }\textbf {\bibinfo {volume} {126}},\ \bibinfo {pages} {190505} (\bibinfo {year} {2021})}\BibitemShut {NoStop}%
	\bibitem [{\citenamefont {Caro}\ \emph {et~al.}(2022)\citenamefont {Caro}, \citenamefont {Huang}, \citenamefont {Cerezo}, \citenamefont {Sharma}, \citenamefont {Sornborger}, \citenamefont {Cincio},\ and\ \citenamefont {Coles}}]{caro2022generalization}%
	\BibitemOpen
	\bibfield  {author} {\bibinfo {author} {\bibfnamefont {M.~C.}\ \bibnamefont {Caro}}, \bibinfo {author} {\bibfnamefont {H.-Y.}\ \bibnamefont {Huang}}, \bibinfo {author} {\bibfnamefont {M.}~\bibnamefont {Cerezo}}, \bibinfo {author} {\bibfnamefont {K.}~\bibnamefont {Sharma}}, \bibinfo {author} {\bibfnamefont {A.}~\bibnamefont {Sornborger}}, \bibinfo {author} {\bibfnamefont {L.}~\bibnamefont {Cincio}},\ and\ \bibinfo {author} {\bibfnamefont {P.~J.}\ \bibnamefont {Coles}},\ }\bibfield  {title} {\bibinfo {title} {Generalization in quantum machine learning from few training data},\ }\href {https://doi.org/10.1038/s41467-022-32550-3} {\bibfield  {journal} {\bibinfo  {journal} {Nat. Commun.}\ }\textbf {\bibinfo {volume} {13}},\ \bibinfo {pages} {4919} (\bibinfo {year} {2022})}\BibitemShut {NoStop}%
	\bibitem [{\citenamefont {Huang}\ \emph {et~al.}(2022)\citenamefont {Huang}, \citenamefont {Broughton}, \citenamefont {Cotler}, \citenamefont {Chen}, \citenamefont {Li}, \citenamefont {Mohseni}, \citenamefont {Neven}, \citenamefont {Babbush}, \citenamefont {Kueng}, \citenamefont {Preskill} \emph {et~al.}}]{huang2022quantum}%
	\BibitemOpen
	\bibfield  {author} {\bibinfo {author} {\bibfnamefont {H.-Y.}\ \bibnamefont {Huang}}, \bibinfo {author} {\bibfnamefont {M.}~\bibnamefont {Broughton}}, \bibinfo {author} {\bibfnamefont {J.}~\bibnamefont {Cotler}}, \bibinfo {author} {\bibfnamefont {S.}~\bibnamefont {Chen}}, \bibinfo {author} {\bibfnamefont {J.}~\bibnamefont {Li}}, \bibinfo {author} {\bibfnamefont {M.}~\bibnamefont {Mohseni}}, \bibinfo {author} {\bibfnamefont {H.}~\bibnamefont {Neven}}, \bibinfo {author} {\bibfnamefont {R.}~\bibnamefont {Babbush}}, \bibinfo {author} {\bibfnamefont {R.}~\bibnamefont {Kueng}}, \bibinfo {author} {\bibfnamefont {J.}~\bibnamefont {Preskill}}, \emph {et~al.},\ }\bibfield  {title} {\bibinfo {title} {Quantum advantage in learning from experiments},\ }\href {https://doi.org/10.1126/science.abn7293} {\bibfield  {journal} {\bibinfo  {journal} {Science}\ }\textbf {\bibinfo {volume} {376}},\ \bibinfo {pages} {1182} (\bibinfo {year} {2022})}\BibitemShut {NoStop}%
	\bibitem [{\citenamefont {Surawy-Stepney}\ \emph {et~al.}(2022)\citenamefont {Surawy-Stepney}, \citenamefont {Kahn}, \citenamefont {Kueng},\ and\ \citenamefont {Guta}}]{surawy2022projected}%
	\BibitemOpen
	\bibfield  {author} {\bibinfo {author} {\bibfnamefont {T.}~\bibnamefont {Surawy-Stepney}}, \bibinfo {author} {\bibfnamefont {J.}~\bibnamefont {Kahn}}, \bibinfo {author} {\bibfnamefont {R.}~\bibnamefont {Kueng}},\ and\ \bibinfo {author} {\bibfnamefont {M.}~\bibnamefont {Guta}},\ }\bibfield  {title} {\bibinfo {title} {Projected least-squares quantum process tomography},\ }\href {https://doi.org/10.22331/q-2022-10-20-844} {\bibfield  {journal} {\bibinfo  {journal} {Quantum}\ }\textbf {\bibinfo {volume} {6}},\ \bibinfo {pages} {844} (\bibinfo {year} {2022})}\BibitemShut {NoStop}%
	\bibitem [{\citenamefont {Biamonte}\ \emph {et~al.}(2017)\citenamefont {Biamonte}, \citenamefont {Wittek}, \citenamefont {Pancotti}, \citenamefont {Rebentrost}, \citenamefont {Wiebe},\ and\ \citenamefont {Lloyd}}]{biamonte2017quantum}%
	\BibitemOpen
	\bibfield  {author} {\bibinfo {author} {\bibfnamefont {J.}~\bibnamefont {Biamonte}}, \bibinfo {author} {\bibfnamefont {P.}~\bibnamefont {Wittek}}, \bibinfo {author} {\bibfnamefont {N.}~\bibnamefont {Pancotti}}, \bibinfo {author} {\bibfnamefont {P.}~\bibnamefont {Rebentrost}}, \bibinfo {author} {\bibfnamefont {N.}~\bibnamefont {Wiebe}},\ and\ \bibinfo {author} {\bibfnamefont {S.}~\bibnamefont {Lloyd}},\ }\bibfield  {title} {\bibinfo {title} {Quantum machine learning},\ }\href {https://doi.org/10.1038/nature23474} {\bibfield  {journal} {\bibinfo  {journal} {Nature}\ }\textbf {\bibinfo {volume} {549}},\ \bibinfo {pages} {195} (\bibinfo {year} {2017})}\BibitemShut {NoStop}%
	\bibitem [{\citenamefont {Schuld}\ and\ \citenamefont {Killoran}(2019)}]{PhysRevLett.122.040504}%
	\BibitemOpen
	\bibfield  {author} {\bibinfo {author} {\bibfnamefont {M.}~\bibnamefont {Schuld}}\ and\ \bibinfo {author} {\bibfnamefont {N.}~\bibnamefont {Killoran}},\ }\bibfield  {title} {\bibinfo {title} {Quantum machine learning in feature hilbert spaces},\ }\href {https://doi.org/10.1103/PhysRevLett.122.040504} {\bibfield  {journal} {\bibinfo  {journal} {Phys. Rev. Lett.}\ }\textbf {\bibinfo {volume} {122}},\ \bibinfo {pages} {040504} (\bibinfo {year} {2019})}\BibitemShut {NoStop}%
	\bibitem [{\citenamefont {Cerezo}\ \emph {et~al.}(2021)\citenamefont {Cerezo}, \citenamefont {Arrasmith}, \citenamefont {Babbush}, \citenamefont {Benjamin}, \citenamefont {Endo}, \citenamefont {Fujii}, \citenamefont {McClean}, \citenamefont {Mitarai}, \citenamefont {Yuan}, \citenamefont {Cincio} \emph {et~al.}}]{cerezo2021variational}%
	\BibitemOpen
	\bibfield  {author} {\bibinfo {author} {\bibfnamefont {M.}~\bibnamefont {Cerezo}}, \bibinfo {author} {\bibfnamefont {A.}~\bibnamefont {Arrasmith}}, \bibinfo {author} {\bibfnamefont {R.}~\bibnamefont {Babbush}}, \bibinfo {author} {\bibfnamefont {S.~C.}\ \bibnamefont {Benjamin}}, \bibinfo {author} {\bibfnamefont {S.}~\bibnamefont {Endo}}, \bibinfo {author} {\bibfnamefont {K.}~\bibnamefont {Fujii}}, \bibinfo {author} {\bibfnamefont {J.~R.}\ \bibnamefont {McClean}}, \bibinfo {author} {\bibfnamefont {K.}~\bibnamefont {Mitarai}}, \bibinfo {author} {\bibfnamefont {X.}~\bibnamefont {Yuan}}, \bibinfo {author} {\bibfnamefont {L.}~\bibnamefont {Cincio}}, \emph {et~al.},\ }\bibfield  {title} {\bibinfo {title} {Variational quantum algorithms},\ }\href {https://doi.org/10.1038/s42254-021-00348-9} {\bibfield  {journal} {\bibinfo  {journal} {Nat. Rev. Phys.}\ }\textbf {\bibinfo {volume} {3}},\ \bibinfo {pages} {625} (\bibinfo {year} {2021})}\BibitemShut {NoStop}%
	\bibitem [{\citenamefont {Kokail}\ \emph {et~al.}(2019)\citenamefont {Kokail}, \citenamefont {Maier}, \citenamefont {van Bijnen}, \citenamefont {Brydges}, \citenamefont {Joshi}, \citenamefont {Jurcevic}, \citenamefont {Muschik}, \citenamefont {Silvi}, \citenamefont {Blatt}, \citenamefont {Roos} \emph {et~al.}}]{kokail2019self}%
	\BibitemOpen
	\bibfield  {author} {\bibinfo {author} {\bibfnamefont {C.}~\bibnamefont {Kokail}}, \bibinfo {author} {\bibfnamefont {C.}~\bibnamefont {Maier}}, \bibinfo {author} {\bibfnamefont {R.}~\bibnamefont {van Bijnen}}, \bibinfo {author} {\bibfnamefont {T.}~\bibnamefont {Brydges}}, \bibinfo {author} {\bibfnamefont {M.~K.}\ \bibnamefont {Joshi}}, \bibinfo {author} {\bibfnamefont {P.}~\bibnamefont {Jurcevic}}, \bibinfo {author} {\bibfnamefont {C.~A.}\ \bibnamefont {Muschik}}, \bibinfo {author} {\bibfnamefont {P.}~\bibnamefont {Silvi}}, \bibinfo {author} {\bibfnamefont {R.}~\bibnamefont {Blatt}}, \bibinfo {author} {\bibfnamefont {C.~F.}\ \bibnamefont {Roos}}, \emph {et~al.},\ }\bibfield  {title} {\bibinfo {title} {Self-verifying variational quantum simulation of lattice models},\ }\href {https://doi.org/10.1038/s41586-019-1177-4} {\bibfield  {journal} {\bibinfo  {journal} {Nature}\ }\textbf {\bibinfo {volume} {569}},\ \bibinfo {pages} {355} (\bibinfo {year} {2019})}\BibitemShut {NoStop}%
	\bibitem [{\citenamefont {Gibbs}\ \emph {et~al.}(2024)\citenamefont {Gibbs}, \citenamefont {Holmes}, \citenamefont {Caro}, \citenamefont {Ezzell}, \citenamefont {Huang}, \citenamefont {Cincio}, \citenamefont {Sornborger},\ and\ \citenamefont {Coles}}]{PhysRevResearch.6.013241}%
	\BibitemOpen
	\bibfield  {author} {\bibinfo {author} {\bibfnamefont {J.}~\bibnamefont {Gibbs}}, \bibinfo {author} {\bibfnamefont {Z.}~\bibnamefont {Holmes}}, \bibinfo {author} {\bibfnamefont {M.~C.}\ \bibnamefont {Caro}}, \bibinfo {author} {\bibfnamefont {N.}~\bibnamefont {Ezzell}}, \bibinfo {author} {\bibfnamefont {H.-Y.}\ \bibnamefont {Huang}}, \bibinfo {author} {\bibfnamefont {L.}~\bibnamefont {Cincio}}, \bibinfo {author} {\bibfnamefont {A.~T.}\ \bibnamefont {Sornborger}},\ and\ \bibinfo {author} {\bibfnamefont {P.~J.}\ \bibnamefont {Coles}},\ }\bibfield  {title} {\bibinfo {title} {Dynamical simulation via quantum machine learning with provable generalization},\ }\href {https://doi.org/10.1103/PhysRevResearch.6.013241} {\bibfield  {journal} {\bibinfo  {journal} {Phys. Rev. Res.}\ }\textbf {\bibinfo {volume} {6}},\ \bibinfo {pages} {013241} (\bibinfo {year} {2024})}\BibitemShut {NoStop}%
	\bibitem [{\citenamefont {Stilck~Fran{\c{c}}a}\ \emph {et~al.}(2024)\citenamefont {Stilck~Fran{\c{c}}a}, \citenamefont {Markovich}, \citenamefont {Dobrovitski}, \citenamefont {Werner},\ and\ \citenamefont {Borregaard}}]{stilck2024efficient}%
	\BibitemOpen
	\bibfield  {author} {\bibinfo {author} {\bibfnamefont {D.}~\bibnamefont {Stilck~Fran{\c{c}}a}}, \bibinfo {author} {\bibfnamefont {L.~A.}\ \bibnamefont {Markovich}}, \bibinfo {author} {\bibfnamefont {V.}~\bibnamefont {Dobrovitski}}, \bibinfo {author} {\bibfnamefont {A.~H.}\ \bibnamefont {Werner}},\ and\ \bibinfo {author} {\bibfnamefont {J.}~\bibnamefont {Borregaard}},\ }\bibfield  {title} {\bibinfo {title} {Efficient and robust estimation of many-qubit {H}amiltonians},\ }\href {https://doi.org/10.1038/s41467-023-44012-5} {\bibfield  {journal} {\bibinfo  {journal} {Nat. Commun.}\ }\textbf {\bibinfo {volume} {15}},\ \bibinfo {pages} {311} (\bibinfo {year} {2024})}\BibitemShut {NoStop}%
	\bibitem [{\citenamefont {Gu}\ \emph {et~al.}(2024)\citenamefont {Gu}, \citenamefont {Cincio},\ and\ \citenamefont {Coles}}]{gu2024practical}%
	\BibitemOpen
	\bibfield  {author} {\bibinfo {author} {\bibfnamefont {A.}~\bibnamefont {Gu}}, \bibinfo {author} {\bibfnamefont {L.}~\bibnamefont {Cincio}},\ and\ \bibinfo {author} {\bibfnamefont {P.~J.}\ \bibnamefont {Coles}},\ }\bibfield  {title} {\bibinfo {title} {Practical {H}amiltonian learning with unitary dynamics and {G}ibbs states},\ }\href {https://doi.org/10.1038/s41467-023-44008-1} {\bibfield  {journal} {\bibinfo  {journal} {Nat. Commun.}\ }\textbf {\bibinfo {volume} {15}},\ \bibinfo {pages} {312} (\bibinfo {year} {2024})}\BibitemShut {NoStop}%
	\bibitem [{\citenamefont {Zhu}(2017)}]{PhysRevA.96.062336}%
	\BibitemOpen
	\bibfield  {author} {\bibinfo {author} {\bibfnamefont {H.}~\bibnamefont {Zhu}},\ }\bibfield  {title} {\bibinfo {title} {Multiqubit {C}lifford groups are unitary 3-designs},\ }\href {https://doi.org/10.1103/PhysRevA.96.062336} {\bibfield  {journal} {\bibinfo  {journal} {Phys. Rev. A}\ }\textbf {\bibinfo {volume} {96}},\ \bibinfo {pages} {062336} (\bibinfo {year} {2017})}\BibitemShut {NoStop}%
	\bibitem [{\citenamefont {Webb}(2016)}]{Webb16}%
	\BibitemOpen
	\bibfield  {author} {\bibinfo {author} {\bibfnamefont {Z.}~\bibnamefont {Webb}},\ }\bibfield  {title} {\bibinfo {title} {The {C}lifford group forms a unitary 3-design},\ }\href {https://dl.acm.org/doi/abs/10.5555/3179439.3179447} {\bibfield  {journal} {\bibinfo  {journal} {Quant. Inf. Comp.}\ }\textbf {\bibinfo {volume} {16}},\ \bibinfo {pages} {1379} (\bibinfo {year} {2016})}\BibitemShut {NoStop}%
	\bibitem [{\citenamefont {Chen}\ \emph {et~al.}(2022)\citenamefont {Chen}, \citenamefont {Cotler}, \citenamefont {Huang},\ and\ \citenamefont {Li}}]{chen2022exponential}%
	\BibitemOpen
	\bibfield  {author} {\bibinfo {author} {\bibfnamefont {S.}~\bibnamefont {Chen}}, \bibinfo {author} {\bibfnamefont {J.}~\bibnamefont {Cotler}}, \bibinfo {author} {\bibfnamefont {H.-Y.}\ \bibnamefont {Huang}},\ and\ \bibinfo {author} {\bibfnamefont {J.}~\bibnamefont {Li}},\ }\bibfield  {title} {\bibinfo {title} {Exponential separations between learning with and without quantum memory},\ }in\ \href {https://doi.org/10.1109/FOCS52979.2021.00063} {\emph {\bibinfo {booktitle} {2021 IEEE 62nd Annual Symposium on Foundations of Computer Science (FOCS)}}}\ (\bibinfo {organization} {IEEE},\ \bibinfo {year} {2022})\ pp.\ \bibinfo {pages} {574--585}\BibitemShut {NoStop}%
	\bibitem [{\citenamefont {Bluhm}\ \emph {et~al.}(2024)\citenamefont {Bluhm}, \citenamefont {Caro},\ and\ \citenamefont {Oufkir}}]{bluhm2024hamiltonian}%
	\BibitemOpen
	\bibfield  {author} {\bibinfo {author} {\bibfnamefont {A.}~\bibnamefont {Bluhm}}, \bibinfo {author} {\bibfnamefont {M.~C.}\ \bibnamefont {Caro}},\ and\ \bibinfo {author} {\bibfnamefont {A.}~\bibnamefont {Oufkir}},\ }\bibfield  {title} {\bibinfo {title} {Hamiltonian property testing},\ }\href {https://arxiv.org/abs/2403.02968} {\bibfield  {journal} {\bibinfo  {journal} {arXiv:2403.02968}\ } (\bibinfo {year} {2024})}\BibitemShut {NoStop}%
	\bibitem [{\citenamefont {Dutkiewicz}\ \emph {et~al.}(2024)\citenamefont {Dutkiewicz}, \citenamefont {O'Brien},\ and\ \citenamefont {Schuster}}]{Dutkiewicz2024advantageofquantum}%
	\BibitemOpen
	\bibfield  {author} {\bibinfo {author} {\bibfnamefont {A.}~\bibnamefont {Dutkiewicz}}, \bibinfo {author} {\bibfnamefont {T.~E.}\ \bibnamefont {O'Brien}},\ and\ \bibinfo {author} {\bibfnamefont {T.}~\bibnamefont {Schuster}},\ }\bibfield  {title} {\bibinfo {title} {The advantage of quantum control in many-body {H}amiltonian learning},\ }\href {https://doi.org/10.22331/q-2024-11-26-1537} {\bibfield  {journal} {\bibinfo  {journal} {{Quantum}}\ }\textbf {\bibinfo {volume} {8}},\ \bibinfo {pages} {1537} (\bibinfo {year} {2024})}\BibitemShut {NoStop}%
	\bibitem [{\citenamefont {Massar}\ and\ \citenamefont {Popescu}(1995)}]{PhysRevLett.74.1259}%
	\BibitemOpen
	\bibfield  {author} {\bibinfo {author} {\bibfnamefont {S.}~\bibnamefont {Massar}}\ and\ \bibinfo {author} {\bibfnamefont {S.}~\bibnamefont {Popescu}},\ }\bibfield  {title} {\bibinfo {title} {Optimal extraction of information from finite quantum ensembles},\ }\href {https://doi.org/10.1103/PhysRevLett.74.1259} {\bibfield  {journal} {\bibinfo  {journal} {Phys. Rev. Lett.}\ }\textbf {\bibinfo {volume} {74}},\ \bibinfo {pages} {1259} (\bibinfo {year} {1995})}\BibitemShut {NoStop}%
	\bibitem [{\citenamefont {Hayashi}(1998)}]{hayashi1998asymptotic}%
	\BibitemOpen
	\bibfield  {author} {\bibinfo {author} {\bibfnamefont {M.}~\bibnamefont {Hayashi}},\ }\bibfield  {title} {\bibinfo {title} {Asymptotic estimation theory for a finite-dimensional pure state model},\ }\href {https://doi.org/10.1088/0305-4470/31/20/006} {\bibfield  {journal} {\bibinfo  {journal} {J. Phys. A: Math. Gen.}\ }\textbf {\bibinfo {volume} {31}},\ \bibinfo {pages} {4633} (\bibinfo {year} {1998})}\BibitemShut {NoStop}%
	\bibitem [{\citenamefont {Bru{\ss}}\ and\ \citenamefont {Macchiavello}(1999)}]{bruss1999optimal}%
	\BibitemOpen
	\bibfield  {author} {\bibinfo {author} {\bibfnamefont {D.}~\bibnamefont {Bru{\ss}}}\ and\ \bibinfo {author} {\bibfnamefont {C.}~\bibnamefont {Macchiavello}},\ }\bibfield  {title} {\bibinfo {title} {Optimal state estimation for $d$-dimensional quantum systems},\ }\href {https://doi.org/10.1016/S0375-9601(99)00099-7} {\bibfield  {journal} {\bibinfo  {journal} {Phys. Lett. A}\ }\textbf {\bibinfo {volume} {253}},\ \bibinfo {pages} {249} (\bibinfo {year} {1999})}\BibitemShut {NoStop}%
	\bibitem [{\citenamefont {Zhu}\ and\ \citenamefont {Hayashi}(2018)}]{zhu2018universally}%
	\BibitemOpen
	\bibfield  {author} {\bibinfo {author} {\bibfnamefont {H.}~\bibnamefont {Zhu}}\ and\ \bibinfo {author} {\bibfnamefont {M.}~\bibnamefont {Hayashi}},\ }\bibfield  {title} {\bibinfo {title} {Universally {F}isher-symmetric informationally complete measurements},\ }\href {https://doi.org/10.1103/PhysRevLett.120.030404} {\bibfield  {journal} {\bibinfo  {journal} {Phys. Rev. Lett.}\ }\textbf {\bibinfo {volume} {120}},\ \bibinfo {pages} {030404} (\bibinfo {year} {2018})}\BibitemShut {NoStop}%
	\bibitem [{\citenamefont {Harrow}(2013)}]{harrow2013church}%
	\BibitemOpen
	\bibfield  {author} {\bibinfo {author} {\bibfnamefont {A.~W.}\ \bibnamefont {Harrow}},\ }\bibfield  {title} {\bibinfo {title} {The church of the symmetric subspace},\ }\href {https://arxiv.org/abs/1308.6595} {\bibfield  {journal} {\bibinfo  {journal} {arXiv:1308.6595}\ } (\bibinfo {year} {2013})}\BibitemShut {NoStop}%
	\bibitem [{\citenamefont {Anshu}\ \emph {et~al.}(2022)\citenamefont {Anshu}, \citenamefont {Landau},\ and\ \citenamefont {Liu}}]{anshu2022distributed}%
	\BibitemOpen
	\bibfield  {author} {\bibinfo {author} {\bibfnamefont {A.}~\bibnamefont {Anshu}}, \bibinfo {author} {\bibfnamefont {Z.}~\bibnamefont {Landau}},\ and\ \bibinfo {author} {\bibfnamefont {Y.}~\bibnamefont {Liu}},\ }\bibfield  {title} {\bibinfo {title} {Distributed quantum inner product estimation},\ }in\ \href {https://doi.org/10.1145/3519935.3519974} {\emph {\bibinfo {booktitle} {Proceedings of the 54th Annual ACM SIGACT Symposium on Theory of Computing}}}\ (\bibinfo {year} {2022})\ pp.\ \bibinfo {pages} {44--51}\BibitemShut {NoStop}%
	\bibitem [{\citenamefont {Hou}\ \emph {et~al.}(2018)\citenamefont {Hou}, \citenamefont {Tang}, \citenamefont {Shang}, \citenamefont {Zhu}, \citenamefont {Li}, \citenamefont {Yuan}, \citenamefont {Wu}, \citenamefont {Xiang}, \citenamefont {Li},\ and\ \citenamefont {Guo}}]{hou2018deterministic}%
	\BibitemOpen
	\bibfield  {author} {\bibinfo {author} {\bibfnamefont {Z.}~\bibnamefont {Hou}}, \bibinfo {author} {\bibfnamefont {J.-F.}\ \bibnamefont {Tang}}, \bibinfo {author} {\bibfnamefont {J.}~\bibnamefont {Shang}}, \bibinfo {author} {\bibfnamefont {H.}~\bibnamefont {Zhu}}, \bibinfo {author} {\bibfnamefont {J.}~\bibnamefont {Li}}, \bibinfo {author} {\bibfnamefont {Y.}~\bibnamefont {Yuan}}, \bibinfo {author} {\bibfnamefont {K.-D.}\ \bibnamefont {Wu}}, \bibinfo {author} {\bibfnamefont {G.-Y.}\ \bibnamefont {Xiang}}, \bibinfo {author} {\bibfnamefont {C.-F.}\ \bibnamefont {Li}},\ and\ \bibinfo {author} {\bibfnamefont {G.-C.}\ \bibnamefont {Guo}},\ }\bibfield  {title} {\bibinfo {title} {Deterministic realization of collective measurements via photonic quantum walks},\ }\href {https://doi.org/10.1038/s41467-018-03849-x} {\bibfield  {journal} {\bibinfo  {journal} {Nat. Commun.}\ }\textbf {\bibinfo {volume} {9}},\ \bibinfo {pages} {1414} (\bibinfo {year} {2018})}\BibitemShut {NoStop}%
	\bibitem [{\citenamefont {Conlon}\ \emph {et~al.}(2023)\citenamefont {Conlon}, \citenamefont {Vogl}, \citenamefont {Marciniak}, \citenamefont {Pogorelov}, \citenamefont {Yung}, \citenamefont {Eilenberger}, \citenamefont {Berry}, \citenamefont {Santana}, \citenamefont {Blatt}, \citenamefont {Monz} \emph {et~al.}}]{conlon2023approaching}%
	\BibitemOpen
	\bibfield  {author} {\bibinfo {author} {\bibfnamefont {L.~O.}\ \bibnamefont {Conlon}}, \bibinfo {author} {\bibfnamefont {T.}~\bibnamefont {Vogl}}, \bibinfo {author} {\bibfnamefont {C.~D.}\ \bibnamefont {Marciniak}}, \bibinfo {author} {\bibfnamefont {I.}~\bibnamefont {Pogorelov}}, \bibinfo {author} {\bibfnamefont {S.~K.}\ \bibnamefont {Yung}}, \bibinfo {author} {\bibfnamefont {F.}~\bibnamefont {Eilenberger}}, \bibinfo {author} {\bibfnamefont {D.~W.}\ \bibnamefont {Berry}}, \bibinfo {author} {\bibfnamefont {F.~S.}\ \bibnamefont {Santana}}, \bibinfo {author} {\bibfnamefont {R.}~\bibnamefont {Blatt}}, \bibinfo {author} {\bibfnamefont {T.}~\bibnamefont {Monz}}, \emph {et~al.},\ }\bibfield  {title} {\bibinfo {title} {Approaching optimal entangling collective measurements on quantum computing platforms},\ }\href {https://doi.org/10.1038/s41567-022-01875-7} {\bibfield  {journal} {\bibinfo  {journal} {Nat. Phys.}\ }\textbf {\bibinfo {volume} {19}},\ \bibinfo {pages} {351} (\bibinfo {year} {2023})}\BibitemShut
	{NoStop}%
	\bibitem [{\citenamefont {Zhou}\ \emph {et~al.}(2025)\citenamefont {Zhou}, \citenamefont {Yi}, \citenamefont {Yan}, \citenamefont {Hou}, \citenamefont {Zhu}, \citenamefont {Xiang}, \citenamefont {Li},\ and\ \citenamefont {Guo}}]{zhou2023experimental}%
	\BibitemOpen
	\bibfield  {author} {\bibinfo {author} {\bibfnamefont {K.}~\bibnamefont {Zhou}}, \bibinfo {author} {\bibfnamefont {C.}~\bibnamefont {Yi}}, \bibinfo {author} {\bibfnamefont {W.-Z.}\ \bibnamefont {Yan}}, \bibinfo {author} {\bibfnamefont {Z.}~\bibnamefont {Hou}}, \bibinfo {author} {\bibfnamefont {H.}~\bibnamefont {Zhu}}, \bibinfo {author} {\bibfnamefont {G.-Y.}\ \bibnamefont {Xiang}}, \bibinfo {author} {\bibfnamefont {C.-F.}\ \bibnamefont {Li}},\ and\ \bibinfo {author} {\bibfnamefont {G.-C.}\ \bibnamefont {Guo}},\ }\bibfield  {title} {\bibinfo {title} {Experimental realization of genuine three-copy collective measurements for optimal information extraction},\ }\href {https://doi.org/10.1103/PhysRevLett.134.210201} {\bibfield  {journal} {\bibinfo  {journal} {Phys. Rev. Lett.}\ }\textbf {\bibinfo {volume} {134}},\ \bibinfo {pages} {210201} (\bibinfo {year} {2025})}\BibitemShut {NoStop}%
	\bibitem [{\citenamefont {Wu}\ \emph {et~al.}(2020)\citenamefont {Wu}, \citenamefont {B\"aumer}, \citenamefont {Tang}, \citenamefont {Hovhannisyan}, \citenamefont {Perarnau-Llobet}, \citenamefont {Xiang}, \citenamefont {Li},\ and\ \citenamefont {Guo}}]{PhysRevLett.125.210401}%
	\BibitemOpen
	\bibfield  {author} {\bibinfo {author} {\bibfnamefont {K.-D.}\ \bibnamefont {Wu}}, \bibinfo {author} {\bibfnamefont {E.}~\bibnamefont {B\"aumer}}, \bibinfo {author} {\bibfnamefont {J.-F.}\ \bibnamefont {Tang}}, \bibinfo {author} {\bibfnamefont {K.~V.}\ \bibnamefont {Hovhannisyan}}, \bibinfo {author} {\bibfnamefont {M.}~\bibnamefont {Perarnau-Llobet}}, \bibinfo {author} {\bibfnamefont {G.-Y.}\ \bibnamefont {Xiang}}, \bibinfo {author} {\bibfnamefont {C.-F.}\ \bibnamefont {Li}},\ and\ \bibinfo {author} {\bibfnamefont {G.-C.}\ \bibnamefont {Guo}},\ }\bibfield  {title} {\bibinfo {title} {Minimizing backaction through entangled measurements},\ }\href {https://doi.org/10.1103/PhysRevLett.125.210401} {\bibfield  {journal} {\bibinfo  {journal} {Phys. Rev. Lett.}\ }\textbf {\bibinfo {volume} {125}},\ \bibinfo {pages} {210401} (\bibinfo {year} {2020})}\BibitemShut {NoStop}%
	\bibitem [{\citenamefont {Choi}(1975)}]{choi1975completely}%
	\BibitemOpen
	\bibfield  {author} {\bibinfo {author} {\bibfnamefont {M.-D.}\ \bibnamefont {Choi}},\ }\bibfield  {title} {\bibinfo {title} {Completely positive linear maps on complex matrices},\ }\href {https://doi.org/10.1016/0024-3795(75)90075-0} {\bibfield  {journal} {\bibinfo  {journal} {Linear Algebra Appl.}\ }\textbf {\bibinfo {volume} {10}},\ \bibinfo {pages} {285} (\bibinfo {year} {1975})}\BibitemShut {NoStop}%
	\bibitem [{\citenamefont {Kothari}(2014)}]{Kothari14}%
	\BibitemOpen
	\bibfield  {author} {\bibinfo {author} {\bibfnamefont {R.}~\bibnamefont {Kothari}},\ }\bibfield  {title} {\bibinfo {title} {An optimal quantum algorithm for the oracle identification problem},\ }in\ \href {https://doi.org/10.4230/LIPIcs.STACS.2014.482} {\emph {\bibinfo {booktitle} {31st International Symposium on Theoretical Aspects of Computer Science}}}\ (\bibinfo {year} {2014})\ pp.\ \bibinfo {pages} {482--493}\BibitemShut {NoStop}%
	\bibitem [{\citenamefont {Arunachalam}\ and\ \citenamefont {de~Wolf}(2017)}]{arunachalam2017survey}%
	\BibitemOpen
	\bibfield  {author} {\bibinfo {author} {\bibfnamefont {S.}~\bibnamefont {Arunachalam}}\ and\ \bibinfo {author} {\bibfnamefont {R.}~\bibnamefont {de~Wolf}},\ }\bibfield  {title} {\bibinfo {title} {Guest column: A survey of quantum learning theory},\ }\href {https://doi.org/10.1145/3106700.3106710} {\bibfield  {journal} {\bibinfo  {journal} {ACM SIGACT News}\ }\textbf {\bibinfo {volume} {48}},\ \bibinfo {pages} {41} (\bibinfo {year} {2017})}\BibitemShut {NoStop}%
	\bibitem [{\citenamefont {Collins}\ and\ \citenamefont {Nechita}(2016)}]{collins2016random}%
	\BibitemOpen
	\bibfield  {author} {\bibinfo {author} {\bibfnamefont {B.}~\bibnamefont {Collins}}\ and\ \bibinfo {author} {\bibfnamefont {I.}~\bibnamefont {Nechita}},\ }\bibfield  {title} {\bibinfo {title} {Random matrix techniques in quantum information theory},\ }\href {https://doi.org/10.1063/1.4936880} {\bibfield  {journal} {\bibinfo  {journal} {J. Math. Phys}\ }\textbf {\bibinfo {volume} {57}},\ \bibinfo {pages} {015215} (\bibinfo {year} {2016})}\BibitemShut {NoStop}%
	\bibitem [{\citenamefont {Nechita}(2007)}]{nechita2007asymptotics}%
	\BibitemOpen
	\bibfield  {author} {\bibinfo {author} {\bibfnamefont {I.}~\bibnamefont {Nechita}},\ }\bibfield  {title} {\bibinfo {title} {Asymptotics of random density matrices},\ }\href {https://link.springer.com/article/10.1007/s00023-007-0345-5} {\bibfield  {journal} {\bibinfo  {journal} {Annales Henri Poincar{\'e}}\ }\textbf {\bibinfo {volume} {8}},\ \bibinfo {pages} {1521} (\bibinfo {year} {2007})}\BibitemShut {NoStop}%
	\bibitem [{\citenamefont {{\.Z}yczkowski}\ \emph {et~al.}(2011)\citenamefont {{\.Z}yczkowski}, \citenamefont {Penson}, \citenamefont {Nechita},\ and\ \citenamefont {Collins}}]{zyczkowski2011}%
	\BibitemOpen
	\bibfield  {author} {\bibinfo {author} {\bibfnamefont {K.}~\bibnamefont {{\.Z}yczkowski}}, \bibinfo {author} {\bibfnamefont {K.~A.}\ \bibnamefont {Penson}}, \bibinfo {author} {\bibfnamefont {I.}~\bibnamefont {Nechita}},\ and\ \bibinfo {author} {\bibfnamefont {B.}~\bibnamefont {Collins}},\ }\bibfield  {title} {\bibinfo {title} {Generating random density matrices},\ }\href {https://pubs.aip.org/aip/jmp/article-abstract/52/6/062201/232445/Generating-random-density-matrices?redirectedFrom=fulltext} {\bibfield  {journal} {\bibinfo  {journal} {J. Math. Phys.}\ }\textbf {\bibinfo {volume} {52}},\ \bibinfo {pages} {062201} (\bibinfo {year} {2011})}\BibitemShut {NoStop}%
	\bibitem [{\citenamefont {Lubkin}(1978)}]{lubkin1978entropy}%
	\BibitemOpen
	\bibfield  {author} {\bibinfo {author} {\bibfnamefont {E.}~\bibnamefont {Lubkin}},\ }\bibfield  {title} {\bibinfo {title} {Entropy of an $n$-system from its correlation with a $k$-reservoir},\ }\href {https://pubs.aip.org/aip/jmp/article-abstract/19/5/1028/460079/Entropy-of-an-n-system-from-its-correlation-with-a?redirectedFrom=PDF} {\bibfield  {journal} {\bibinfo  {journal} {J. Math. Phys}\ }\textbf {\bibinfo {volume} {19}},\ \bibinfo {pages} {1028} (\bibinfo {year} {1978})}\BibitemShut {NoStop}%
	\bibitem [{\citenamefont {O'Brien}\ \emph {et~al.}(2004)\citenamefont {O'Brien}, \citenamefont {Pryde}, \citenamefont {Gilchrist}, \citenamefont {James}, \citenamefont {Langford}, \citenamefont {Ralph},\ and\ \citenamefont {White}}]{PhysRevLett.93.080502}%
	\BibitemOpen
	\bibfield  {author} {\bibinfo {author} {\bibfnamefont {J.~L.}\ \bibnamefont {O'Brien}}, \bibinfo {author} {\bibfnamefont {G.~J.}\ \bibnamefont {Pryde}}, \bibinfo {author} {\bibfnamefont {A.}~\bibnamefont {Gilchrist}}, \bibinfo {author} {\bibfnamefont {D.~F.~V.}\ \bibnamefont {James}}, \bibinfo {author} {\bibfnamefont {N.~K.}\ \bibnamefont {Langford}}, \bibinfo {author} {\bibfnamefont {T.~C.}\ \bibnamefont {Ralph}},\ and\ \bibinfo {author} {\bibfnamefont {A.~G.}\ \bibnamefont {White}},\ }\bibfield  {title} {\bibinfo {title} {Quantum process tomography of a controlled-{NOT} gate},\ }\href {https://doi.org/10.1103/PhysRevLett.93.080502} {\bibfield  {journal} {\bibinfo  {journal} {Phys. Rev. Lett.}\ }\textbf {\bibinfo {volume} {93}},\ \bibinfo {pages} {080502} (\bibinfo {year} {2004})}\BibitemShut {NoStop}%
	\bibitem [{\citenamefont {Yang}\ \emph {et~al.}(2020)\citenamefont {Yang}, \citenamefont {Renner},\ and\ \citenamefont {Chiribella}}]{PhysRevLett.125.210501}%
	\BibitemOpen
	\bibfield  {author} {\bibinfo {author} {\bibfnamefont {Y.}~\bibnamefont {Yang}}, \bibinfo {author} {\bibfnamefont {R.}~\bibnamefont {Renner}},\ and\ \bibinfo {author} {\bibfnamefont {G.}~\bibnamefont {Chiribella}},\ }\bibfield  {title} {\bibinfo {title} {Optimal universal programming of unitary gates},\ }\href {https://doi.org/10.1103/PhysRevLett.125.210501} {\bibfield  {journal} {\bibinfo  {journal} {Phys. Rev. Lett.}\ }\textbf {\bibinfo {volume} {125}},\ \bibinfo {pages} {210501} (\bibinfo {year} {2020})}\BibitemShut {NoStop}%
	\bibitem [{\citenamefont {Giovannetti}\ \emph {et~al.}(2011)\citenamefont {Giovannetti}, \citenamefont {Lloyd},\ and\ \citenamefont {Maccone}}]{giovannetti2011advances}%
	\BibitemOpen
	\bibfield  {author} {\bibinfo {author} {\bibfnamefont {V.}~\bibnamefont {Giovannetti}}, \bibinfo {author} {\bibfnamefont {S.}~\bibnamefont {Lloyd}},\ and\ \bibinfo {author} {\bibfnamefont {L.}~\bibnamefont {Maccone}},\ }\bibfield  {title} {\bibinfo {title} {Advances in quantum metrology},\ }\href {https://doi.org/10.1038/nphoton.2011.35} {\bibfield  {journal} {\bibinfo  {journal} {Nat. Photon.}\ }\textbf {\bibinfo {volume} {5}},\ \bibinfo {pages} {222} (\bibinfo {year} {2011})}\BibitemShut {NoStop}%
	\bibitem [{\citenamefont {Degen}\ \emph {et~al.}(2017)\citenamefont {Degen}, \citenamefont {Reinhard},\ and\ \citenamefont {Cappellaro}}]{RevModPhys.89.035002}%
	\BibitemOpen
	\bibfield  {author} {\bibinfo {author} {\bibfnamefont {C.~L.}\ \bibnamefont {Degen}}, \bibinfo {author} {\bibfnamefont {F.}~\bibnamefont {Reinhard}},\ and\ \bibinfo {author} {\bibfnamefont {P.}~\bibnamefont {Cappellaro}},\ }\bibfield  {title} {\bibinfo {title} {Quantum sensing},\ }\href {https://doi.org/10.1103/RevModPhys.89.035002} {\bibfield  {journal} {\bibinfo  {journal} {Rev. Mod. Phys.}\ }\textbf {\bibinfo {volume} {89}},\ \bibinfo {pages} {035002} (\bibinfo {year} {2017})}\BibitemShut {NoStop}%
	\bibitem [{\citenamefont {Shulman}\ \emph {et~al.}(2014)\citenamefont {Shulman}, \citenamefont {Harvey}, \citenamefont {Nichol}, \citenamefont {Bartlett}, \citenamefont {Doherty}, \citenamefont {Umansky},\ and\ \citenamefont {Yacoby}}]{shulman2014suppressing}%
	\BibitemOpen
	\bibfield  {author} {\bibinfo {author} {\bibfnamefont {M.~D.}\ \bibnamefont {Shulman}}, \bibinfo {author} {\bibfnamefont {S.~P.}\ \bibnamefont {Harvey}}, \bibinfo {author} {\bibfnamefont {J.~M.}\ \bibnamefont {Nichol}}, \bibinfo {author} {\bibfnamefont {S.~D.}\ \bibnamefont {Bartlett}}, \bibinfo {author} {\bibfnamefont {A.~C.}\ \bibnamefont {Doherty}}, \bibinfo {author} {\bibfnamefont {V.}~\bibnamefont {Umansky}},\ and\ \bibinfo {author} {\bibfnamefont {A.}~\bibnamefont {Yacoby}},\ }\bibfield  {title} {\bibinfo {title} {Suppressing qubit dephasing using real-time {H}amiltonian estimation},\ }\href {https://doi.org/10.1038/ncomms6156} {\bibfield  {journal} {\bibinfo  {journal} {Nat. Commun.}\ }\textbf {\bibinfo {volume} {5}},\ \bibinfo {pages} {5156} (\bibinfo {year} {2014})}\BibitemShut {NoStop}%
	\bibitem [{\citenamefont {Innocenti}\ \emph {et~al.}(2020)\citenamefont {Innocenti}, \citenamefont {Banchi}, \citenamefont {Ferraro}, \citenamefont {Bose},\ and\ \citenamefont {Paternostro}}]{innocenti2020supervised}%
	\BibitemOpen
	\bibfield  {author} {\bibinfo {author} {\bibfnamefont {L.}~\bibnamefont {Innocenti}}, \bibinfo {author} {\bibfnamefont {L.}~\bibnamefont {Banchi}}, \bibinfo {author} {\bibfnamefont {A.}~\bibnamefont {Ferraro}}, \bibinfo {author} {\bibfnamefont {S.}~\bibnamefont {Bose}},\ and\ \bibinfo {author} {\bibfnamefont {M.}~\bibnamefont {Paternostro}},\ }\bibfield  {title} {\bibinfo {title} {Supervised learning of time-independent {H}amiltonians for gate design},\ }\href {https://doi.org/10.1088/1367-2630/ab8aaf} {\bibfield  {journal} {\bibinfo  {journal} {New J. Phys.}\ }\textbf {\bibinfo {volume} {22}},\ \bibinfo {pages} {065001} (\bibinfo {year} {2020})}\BibitemShut {NoStop}%
	\bibitem [{\citenamefont {Ma}\ \emph {et~al.}(2024)\citenamefont {Ma}, \citenamefont {Flammia}, \citenamefont {Preskill},\ and\ \citenamefont {Tong}}]{ma2024learning}%
	\BibitemOpen
	\bibfield  {author} {\bibinfo {author} {\bibfnamefont {M.}~\bibnamefont {Ma}}, \bibinfo {author} {\bibfnamefont {S.~T.}\ \bibnamefont {Flammia}}, \bibinfo {author} {\bibfnamefont {J.}~\bibnamefont {Preskill}},\ and\ \bibinfo {author} {\bibfnamefont {Y.}~\bibnamefont {Tong}},\ }\bibfield  {title} {\bibinfo {title} {Learning $k$-body {H}amiltonians via compressed sensing},\ }\href {https://arxiv.org/abs/2410.18928} {\bibfield  {journal} {\bibinfo  {journal} {arXiv:2410.18928}\ } (\bibinfo {year} {2024})}\BibitemShut {NoStop}%
	\bibitem [{\citenamefont {Zhao}\ \emph {et~al.}(2022)\citenamefont {Zhao}, \citenamefont {Zhou}, \citenamefont {Shaw}, \citenamefont {Li},\ and\ \citenamefont {Childs}}]{Zhao2022Random}%
	\BibitemOpen
	\bibfield  {author} {\bibinfo {author} {\bibfnamefont {Q.}~\bibnamefont {Zhao}}, \bibinfo {author} {\bibfnamefont {Y.}~\bibnamefont {Zhou}}, \bibinfo {author} {\bibfnamefont {A.~F.}\ \bibnamefont {Shaw}}, \bibinfo {author} {\bibfnamefont {T.}~\bibnamefont {Li}},\ and\ \bibinfo {author} {\bibfnamefont {A.~M.}\ \bibnamefont {Childs}},\ }\bibfield  {title} {\bibinfo {title} {Hamiltonian simulation with random inputs},\ }\href {https://doi.org/10.1103/PhysRevLett.129.270502} {\bibfield  {journal} {\bibinfo  {journal} {Phys. Rev. Lett.}\ }\textbf {\bibinfo {volume} {129}},\ \bibinfo {pages} {270502} (\bibinfo {year} {2022})}\BibitemShut {NoStop}%
	\bibitem [{\citenamefont {Bairey}\ \emph {et~al.}(2019)\citenamefont {Bairey}, \citenamefont {Arad},\ and\ \citenamefont {Lindner}}]{PhysRevLett.122.020504}%
	\BibitemOpen
	\bibfield  {author} {\bibinfo {author} {\bibfnamefont {E.}~\bibnamefont {Bairey}}, \bibinfo {author} {\bibfnamefont {I.}~\bibnamefont {Arad}},\ and\ \bibinfo {author} {\bibfnamefont {N.~H.}\ \bibnamefont {Lindner}},\ }\bibfield  {title} {\bibinfo {title} {Learning a local {H}amiltonian from local measurements},\ }\href {https://doi.org/10.1103/PhysRevLett.122.020504} {\bibfield  {journal} {\bibinfo  {journal} {Phys. Rev. Lett.}\ }\textbf {\bibinfo {volume} {122}},\ \bibinfo {pages} {020504} (\bibinfo {year} {2019})}\BibitemShut {NoStop}%
	\bibitem [{\citenamefont {Yu}\ \emph {et~al.}(2023)\citenamefont {Yu}, \citenamefont {Sun}, \citenamefont {Han},\ and\ \citenamefont {Yuan}}]{Yu2023robustefficient}%
	\BibitemOpen
	\bibfield  {author} {\bibinfo {author} {\bibfnamefont {W.}~\bibnamefont {Yu}}, \bibinfo {author} {\bibfnamefont {J.}~\bibnamefont {Sun}}, \bibinfo {author} {\bibfnamefont {Z.}~\bibnamefont {Han}},\ and\ \bibinfo {author} {\bibfnamefont {X.}~\bibnamefont {Yuan}},\ }\bibfield  {title} {\bibinfo {title} {Robust and {E}fficient {H}amiltonian {L}earning},\ }\href {https://doi.org/10.22331/q-2023-06-29-1045} {\bibfield  {journal} {\bibinfo  {journal} {{Quantum}}\ }\textbf {\bibinfo {volume} {7}},\ \bibinfo {pages} {1045} (\bibinfo {year} {2023})}\BibitemShut {NoStop}%
	\bibitem [{\citenamefont {Hu}\ \emph {et~al.}(2025)\citenamefont {Hu}, \citenamefont {Ma}, \citenamefont {Gong}, \citenamefont {Ye}, \citenamefont {Tong}, \citenamefont {Flammia},\ and\ \citenamefont {Yelin}}]{hu2025ansatz}%
	\BibitemOpen
	\bibfield  {author} {\bibinfo {author} {\bibfnamefont {H.-Y.}\ \bibnamefont {Hu}}, \bibinfo {author} {\bibfnamefont {M.}~\bibnamefont {Ma}}, \bibinfo {author} {\bibfnamefont {W.}~\bibnamefont {Gong}}, \bibinfo {author} {\bibfnamefont {Q.}~\bibnamefont {Ye}}, \bibinfo {author} {\bibfnamefont {Y.}~\bibnamefont {Tong}}, \bibinfo {author} {\bibfnamefont {S.~T.}\ \bibnamefont {Flammia}},\ and\ \bibinfo {author} {\bibfnamefont {S.~F.}\ \bibnamefont {Yelin}},\ }\bibfield  {title} {\bibinfo {title} {Ansatz-free {H}amiltonian learning with {H}eisenberg-limited scaling},\ }\href {https://doi.org/10.1103/j7b8-pb77} {\bibfield  {journal} {\bibinfo  {journal} {PRX Quantum}\ }\textbf {\bibinfo {volume} {6}},\ \bibinfo {pages} {040315} (\bibinfo {year} {2025})}\BibitemShut {NoStop}%
	\bibitem [{\citenamefont {Zhao}(2025)}]{zhao2024learning}%
	\BibitemOpen
	\bibfield  {author} {\bibinfo {author} {\bibfnamefont {A.}~\bibnamefont {Zhao}},\ }\bibfield  {title} {\bibinfo {title} {Learning the structure of any {H}amiltonian from minimal assumptions},\ }in\ \href {https://doi.org/10.1145/3717823.3718115} {\emph {\bibinfo {booktitle} {Proceedings of the 57th Annual ACM SIGACT Symposium on Theory of Computing}}}\ (\bibinfo {year} {2025})\ pp.\ \bibinfo {pages} {1201--1211}\BibitemShut {NoStop}%
	\bibitem [{\citenamefont {Huang}\ \emph {et~al.}(2023{\natexlab{b}})\citenamefont {Huang}, \citenamefont {Tong}, \citenamefont {Fang},\ and\ \citenamefont {Su}}]{PhysRevLett.130.200403}%
	\BibitemOpen
	\bibfield  {author} {\bibinfo {author} {\bibfnamefont {H.-Y.}\ \bibnamefont {Huang}}, \bibinfo {author} {\bibfnamefont {Y.}~\bibnamefont {Tong}}, \bibinfo {author} {\bibfnamefont {D.}~\bibnamefont {Fang}},\ and\ \bibinfo {author} {\bibfnamefont {Y.}~\bibnamefont {Su}},\ }\bibfield  {title} {\bibinfo {title} {Learning many-body {H}amiltonians with {H}eisenberg-limited scaling},\ }\href {https://doi.org/10.1103/PhysRevLett.130.200403} {\bibfield  {journal} {\bibinfo  {journal} {Phys. Rev. Lett.}\ }\textbf {\bibinfo {volume} {130}},\ \bibinfo {pages} {200403} (\bibinfo {year} {2023}{\natexlab{b}})}\BibitemShut {NoStop}%
	\bibitem [{\citenamefont {Haah}\ \emph {et~al.}(2024)\citenamefont {Haah}, \citenamefont {Kothari},\ and\ \citenamefont {Tang}}]{haah2024learning}%
	\BibitemOpen
	\bibfield  {author} {\bibinfo {author} {\bibfnamefont {J.}~\bibnamefont {Haah}}, \bibinfo {author} {\bibfnamefont {R.}~\bibnamefont {Kothari}},\ and\ \bibinfo {author} {\bibfnamefont {E.}~\bibnamefont {Tang}},\ }\bibfield  {title} {\bibinfo {title} {Learning quantum {H}amiltonians from high-temperature {G}ibbs states and real-time evolutions},\ }\href {https://doi.org/10.1038/s41567-023-02376-x} {\bibfield  {journal} {\bibinfo  {journal} {Nat. Phys.}\ }\textbf {\bibinfo {volume} {20}},\ \bibinfo {pages} {1027} (\bibinfo {year} {2024})}\BibitemShut {NoStop}%
	\bibitem [{\citenamefont {Castaneda}\ and\ \citenamefont {Wiebe}(2025)}]{castaneda2023hamiltonian}%
	\BibitemOpen
	\bibfield  {author} {\bibinfo {author} {\bibfnamefont {J.}~\bibnamefont {Castaneda}}\ and\ \bibinfo {author} {\bibfnamefont {N.}~\bibnamefont {Wiebe}},\ }\bibfield  {title} {\bibinfo {title} {Hamiltonian {L}earning via {S}hadow {T}omography of {P}seudo-{C}hoi {S}tates},\ }\href {https://doi.org/10.22331/q-2025-04-09-1700} {\bibfield  {journal} {\bibinfo  {journal} {{Quantum}}\ }\textbf {\bibinfo {volume} {9}},\ \bibinfo {pages} {1700} (\bibinfo {year} {2025})}\BibitemShut {NoStop}%
	\bibitem [{\citenamefont {Bakshi}\ \emph {et~al.}(2024)\citenamefont {Bakshi}, \citenamefont {Liu}, \citenamefont {Moitra},\ and\ \citenamefont {Tang}}]{bakshi2024structure1}%
	\BibitemOpen
	\bibfield  {author} {\bibinfo {author} {\bibfnamefont {A.}~\bibnamefont {Bakshi}}, \bibinfo {author} {\bibfnamefont {A.}~\bibnamefont {Liu}}, \bibinfo {author} {\bibfnamefont {A.}~\bibnamefont {Moitra}},\ and\ \bibinfo {author} {\bibfnamefont {E.}~\bibnamefont {Tang}},\ }\bibfield  {title} {\bibinfo {title} {Structure learning of {H}amiltonians from real-time evolution},\ }in\ \href {https://ieeexplore.ieee.org/abstract/document/10756152} {\emph {\bibinfo {booktitle} {2024 IEEE 65th Annual Symposium on Foundations of Computer Science (FOCS)}}}\ (\bibinfo {organization} {IEEE},\ \bibinfo {year} {2024})\ pp.\ \bibinfo {pages} {1037--1050}\BibitemShut {NoStop}%
	\bibitem [{\citenamefont {Xu}\ and\ \citenamefont {Swingle}(2024)}]{PRXQuantum.5.010201}%
	\BibitemOpen
	\bibfield  {author} {\bibinfo {author} {\bibfnamefont {S.}~\bibnamefont {Xu}}\ and\ \bibinfo {author} {\bibfnamefont {B.}~\bibnamefont {Swingle}},\ }\bibfield  {title} {\bibinfo {title} {Scrambling dynamics and out-of-time-ordered correlators in quantum many-body systems},\ }\href {https://doi.org/10.1103/PRXQuantum.5.010201} {\bibfield  {journal} {\bibinfo  {journal} {PRX Quantum}\ }\textbf {\bibinfo {volume} {5}},\ \bibinfo {pages} {010201} (\bibinfo {year} {2024})}\BibitemShut {NoStop}%
	\bibitem [{\citenamefont {Chen}\ \emph {et~al.}(2023)\citenamefont {Chen}, \citenamefont {Lucas},\ and\ \citenamefont {Yin}}]{chen2023speed}%
	\BibitemOpen
	\bibfield  {author} {\bibinfo {author} {\bibfnamefont {C.-F.}\ \bibnamefont {Chen}}, \bibinfo {author} {\bibfnamefont {A.}~\bibnamefont {Lucas}},\ and\ \bibinfo {author} {\bibfnamefont {C.}~\bibnamefont {Yin}},\ }\bibfield  {title} {\bibinfo {title} {Speed limits and locality in many-body quantum dynamics},\ }\href {https://doi.org/10.1088/1361-6633/acfaae} {\bibfield  {journal} {\bibinfo  {journal} {Rep. Prog. Phys.}\ }\textbf {\bibinfo {volume} {86}},\ \bibinfo {pages} {116001} (\bibinfo {year} {2023})}\BibitemShut {NoStop}%
	\bibitem [{\citenamefont {Shen}\ \emph {et~al.}(2020)\citenamefont {Shen}, \citenamefont {Zhang}, \citenamefont {You},\ and\ \citenamefont {Zhai}}]{PhysRevLett.124.200504}%
	\BibitemOpen
	\bibfield  {author} {\bibinfo {author} {\bibfnamefont {H.}~\bibnamefont {Shen}}, \bibinfo {author} {\bibfnamefont {P.}~\bibnamefont {Zhang}}, \bibinfo {author} {\bibfnamefont {Y.-Z.}\ \bibnamefont {You}},\ and\ \bibinfo {author} {\bibfnamefont {H.}~\bibnamefont {Zhai}},\ }\bibfield  {title} {\bibinfo {title} {Information scrambling in quantum neural networks},\ }\href {https://doi.org/10.1103/PhysRevLett.124.200504} {\bibfield  {journal} {\bibinfo  {journal} {Phys. Rev. Lett.}\ }\textbf {\bibinfo {volume} {124}},\ \bibinfo {pages} {200504} (\bibinfo {year} {2020})}\BibitemShut {NoStop}%
	\bibitem [{\citenamefont {von Keyserlingk}\ \emph {et~al.}(2018)\citenamefont {von Keyserlingk}, \citenamefont {Rakovszky}, \citenamefont {Pollmann},\ and\ \citenamefont {Sondhi}}]{PhysRevX.8.021013}%
	\BibitemOpen
	\bibfield  {author} {\bibinfo {author} {\bibfnamefont {C.~W.}\ \bibnamefont {von Keyserlingk}}, \bibinfo {author} {\bibfnamefont {T.}~\bibnamefont {Rakovszky}}, \bibinfo {author} {\bibfnamefont {F.}~\bibnamefont {Pollmann}},\ and\ \bibinfo {author} {\bibfnamefont {S.~L.}\ \bibnamefont {Sondhi}},\ }\bibfield  {title} {\bibinfo {title} {Operator hydrodynamics, {OTOC}s, and entanglement growth in systems without conservation laws},\ }\href {https://doi.org/10.1103/PhysRevX.8.021013} {\bibfield  {journal} {\bibinfo  {journal} {Phys. Rev. X}\ }\textbf {\bibinfo {volume} {8}},\ \bibinfo {pages} {021013} (\bibinfo {year} {2018})}\BibitemShut {NoStop}%
	\bibitem [{\citenamefont {Vermersch}\ \emph {et~al.}(2019)\citenamefont {Vermersch}, \citenamefont {Elben}, \citenamefont {Sieberer}, \citenamefont {Yao},\ and\ \citenamefont {Zoller}}]{PhysRevX.9.021061}%
	\BibitemOpen
	\bibfield  {author} {\bibinfo {author} {\bibfnamefont {B.}~\bibnamefont {Vermersch}}, \bibinfo {author} {\bibfnamefont {A.}~\bibnamefont {Elben}}, \bibinfo {author} {\bibfnamefont {L.~M.}\ \bibnamefont {Sieberer}}, \bibinfo {author} {\bibfnamefont {N.~Y.}\ \bibnamefont {Yao}},\ and\ \bibinfo {author} {\bibfnamefont {P.}~\bibnamefont {Zoller}},\ }\bibfield  {title} {\bibinfo {title} {Probing scrambling using statistical correlations between randomized measurements},\ }\href {https://doi.org/10.1103/PhysRevX.9.021061} {\bibfield  {journal} {\bibinfo  {journal} {Phys. Rev. X}\ }\textbf {\bibinfo {volume} {9}},\ \bibinfo {pages} {021061} (\bibinfo {year} {2019})}\BibitemShut {NoStop}%
	\bibitem [{\citenamefont {Garcia}\ \emph {et~al.}(2021)\citenamefont {Garcia}, \citenamefont {Zhou},\ and\ \citenamefont {Jaffe}}]{PhysRevResearch.3.033155}%
	\BibitemOpen
	\bibfield  {author} {\bibinfo {author} {\bibfnamefont {R.~J.}\ \bibnamefont {Garcia}}, \bibinfo {author} {\bibfnamefont {Y.}~\bibnamefont {Zhou}},\ and\ \bibinfo {author} {\bibfnamefont {A.}~\bibnamefont {Jaffe}},\ }\bibfield  {title} {\bibinfo {title} {Quantum scrambling with classical shadows},\ }\href {https://doi.org/10.1103/PhysRevResearch.3.033155} {\bibfield  {journal} {\bibinfo  {journal} {Phys. Rev. Res.}\ }\textbf {\bibinfo {volume} {3}},\ \bibinfo {pages} {033155} (\bibinfo {year} {2021})}\BibitemShut {NoStop}%
	\bibitem [{\citenamefont {Leone}\ \emph {et~al.}(2021)\citenamefont {Leone}, \citenamefont {Oliviero}, \citenamefont {Zhou},\ and\ \citenamefont {Hamma}}]{Leone2021quantumchaosis}%
	\BibitemOpen
	\bibfield  {author} {\bibinfo {author} {\bibfnamefont {L.}~\bibnamefont {Leone}}, \bibinfo {author} {\bibfnamefont {S.~F.~E.}\ \bibnamefont {Oliviero}}, \bibinfo {author} {\bibfnamefont {Y.}~\bibnamefont {Zhou}},\ and\ \bibinfo {author} {\bibfnamefont {A.}~\bibnamefont {Hamma}},\ }\bibfield  {title} {\bibinfo {title} {Quantum {C}haos is {Q}uantum},\ }\href {https://doi.org/10.22331/q-2021-05-04-453} {\bibfield  {journal} {\bibinfo  {journal} {{Quantum}}\ }\textbf {\bibinfo {volume} {5}},\ \bibinfo {pages} {453} (\bibinfo {year} {2021})}\BibitemShut {NoStop}%
	\bibitem [{\citenamefont {Mele}(2024)}]{Mele23}%
	\BibitemOpen
	\bibfield  {author} {\bibinfo {author} {\bibfnamefont {A.~A.}\ \bibnamefont {Mele}},\ }\bibfield  {title} {\bibinfo {title} {Introduction to {H}aar measure tools in quantum information: A beginner's tutorial},\ }\href {https://doi.org/10.22331/q-2024-05-08-1340} {\bibfield  {journal} {\bibinfo  {journal} {Quantum}\ }\textbf {\bibinfo {volume} {8}},\ \bibinfo {pages} {1340} (\bibinfo {year} {2024})}\BibitemShut {NoStop}%
\end{thebibliography}
\end{document}